\documentclass[runningheads]{llncs}

\newif\ifcav{}

\cavfalse{}  

\ifcav{}
\newcommand\aref[1]{\cite{porse-arxiv}}
\else
\newcommand\aref[1]{\cref{#1}}
\fi

\usepackage[utf8]{inputenc} 
\usepackage{amsmath}
\usepackage{amsfonts}
\usepackage{amssymb}
\usepackage{mathtools}
\usepackage[only,llbracket,rrbracket,bigsqcap,bigcurlyvee,bigcurlywedge,rightarrowtriangle,shortrightarrow,shortleftarrow,Mapsto,lightning]{stmaryrd}
\usepackage{bigstrut}
\usepackage{xcolor}
\definecolor{rowShade}{gray}{0.9}
\definecolor{colShade}{gray}{0.9}
\usepackage{xspace}
\usepackage{thm-restate}
\usepackage[mathscr]{eucal}
\usepackage{multicol}
\usepackage{multirow}
\usepackage{textgreek}
\usepackage[binary-units]{siunitx}
\usepackage{url}
\usepackage[center,width=152mm,height=235mm]{crop}
\usepackage{listings} 
\usepackage{tikz}
\usepackage{subcaption}
\usepackage{verbatimbox}
\usepackage{booktabs}
\usepackage{setspace}
\usepackage{hyperref}
\usepackage{enumitem}
\usepackage{marginnote}
\usepackage{semantic}
\usepackage{rotating}
\usepackage[linesnumbered,vlined,ruled]{algorithm2e}
\SetAlFnt{\small}
\SetAlCapFnt{\small}
\SetAlCapNameFnt{\small}

\usepackage[capitalise,english,nameinlink]{cleveref} \usepackage[strings]{underscore}

\usepackage{wrapfig}
\usepackage{flushend}

\hypersetup{
  bookmarksdepth=2,
  bookmarksnumbered=true,
  bookmarksopen=true,
  bookmarksopenlevel=2,
  colorlinks=true,
  linktocpage=true,
  breaklinks=true,
  pageanchor=true,
  allcolors=[rgb]{0.6,0.0,0.0},
  pdftitle={Symbolic Partial-Order Execution for Testing Multi-Threaded Programs},
  pdfauthor={Daniel Schemmel, Julian Büning, César Rodríguez, David Laprell and Klaus Wehrle}
}

\usetikzlibrary{arrows,arrows.meta,calc,intersections,petri,positioning,shapes}

\crefname{definition}{Def.}{Defs.}
\crefname{equation}{}{}
\crefname{proposition}{Prop.}{Props.}
\Crefname{proposition}{Proposition}{Propsitions}
\crefname{section}{Sec.}{Secs.}
\Crefname{section}{Section}{Sections}
\crefname{page}{p.}{pp.}
\Crefname{page}{Page}{Pages}
\crefname{chapter}{Ch.}{Ch.}
\Crefname{chapter}{Chapter}{Chapters}
\crefname{remark}{Rmk.}{Rmks.}
\Crefname{remark}{Remark}{Remarks}
\crefname{appendix}{App.}{App.}
\Crefname{appendix}{Appendix}{Appendix}
\crefname{algorithm}{Alg.}{Alg.}
\Crefname{algorithm}{Algorithm}{Algorithms}
\crefname{line}{line}{line}
\Crefname{line}{Line}{Line}

\SetKwProg{Proc}{Procedure}{}{}
\SetKwProg{Fn}{Function}{}{}
\SetKwIF{If}{ElseIf}{Else}{if}{then}{else if}{else}{end}
\SetKwIF{Record}{{}}{{}}{Structure}{:}{end record}{{}}{{}}
\SetKwFor{ForEach}{foreach}{do}{end}
\SetKwFor{For}{for}{}{end}
\SetKwInOut{Input}{Input}
\SetKwInOut{Output}{Output}
\SetKw{KwContinue}{continue}
\SetKw{KwBreak}{break}
\SetKw{KwAssert}{assert}
\SetKwData{Null}{null}
\SetKwData{False}{false}
\SetKwData{True}{true}
\SetInd{0.2em}{0.6em}
\SetNlSty{}{\color{gray}}{}
\SetVlineSkip{2pt}

\SetKwFunction{alternatives}{alt}
\SetKwFunction{encutoff}{ena}
\SetKwFunction{cexmain}{cex}
\SetKwFunction{cexlocal}{cex-local}
\SetKwFunction{cexacquire}{cex-acq-w2}
\SetKwFunction{cexwait}{cex-w1}
\SetKwFunction{cexnotify}{cex-notify}
\SetKwFunction{outstandingw}{outstanding-w1}
\SetKwFunction{lastof}{last-of}
\SetKwFunction{lastlock}{last-lock}
\SetKwFunction{lastnotify}{last-notify}
\SetKwFunction{concurrent}{conc}

\SetKwFunction{main}{main}
\SetKwFunction{allocatenode}{alloc-nod}
\SetKwFunction{parent}{parent}
\SetKwFunction{haschildren}{has-chldn}
\SetKwFunction{remove}{remove}
\SetKwFunction{makeleft}{make-left}
\SetKwFunction{makeright}{make-right}
\SetKwFunction{expandleft}{expand-left}
\SetKwFunction{createrightbranches}{create-right-brs}
\SetKwFunction{makerightbranch}{make-right-br}
\SetKwFunction{backtrack}{backtrack}
\SetKwFunction{updatesweepbit}{update-sweep-bit}
\SetKwFunction{runtotermination}{run-to-termination}

\SetKwFunction{mainsimple}{explore}
\SetKwFunction{newnode}{nod}
\SetKwFunction{iscutoff}{cutoff}

\usepackage[ruled]{algorithm2e}

\makeatletter
\let\old@algocf@pre@ruled\@algocf@pre@ruled
\renewcommand{\@algocf@pre@ruled}{\Hy@raisedlink{\hyper@anchorstart{algocf.\thealgocf}\hyper@anchorend}\old@algocf@pre@ruled}
\makeatother

 \newcommand{\defn}{\ensuremath{\mathrel{\hat{=}}}}
\newlength{\defwidth}
\newcommand\cfl      {\mathrel{\#}}

\newcommand\indep    {\mathrel{\text{\rotatebox[origin=c]{45}{\smaller$\Box$}}}}
\newcommand\depen    {\mathrel{\text{\rotatebox[origin=c]{45}{\smaller$\boxplus$}}}}
\newcommand\eqdef    {\mathrel{:=}}

\newcommand\fire[1]  {\mathrel{\raisebox{-1.9pt}{$\xrightarrow{#1}$}}}

\newcommand\conf[1]     {\mathop{\mathit{conf}} (#1)}

\newcommand\reach[1]    {\mathop{\mathit{reach}} (#1)}
\newcommand\runs[1]     {\mathop{\mathit{runs}} (#1)}

\newcommand\state[1]    {\mathop{\mathit{state}} (#1)}

\newcommand\inter[1]    {\mathop{\mathit{inter}} (#1)}

\newcommand\enabl[1]    {\mathop{\mathit{enabl}} (#1)}
\newcommand\en[1]       {\mathop{\mathit{en}} (#1)}
\newcommand\ex[1]       {\mathop{\mathit{ex}} (#1)}
\newcommand\cex[1]      {\mathop{\mathit{cex}} (#1)}

\newcommand\effect[1]   {\mathop{\mathit{effect}} (#1)}
\newcommand\tid[1]      {\mathop{\mathit{tid}} (#1)}

\newcommand\cn[1]      {\mathop{\mathit{cn}} (#1)}

\newcommand\causes[1]   {\left \lceil #1 \right \rceil}
\newcommand\future[1]   {\left \lfloor #1 \right \rfloor}

\newcommand\po        {\mathcal{E}}

\newcommand\les         {\mathcal{E}}

\newcommand\unf[1]      {\mathcal{U}_{#1}}

\newcommand\unfpindep   {\unf{P,\indep}\!}

\newcommand\mem      {\mathcal{M}}
\newcommand\locks    {\mathcal{L}}
\newcommand\conds    {\mathcal{C}}

\newcommand\N        {\mathbb{N}}
\newcommand\Np       {{\N^*}}

\newcommand\Z        {\mathbb{Z}}

\newcommand\local    {\texttt{local}}
\newcommand\lock     {\texttt{lock}}
\newcommand\unlock   {\texttt{unlock}}
\newcommand\wait     {\texttt{wait}}
\newcommand\signal   {\texttt{signal}}
\newcommand\bcast    {\texttt{bcast}}

\newcommand\loc      {\textsf{loc}}
\newcommand\acq      {\textsf{acq}}
\newcommand\rel      {\textsf{rel}}
\newcommand\wa       {\textsf{w$_1$}}
\newcommand\waa      {\textsf{w$_2$}}
\newcommand\sigg     {\textsf{sig}}
\newcommand\bro      {\textsf{bro}}

\newcommand\ruleloc  {\textsc{Loc}}
\newcommand\ruleacq  {\textsc{Acq}}
\newcommand\rulerel  {\textsc{Rel}}
\newcommand\rulewa   {\textsc{W1}}
\newcommand\rulewaa  {\textsc{W2}}
\newcommand\rulesigg {\textsc{Sig}}
\newcommand\rulesiggl{\textsc{Sig'}}
\newcommand\rulebro  {\textsc{Bro}}

\newcommand\vs       {vs.\@\xspace}

\newcommand\aka      {a.k.a.\xspace}
\newcommand\wrt      {w.r.t.\@\xspace}

\newcommand\eg       {e.g.\@\xspace}
\newcommand\etc      {etc.\@\xspace}
\newcommand\ie       {i.e.\@\xspace}

\newcommand\st       {s.t.\@\xspace}
\newcommand\resp     {resp.\@\xspace}

\newcommand\tup[1]   {\langle#1\rangle}
\newcommand\set[1]   {{\{ #1 \mathclose \}}}

\newcommand\eqtag[1] {\hfill{} \refstepcounter{equation} \label{#1} (\arabic{equation})}

\ifcav{}
\usepackage[firstpage]{draftwatermark}
\SetWatermarkText{\hspace*{4.9in}\raisebox{8in}{\includegraphics[scale=0.1]{aec-badge-cav}}}
\SetWatermarkAngle{0}
\fi

\begin{document}

\title{Symbolic Partial-Order Execution\\ for Testing Multi-Threaded Programs}

\ifcav{}
\else
\subtitle{Extended Version}
\fi

\ifcav{}
\author{Daniel Schemmel\inst{1}\orcidID{0000-0001-8769-7813} \and
Julian Büning\inst{1}\orcidID{0000-0003-3917-6858} \and
César Rodríguez\inst{2,3} \and
David Laprell\inst{1}\orcidID{0000-0002-2219-0867} \and
Klaus Wehrle\inst{1}\orcidID{0000-0001-7252-4186}}
\else
\author{Daniel Schemmel\inst{1} \and
Julian Büning\inst{1} \and
César Rodríguez\inst{2,3} \and
David Laprell\inst{1} \and
Klaus Wehrle\inst{1}}
\fi
\authorrunning{D. Schemmel et al.}
\institute{RWTH Aachen University, Aachen, Germany
\email{\{daniel.schemmel,julian.buening,david.laprell,wehrle\}@comsys.rwth-aachen.de}
\and Diffblue Ltd., Oxford, UK
\and Université Paris 13, Sorbonne Paris Cité, CNRS, France
\email{cesar.rodriguez@lipn.fr}
}
 \maketitle

\begin{abstract}

We describe a technique for systematic testing of multi-threaded programs.
We combine Quasi-Optimal Partial-Order Reduction, a state-of-the-art technique
that tackles path explosion due to
interleaving non-determinism, with symbolic execution to handle data non-determinism.
Our technique iteratively and exhaustively finds all executions of the program.
It represents program executions using partial orders and finds the next execution using an underlying unfolding semantics.
We avoid the exploration of redundant program traces using cutoff events.
We implemented our technique as an extension of KLEE and evaluated it on a set of large multi-threaded C programs.
Our experiments found several previously undiscovered bugs and undefined behaviors in memcached and GNU sort, showing that the new method is capable of finding bugs in industrial-size benchmarks.

\end{abstract}

\mathligsoff{}

\setcounter{footnote}{0}

\section{Introduction}\label{sec:intro}

Advances in formal testing and the increased
availability of affordable concurrency have spawned two opposing trends:
While it has become possible to analyze increasingly complex sequential
programs in new and powerful ways, many projects are now embracing parallel
processing to fully exploit modern hardware, thus raising the bar for
practically useful formal testing.
In order to make formal testing accessible to software developers working on
parallel programs, two main problems need to be solved.  Firstly, a significant
portion of the API in concurrency libraries such as
\mbox{\texttt{libpthread}} must be supported.  Secondly, the analysis must be
accessible to non-experts in formal verification.  Currently, this niche is
mostly occupied by manual and fuzz testing, oftentimes combined with dynamic
concurrency checkers such as ThreadSanitizer~\cite{tsan} or
Helgrind~\cite{helgrind}.

Data non-determinism in sequential and concurrent programs,
and scheduling non-determinism are two major sources of path explosion in
program analysis.
\emph{Symbolic
execution}~\cite{kingSymbolicExecutionProgram1976,KLEE,threedecades,sage,symbolicpathfinder}
is a technique to reason about input data in sequential
programs. It is capable of dealing with real-world programs.
Partial-Order Reductions (PORs)~\cite{AAJS14,RSSK15,God96,FG05}
are a large family of techniques to explore a reduced number of
thread interleavings without missing any relevant behavior.

In this paper we propose a technique that combines symbolic execution and
a Quasi-Optimal POR~\cite{NRSCP18}. In essence, our approach
(1) runs the program using a symbolic executor,
(2) builds a partial order representing the occurrence of
POSIX threading synchronization primitives
(library functions \mbox{\texttt{pthread\_*}})
seen during that execution,
(3) adds the partial order to an underlying tree-like, unfolding~\cite{RSSK15,Mcm93} data
structure,
(4) computes the first events of the next partial orders to explore,
and
(5) selects a new partial order to explore and starts again.
We use cutoff events~\cite{Mcm93}
to prune the exploration of different traces that reach the
same state, thus natively dealing with non-terminating executions.

We implemented our technique as an extension of KLEE.\@
During the evaluation of this prototype we found nine bugs (that we attribute to
four root causes) in the production version of memcached.
All of these bugs have since been confirmed by the memcached maintainers and are
fixed as of version 1.5.21.
Our tool handles
a significant portion of the POSIX threading API~\cite{POSIX},
including barriers, mutexes and condition variables without being significantly
harder to use than common fuzz testing tools.

The main challenge that our approach needs to address is that of scalability in
the face of an enormous state space.
We tackle this challenge by detecting whenever any two Mazurkiewicz traces reach
the same program state to only further explore one of them.
Additionally, we exploit the fact that data races on non-atomic variables cause
undefined behavior in C~\cite[§~5.1.2.4/35]{C18}, which means that any
unsynchronized memory access is, strictly speaking, a bug.
By adding a data race detection algorithm, we can thereby restrict thread
scheduling decisions to synchronization primitives, such as operations on
mutexes and condition variables, which significantly reduces the state space.

This work has three core contributions, the combination of which enables the
analysis of real-world multi-threaded programs (see also \cref{sec:related} for
related work):

\begin{enumerate}
   \item A partial-order reduction algorithm capable of handling real-world POSIX
      programs that use an arbitrary amount of threads,
      mutexes and condition variables. Our algorithm continues
      analysis in the face of deadlocks.
   \item A cutoff algorithm that recognizes whenever two Mazurkiewicz traces
      reach the same program state, as identified by its actual memory contents.
      This significantly prunes the search space and even enables the
      partial-order reduction to deal with non-terminating executions.
	\item An implementation that finds real-world bugs.
\end{enumerate}

\ifcav{}
We also present an extended, more in-depth version of this paper~\cite{porse-arxiv}.
\else
This paper is an extended version of a paper~\cite{porse-cav} that we presented at CAV 2020 (32nd International Conference on Computer Aided Verification).
\fi \section{Overview}\label{sec:overview}

The technique proposed in this paper can be described as a
process of 5 conceptual steps, each of which we describe in a section below:

\subsection{Sequential Executions}

Consider the program shown in \cref{fig:overview:program}.
Assume that all variables are initially set to zero.
The statement \texttt{a = in()} initializes variable \texttt{a} non-deterministically.
A \emph{run} of the program is a sequence of \emph{actions}, \ie, pairs
\(\tup{i,s}\) where~\( i \in \N \) identifies a thread that executes a statement~\(s\).
For instance, the sequence
\[
\sigma_1 \eqdef
\tup{1, \text{\texttt{a=in()}}},
\tup{1, \text{\texttt{c=3}}},
\tup{2, \text{\texttt{b=c}}},
\tup{2, \text{\texttt{a<0}}},
\tup{2, \text{\texttt{puts("n")}}}
\]
is a run of \cref{fig:overview:program}.
This run represents all program paths where both statements of thread~1
run before the statements of thread~2, and where the statement
\texttt{a = in()} initializes variable~\texttt{a} to a negative number.
In our notion of run, concurrency is represented explicitly (via thread
identifiers) and data non-determinism is represented symbolically (via
constraints on program variables).
To keep things simple the example only has atomic integers (implicitly guarded
by locks) instead of POSIX synchronization primitives.

\begin{figure*}[t]
\centering
\begin{tikzpicture}[
	font=\fontsize{6}{6.25}\selectfont{},
	event/.style={
		draw,
		line width=0.015cm,
		text width=0.3cm,
		align=center,
		inner xsep=0,
		inner ysep=0.05cm,
	},
	lbl/.style={
		inner xsep=0,
		inner ysep=0.04cm,
},
	succ/.style={
		draw,
		->,
		>={Triangle[length=0.1cm,width=0.1cm]},
		line width=0.02cm,
	},
	conflict/.style={
		draw,
		-,
		color=red!80!black,
		>={Triangle[length=0.1cm,width=0.1cm]},
		densely dashed,
		line width=0.02cm,
	},
	sep/.style={
		draw,
		line width=0.01cm,
		dash pattern={on 0.3cm off 0.2cm},
	},
	epast/.style={
		draw,
		line width=0.015cm,
		text width=0.3cm,
		align=center,
		inner xsep=0,
		inner ysep=0.05cm,
	},
	epresent/.style={
		draw,
		color=black,
		fill=black,
		text=white,
		line width=0.015cm,
		text width=0.3cm,
		align=center,
		inner xsep=0,
		inner ysep=0.05cm,
	},
	efuture/.style={
		draw,
		circle,
		line width=0.015cm,
		text width=0.27cm,
		align=center,
		inner sep=0,
	},
	subcap/.style={
		anchor=north,
		inner sep=0,
},
]
	\begin{scope}[shift={(3.0,0)}]
		\def\xsep{0.90cm}
		\def\ysep{0.90cm}
		\def\yseps{0.65cm}
		\def\arrowstub{0.035cm}
		\def\xspace{0.515cm}
		\def\subcapy{-2.4cm}
		\def\subcapw{1cm}
		\def\sepstart{0cm}
		\def\sepend{-2.9cm}

		\def\txtaa{\(\tup{1, \text{\texttt{a=in()}}}\)}
		\def\txtab{\(\tup{1, \text{\texttt{c=3}}}\)}
		\def\txtba{\(\tup{2, \text{\texttt{b=c}}}\)}
		\def\txtbbt{\(\tup{2, \text{\texttt{a>=0}}}\)}
		\def\txtbbf{\(\tup{2, \text{\texttt{a<0}}}\)}
		\def\txtbct{\(\tup{2, \text{\texttt{"y"}}}\)}
		\def\txtbcf{\(\tup{2, \text{\texttt{"n"}}}\)}

		\begin{scope}[shift={(-3.0,0)}]
			\node[anchor=base west,inner sep=0,font=\scshape{}\fontsize{8}{8.5}\selectfont{}] at (-0.15,-0.1) {Thread 1};
			\node[anchor=north west] at (0,0) {
				\begin{lstlisting}[
					basicstyle=\ttfamily{}\fontsize{6}{7}\selectfont{},
					numberstyle=\fontsize{5}{5}\selectfont{}\color{black!60},
					numbersep=3pt,
				]
atomic_int a = in();
atomic_int c = 3;
				\end{lstlisting}
			};

			\node[anchor=base west,inner sep=0,font=\scshape{}\fontsize{8}{8.5}\selectfont{}] at (-0.15,-1.1) {Thread 2};
			\node[anchor=north west] at (0,-1.0) {
				\begin{lstlisting}[
					basicstyle=\ttfamily{}\fontsize{6}{7}\selectfont{},
					numberstyle=\fontsize{5}{5}\selectfont{}\color{black!60},
					numbersep=3pt,
				]
atomic_int b = c;
if(a >= 0)
	puts("y");
else
	puts("n");
				\end{lstlisting}
			};
			\node[subcap] at (1cm,\subcapy) {\begin{minipage}{1cm}\subcaption{}\label{fig:overview:program}\end{minipage}};
		\end{scope}
\begin{scope}[shift={(0,0)}]
			\node[event] (e1) at (0*\xsep,-0*\yseps) {1};
			\node[lbl,anchor=west,inner ysep=0] (l1) at (e1.east) {\txtaa};

			\node[event] (e2) at (0*\xsep,-1*\yseps) {2};
			\node[lbl,anchor=west] (l2) at (e2.east) {\txtab};
			\path[succ] (e1.south) -- (e2.north);

			\node[event] (e3) at (1*\xsep,-2*\yseps) {3};
			\node[lbl,anchor=east] (l3) at (e3.west) {\txtba};
			\path[succ] (e2.south) -- (e3.north west);

			\node[event] (e4) at (1*\xsep,-3*\yseps) {4};
			\node[lbl,anchor=east] (l4) at (e4.west) {\txtbbt};
			\path[succ] (e3.south) -- (e4.north);

			\node[event] (e5) at (1*\xsep,-4*\yseps) {5};
			\node[lbl,anchor=east] (l5) at (e5.west) {\txtbct};
			\path[succ] (e4.south) -- (e5.north);

			\node[subcap] at (1.3cm,\subcapy) {\begin{minipage}{\subcapw}\subcaption{}\label{fig:overview:config1}\end{minipage}};
		\end{scope}
\begin{scope}[shift={(1*\xsep + 2*\xspace,0)}]
			\node[event] (e1) at (0*\xsep,-0*\yseps) {1};
			\node[lbl,anchor=west,inner ysep=0] (l1) at (e1.east) {\txtaa};

			\node[event] (e2) at (0*\xsep,-1*\yseps) {2};
			\node[lbl,anchor=west] (l2) at (e2.east) {\txtab};
			\path[succ] (e1.south) -- (e2.north);

			\node[event] (e3) at (1*\xsep,-2*\yseps) {3};
			\node[lbl,anchor=east] (l3) at (e3.west) {\txtba};
			\path[succ] (e2.south) -- (e3.north west);

			\node[event] (e6) at (1*\xsep,-3*\yseps) {6};
			\node[lbl,anchor=east] (l6) at (e6.west) {\txtbbf};
			\path[succ] (e3.south) -- (e6.north);

			\node[event] (e7) at (1*\xsep,-4*\yseps) {7};
			\node[lbl,anchor=east] (l7) at (e7.west) {\txtbcf};
			\path[succ] (e6.south) -- (e7.north);

			\node[subcap] at (1.3cm,\subcapy) {\begin{minipage}{\subcapw}\subcaption{}\label{fig:overview:config2}\end{minipage}};
		\end{scope}
\begin{scope}[shift={(2*\xsep + 4*\xspace,0)}]
			\node[event] (e1) at (0*\xsep,-0*\ysep) {1};
			\node[lbl,anchor=north] (l1) at (e1.south) {\txtaa};

			\path let \p1 = (e1.north) in
				node[lbl,anchor=north] (l8) at (1*\xsep,\y1) {\txtba};
			\node[event,anchor=north] (e8) at ($(l8.south)+(0,\arrowstub)$) {8};

			\node[event] (e9) at (0*\xsep,-1*\ysep) {9};
			\node[lbl,anchor=north] (l9) at (e9.south) {\txtab};
			\path[succ] ($(l1.south)+(0,\arrowstub)$) -- (e9.north);
			\path[succ] (e8.south) -- (e9.north east);

			\path let \p1 = (e9.north) in
				node[lbl,anchor=north] (l10) at (1*\xsep,\y1+\arrowstub) {\txtbbt};
			\node[event,anchor=north] (e10) at ($(l10.south)+(0,\arrowstub)$) {10};
			\path[succ] let \p1 = (e9.north), \p2 = (e9.north east) in
				($(l1.south)+(0,\arrowstub)$) -- ($(l10.north)-(\x2-\x1,\arrowstub)$);
			\path[succ] (e8.south) -- ($(l10.north)-(0,\arrowstub)$);

			\path let \p1 = (e1.north), \p2 = (e1.center) in
				node[lbl,anchor=north] (l11) at (1*\xsep, -2*\ysep + \y1 - \y2 + \arrowstub) {\txtbct};
			\node[event,anchor=north] (e11) at ($(l11.south)+(0,\arrowstub)$) {11};
			\path[succ] (e10.south) -- ($(l11.north)-(0,\arrowstub)$);

			\node[subcap] at (0.5*\xsep,\subcapy) {\begin{minipage}{\subcapw}\subcaption{}\label{fig:overview:config3}\end{minipage}};
		\end{scope}
\begin{scope}[shift={(3*\xsep + 6*\xspace,0)}]
			\node[event] (e1) at (0*\xsep,-0*\ysep) {1};
			\node[lbl,anchor=north] (l1) at (e1.south) {\txtaa};

			\path let \p1 = (e1.north) in
				node[lbl,anchor=north] (l8) at (1*\xsep,\y1) {\txtba};
			\node[event,anchor=north] (e8) at ($(l8.south)+(0,\arrowstub)$) {8};

			\node[event] (e9) at (0*\xsep,-1*\ysep) {9};
			\node[lbl,anchor=north] (l9) at (e9.south) {\txtab};
			\path[succ] ($(l1.south)+(0,\arrowstub)$) -- (e9.north);
			\path[succ] (e8.south) -- (e9.north east);

			\path let \p1 = (e9.north) in
				node[lbl,anchor=north] (l12) at (1*\xsep,\y1+\arrowstub) {\txtbbf};
			\node[event,anchor=north] (e12) at ($(l12.south)+(0,\arrowstub)$) {12};
			\path[succ] let \p1 = (e9.north), \p2 = (e9.north east) in
				($(l1.south)+(0,\arrowstub)$) -- ($(l12.north)-(\x2-\x1,\arrowstub)$);
			\path[succ] (e8.south) -- ($(l12.north)-(0,\arrowstub)$);

			\path let \p1 = (e1.north), \p2 = (e1.center) in
				node[lbl,anchor=north] (l13) at (1*\xsep, -2*\ysep + \y1 - \y2 + \arrowstub) {\txtbcf};
			\node[event,anchor=north] (e13) at ($(l13.south)+(0,\arrowstub)$) {13};
			\path[succ] (e12.south) -- ($(l13.north)-(0,\arrowstub)$);

			\node[subcap] at (0.5*\xsep,\subcapy) {\begin{minipage}{\subcapw}\subcaption{}\label{fig:overview:config4}\end{minipage}};
		\end{scope}
\begin{scope}[shift={(4*\xsep + 8*\xspace,0)}]
			\path let \p1 = (e1.north) in
				node[lbl,anchor=north] (l8) at (1*\xsep,\y1) {\txtba};
			\node[event,anchor=north] (e8) at ($(l8.south)+(0,\arrowstub)$) {8};

			\path let \p1 = (e9.north) in
				node[lbl,anchor=north] (l14) at (1*\xsep,\y1+\arrowstub) {\txtbbt};
			\node[event,anchor=north] (e14) at ($(l14.south)+(0,\arrowstub)$) {14};
			\path[succ] (e8.south) -- ($(l14.north)-(0,\arrowstub)$);

			\node[event] (e15) at (0*\xsep,-2*\ysep) {15};
			\node[lbl,anchor=north] (l15) at (e15.south) {\txtaa};
			\path[succ] (e14.south) -- (e15.north east);

			\node[event] (e16) at (0*\xsep,-3*\ysep) {16};
			\node[lbl,anchor=east] (l16) at (e16.west) {\txtab};
			\path[succ] ($(l15.south)+(0,\arrowstub)$) -- (e16.north);

			\path let \p1 = (e8.north), \p2 = (e8.center) in
				node[lbl,anchor=north,inner xsep=0] (l17) at (1*\xsep, -2*\ysep + \y1 - \y2 + \arrowstub) {\txtbct};
			\node[event,anchor=north] (e17) at ($(l17.south)+(0,\arrowstub)$) {17};
			\path[succ] (e14.south) -- ($(l17.north)-(0,\arrowstub)$);

			\node[subcap] at (1*\xsep,\subcapy) {\begin{minipage}[t][0cm]{0.4cm}\subcaption{}\label{fig:overview:config5}\end{minipage}};
		\end{scope}
	\end{scope}

	\begin{scope}[shift={(0,-3.3)}]
		\def\xsep{0.5655cm}
		\def\ysep{0.60cm}
		\def\xspace{0.75cm}
		\def\subcapy{-2.2cm}
		\def\subcapspace{0.25cm}
		\def\sepstart{0cm}
		\def\sepend{-2.9cm}

		\begin{scope}[shift={(0*\xsep,0)}]
			\node[event] (e1) at (1*\xsep,-0*\ysep) {1};
			
			\node[event] (e2) at (0*\xsep,-1*\ysep) {2};
			\path[succ] (e1.south) -- (e2.north);

			\node[event] (e3) at (0*\xsep,-2*\ysep) {3};
			\path[succ] (e2.south) -- (e3.north);

			\node[event] (e4) at (0*\xsep,-3*\ysep) {4};
			\path[succ] (e3.south) -- (e4.north);

			\node[event] (e5) at (0*\xsep,-4*\ysep) {5};
			\path[succ] (e4.south) -- (e5.north);

			\node[event] (e6) at (1*\xsep,-3*\ysep) {6};
			\path[succ] (e3.south) -- (e6.north);
			\path[conflict] (e4.east) -- (e6.west);

			\node[event] (e7) at (1*\xsep,-4*\ysep) {7};
			\path[succ] (e6.south) -- (e7.north);

			\node[event] (e8) at (3*\xsep,-0*\ysep) {8};
			\path[conflict] (e2.east) -- (e8.west);
			
			\node[event] (e9) at (1*\xsep,-1*\ysep) {9};
			\path[succ] (e1.south) -- (e9.north);
			\path[succ] (e8.south) -- (e9.north);
			
			\node[event] (e10) at (2*\xsep,-1*\ysep) {10};
			\path[succ] (e1.south) -- (e10.north);
			\path[succ] (e8.south) -- (e10.north);
			
			\node[event] (e11) at (2*\xsep,-2*\ysep) {11};
			\path[succ] (e10.south) -- (e11.north);
			
			\node[event] (e12) at (3*\xsep,-1*\ysep) {12};
			\path[succ] (e1.south) -- (e12.north);
			\path[succ] (e8.south) -- (e12.north);
			\path[conflict] (e10.east) -- (e12.west);
			
			\node[event] (e13) at (3*\xsep,-2*\ysep) {13};
			\path[succ] (e12.south) -- (e13.north);
			
			\node[event] (e14) at (4*\xsep,-1*\ysep) {14};
			\path[conflict] (e1.east) -- (e14.west);
			\path[succ] (e8.south) -- (e14.north);
			
			\node[event] (e15) at (3*\xsep,-3*\ysep) {15};
			\path[succ] (e14.south) -- (e15.north);
			
			\node[event] (e16) at (3*\xsep,-4*\ysep) {16};
			\path[succ] (e15.south) -- (e16.north);
			
			\node[event] (e17) at (4*\xsep,-3*\ysep) {17};
			\path[succ] (e14.south) -- (e17.north);

			\node[subcap] at (2*\xsep,\subcapy) {\begin{minipage}[t][0cm]{1cm}\subcaption{}\label{fig:overview:unfolding}\end{minipage}};
		\end{scope}
\begin{scope}[shift={(5*\xsep,0)}]
			\node[epresent] (e1) at (0.75*\xsep,-0*\ysep) {1};
			
			\node[epresent] (e2) at (0*\xsep,-1*\ysep) {2};
			\path[succ] (e1.south) -- (e2.north);

			\node[epresent] (e3) at (0*\xsep,-2*\ysep) {3};
			\path[succ] (e2.south) -- (e3.north);

			\node[epresent] (e4) at (0*\xsep,-3*\ysep) {4};
			\path[succ] (e3.south) -- (e4.north);

			\node[epresent] (e5) at (0*\xsep,-4*\ysep) {5};
			\path[succ] (e4.south) -- (e5.north);

			\node[efuture] (e6) at (1*\xsep,-3*\ysep) {6};
			\path[succ] (e3.south) -- (e6.north);
			\path[conflict] (e4.east) -- (e6.west);

			\node[efuture] (e8) at (2*\xsep,-0*\ysep) {8};
			\path[conflict] (e2.east) -- (e8.west);
			
			\node[subcap] at (0.875*\xsep,\subcapy) {\begin{minipage}[t][0cm]{1cm}\subcaption{}\label{fig:overview:algo1}\end{minipage}};
		\end{scope}
\begin{scope}[shift={(7.25*\xsep,0)}]
			\node[epresent] (e1) at (0.75*\xsep,-0*\ysep) {1};
			
			\node[epresent] (e2) at (0*\xsep,-1*\ysep) {2};
			\path[succ] (e1.south) -- (e2.north);

			\node[epresent] (e3) at (0*\xsep,-2*\ysep) {3};
			\path[succ] (e2.south) -- (e3.north);

			\node[epast] (e4) at (0*\xsep,-3*\ysep) {4};
			\path[succ] (e3.south) -- (e4.north);

			\node[epresent] (e6) at (1*\xsep,-3*\ysep) {6};
			\path[succ] (e3.south) -- (e6.north);
			\path[conflict] (e4.east) -- (e6.west);

			\node[epresent] (e7) at (1*\xsep,-4*\ysep) {7};
			\path[succ] (e6.south) -- (e7.north);

			\node[efuture] (e8) at (2*\xsep,-0*\ysep) {8};
			\path[conflict] (e2.east) -- (e8.west);

			\node[subcap] at (0.125*\xsep,\subcapy) {\begin{minipage}[t][0cm]{1cm}\subcaption{}\label{fig:overview:algo2}\end{minipage}};
		\end{scope}
\begin{scope}[shift={(9.625*\xsep,0)}]
			\node[epresent] (e1) at (1*\xsep,-0*\ysep) {1};
			
			\node[epast] (e2) at (0*\xsep,-1*\ysep) {2};
			\path[succ] (e1.south) -- (e2.north);

			\node[epresent] (e8) at (3*\xsep,-0*\ysep) {8};
			\path[conflict] (e2.east) -- (e8.west);
			
			\node[epresent] (e9) at (1*\xsep,-1*\ysep) {9};
			\path[succ] (e1.south) -- (e9.north);
			\path[succ] (e8.south) -- (e9.north);
			
			\node[epresent] (e10) at (2*\xsep,-1*\ysep) {10};
			\path[succ] (e1.south) -- (e10.north);
			\path[succ] (e8.south) -- (e10.north);
			
			\node[epresent] (e11) at (2*\xsep,-2*\ysep) {11};
			\path[succ] (e10.south) -- (e11.north);
			
			\node[efuture] (e12) at (3*\xsep,-1*\ysep) {12};
			\path[succ] (e1.south) -- (e12.north);
			\path[succ] (e8.south) -- (e12.north);
			\path[conflict] (e10.east) -- (e12.west);
			
			\node[efuture] (e14) at (4*\xsep,-1*\ysep) {14};
			\path[conflict] (e1.east) -- (e14.west);
			\path[succ] (e8.south) -- (e14.north);

			\node[subcap] at (2*\xsep,\subcapy) {\begin{minipage}[t][0cm]{1cm}\subcaption{}\label{fig:overview:algo3}\end{minipage}};
		\end{scope}
\begin{scope}[shift={(14.625*\xsep,0)}]
			\node[epresent] (e1) at (1*\xsep,-0*\ysep) {1};
			
			\node[epast] (e2) at (0*\xsep,-1*\ysep) {2};
			\path[succ] (e1.south) -- (e2.north);

			\node[epresent] (e8) at (3*\xsep,-0*\ysep) {8};
			\path[conflict] (e2.east) -- (e8.west);
			
			\node[epresent] (e9) at (1*\xsep,-1*\ysep) {9};
			\path[succ] (e1.south) -- (e9.north);
			\path[succ] (e8.south) -- (e9.north);
			
			\node[epast] (e10) at (2*\xsep,-1*\ysep) {10};
			\path[succ] (e1.south) -- (e10.north);
			\path[succ] (e8.south) -- (e10.north);
			
			\node[epresent] (e12) at (3*\xsep,-1*\ysep) {12};
			\path[succ] (e1.south) -- (e12.north);
			\path[succ] (e8.south) -- (e12.north);
			\path[conflict] (e10.east) -- (e12.west);
			
			\node[epresent] (e13) at (3*\xsep,-2*\ysep) {13};
			\path[succ] (e12.south) -- (e13.north);
			
			\node[efuture] (e14) at (4*\xsep,-1*\ysep) {14};
			\path[conflict] (e1.east) -- (e14.west);
			\path[succ] (e8.south) -- (e14.north);

			\node[subcap] at (2*\xsep,\subcapy) {\begin{minipage}[t][0cm]{1cm}\subcaption{}\label{fig:overview:algo4}\end{minipage}};
		\end{scope}
\begin{scope}[shift={(19*\xsep,0)}]
			\node[epast] (e1) at (0*\xsep,-0*\ysep) {1};

			\node[epresent] (e8) at (1.25*\xsep,-0*\ysep) {8};
			
			\node[epresent] (e14) at (2*\xsep,-1*\ysep) {14};
			\path[conflict] (e1.east) -- (e14.west);
			\path[succ] (e8.south) -- (e14.north);
			
			\node[epresent] (e15) at (1*\xsep,-3*\ysep) {15};
			\path[succ] (e14.south) -- (e15.north);
			
			\node[epresent] (e16) at (1*\xsep,-4*\ysep) {16};
			\path[succ] (e15.south) -- (e16.north);
			
			\node[epresent] (e17) at (2*\xsep,-3*\ysep) {17};
			\path[succ] (e14.south) -- (e17.north);

			\node[subcap] at (1.875*\xsep,\subcapy) {\begin{minipage}[t][0cm]{0.4cm}\subcaption{}\label{fig:overview:algo5}\end{minipage}};
		\end{scope}
	\end{scope}

\end{tikzpicture}    \caption{A program (a) with its 5 partial-order runs (b-f), its
   unfolding (g) and the 5 steps used by our algorithm to visit the unfolding
   (h-l).}\label{fig:overview.atomic}
\end{figure*}
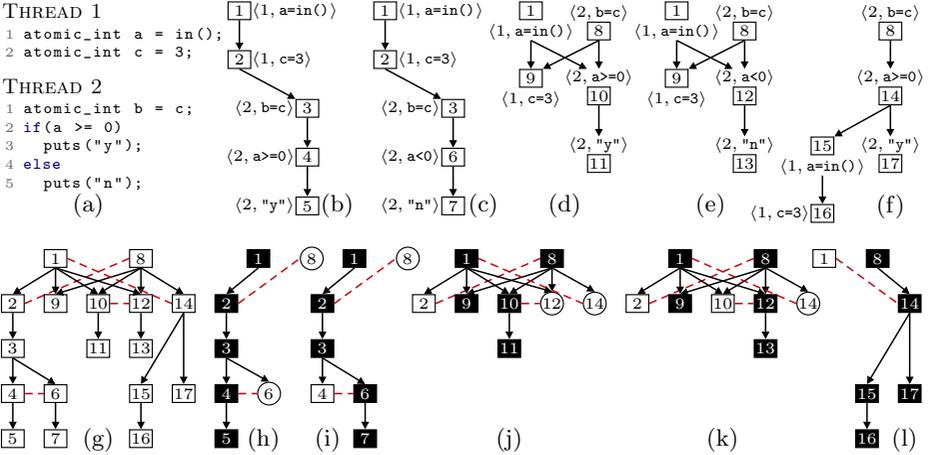

\subsection{Independence between Actions and Partial-Order Runs}\label{sec:overview.independence}

Many POR techniques use a notion called
\emph{independence}~\cite{God96} to avoid
exploring concurrent interleavings that lead to the same state.
An independence relation associates pairs of actions that commute (running them
in either order results in the same state).
For illustration purposes, in \cref{fig:overview.atomic} let us
consider two actions as \emph{dependent} iff either both of them belong to the same
thread or one of them writes into a variable which is read/written by the other.
Furthermore, two actions will be \emph{independent} iff they are not dependent.

A sequential run of the program can be viewed as a partial order
when we take into account the independence of actions.
These partial orders are
known as \emph{dependency graphs} in Mazurkiewicz trace theory~\cite{Maz87}
and as \emph{partial-order runs} in this paper.
\Cref{fig:overview:config1,fig:overview:config2,fig:overview:config3,fig:overview:config4,fig:overview:config5}
show all the partial-order runs of \cref{fig:overview:program}.
The partial-order run associated to the run \( \sigma_1 \) above
is~\cref{fig:overview:config2}.
For
\[
\sigma_2 \eqdef
\tup{2, \text{\texttt{b=c}}},
\tup{2, \text{\texttt{a>=0}}},
\tup{1, \text{\texttt{a=in()}}},
\tup{2, \text{\texttt{puts("y")}}},
\tup{1, \text{\texttt{c=3}}},
\]
we get the partial order shown in~\cref{fig:overview:config5}.

\subsection{Unfolding: Merging the Partial Orders}\label{sec:overview.unfolding}

An unfolding~\cite{Mcm93,ERV02,NPW81} is a tree-like structure that uses partial
orders to represent concurrent executions and conflict relations to represent
thread interference and data non-determinism.
We can define unfolding semantics for programs in two conceptual steps:
(1) identify isomorphic events that occur in different partial-order runs;
(2) bind the partial orders together using a conflict relation.

Two events are \emph{isomorphic} when they are structurally equivalent,
\ie, they have the same label (run the same action) and their causal (\ie,
happens-before) predecessors are (transitively) isomorphic.
The number within every event in
\cref{fig:overview:config1,fig:overview:config2,fig:overview:config3,fig:overview:config4,fig:overview:config5}
identifies isomorphic events.

Isomorphic events from different partial orders can be merged together
using a conflict relation for the un-merged parts of those partial orders.
To understand why conflict is necessary, consider the set of events \(C \eqdef \set{1, 2}\).
It obviously represents part of a partial-order run (\cref{fig:overview:config2}, for instance).
Similarly, events \(C' \eqdef \set{1, 8, 9}\) represent (part of) a run.
However, their union \(C \cup C'\) does not represent any run,
because (1) it does not describe what happens-before relation exists between
the dependent actions of events~2 and~8, and (2) it executes the
statement \texttt{c=3} twice.
Unfoldings fix this problem by introducing a \emph{conflict} relation between
events.
Conflicts are to unfoldings what branches are to trees.
If we declare that events~2 and~8 are in conflict, then any conflict-free (and
causally-closed) subset of~\(C \cup C'\) is exactly one of the original partial
orders.
This lets us merge the common parts of multiple partial orders without losing track of the original partial orders.

\Cref{fig:overview:unfolding} represents the unfolding of the program (after merging all 5
partial-order runs).
Conflicts between events are represented by dashed red lines.
Each original partial order can be retrieved by taking a (\ensuremath{\subseteq}-maximal)
set of events which is conflict-free (no two events in conflict are in the set)
and causally closed (if you take some event, then also take all its causal
predecessors).

For instance, the partial order in \cref{fig:overview:config3} can be retrieved by
resolving the conflicts between events 1 \vs{} 14, 2 \vs{} 8, 10 \vs{} 12
in favor of, \resp, 1, 8, 10.
Resolving in favor of 1 means that
events~14 to~17 cannot be selected, because they causally succeed 14.
Similarly, resolving in favor of 8 and 10 means that only events 9 and 11 remain eligible, which hold no conflicts among them---all other events are causal successors of either~2 or~12.

\subsection{Exploring the Unfolding}

Since the unfolding represents all runs of the program via a set of
compactly-merged, prefix-sharing partial orders,
enumerating all the behaviors of the program reduces to
exploring all partial-order runs represented in its unfolding.
Our algorithm iteratively enumerates all \ensuremath{\subseteq}-maximal partial-order runs.

In simplified terms, it proceeds as follows.
Initially we explore the black events shown in \cref{fig:overview:algo1},
therefore exploring the run shown in~\cref{fig:overview:config1}.
We discover the next partial order by computing the so-called \emph{conflicting
extensions} of the current partial order.
These are, intuitively, events in conflict with some event in our current
partial order but such that all its causal predecessors are in our current
partial order.
In \cref{fig:overview:algo1} these are shown in circles, events~8 and~6.

We now find the next partial order by
(1) selecting a conflicting extension, say event~6,
(2) removing all events in conflict with the selected extension and their causal
successors, in this case events~4 and~5, and (3) expanding the partial order
until it becomes maximal, thus exploring the
partial order~\cref{fig:overview:config2}, shown as the black events
of~\cref{fig:overview:algo2}.
Next we select event~8 (removing 2 and its causal successors) and explore
the partial order \cref{fig:overview:config3}, shown as the black events
of~\cref{fig:overview:algo3}.
Note that this reveals two new conflicting extensions that were hidden until
now, events~12 and~14 (hidden because 8 is a causal predecessor of them, but
was not in our partial order). Selecting either of the two extensions makes the
algorithm explore the last two partial orders.

\subsection{Cutoff Events: Pruning the Unfolding}

When the program has non-terminating runs, its unfolding will contain infinite
partial orders and the algorithm above will not finish.
To analyze non-terminating programs we use \emph{cutoff events}~\cite{Mcm93}.
In short, certain events do not need to be explored because they reach the same
state as another event that has been already explored using a shorter
(partial-order) run.
Our algorithm prunes the unfolding at these cutoff events, thus handling
terminating and non-terminating programs that repeatedly reach the same state.

\section{Main Algorithm}\label{sec:algo}

This section formally describes the approach presented in this paper.

\subsection{Programs, Actions, and Runs}\label{sec:programs.actions}

Let \(P \eqdef \tup{T, \locks, \conds}\) represent a (possibly non-terminating)
multi-threaded POSIX C program, where
\(T\) is the set of statements,
\( \locks \) is the set of POSIX mutexes used in the program, and
\( \conds \) is the set of condition variables.
This is a deliberately simplified presentation of our program syntax,
see~\aref{app:model} for full details.
We represent the behavior of each statement in~$P$ by an
\emph{action},
\ie, a pair \(\tup{i, b}\) in $A \subseteq \N \times B$, where $i \ge 1$ identifies
the thread executing the statement and~$b$ is the \emph{effect} of the
statement.
We consider the following effects:
\begin{align*}
 B \eqdef& ~(\set{\loc} \times T)
            \cup (\set{\acq,\rel} \times \locks)
            \cup (\set{\sigg} \times \conds \times \N) \\
     \cup& ~(\set{\bro} \times \conds \times 2^\N)
            \cup (\set{\wa,\waa} \times \conds \times \locks)
\end{align*}
Below we informally explain the intent of an effect and how
actions of different effects interleave with each other.
In~\aref{app:lts} we use \emph{actions} and \emph{effects} to
define labeled transition system semantics to~$P$.
Below we also (informally) define an independence relation
(see \cref{sec:overview.independence}) between actions.

\paragraph{Local actions.}
An action \(\tup{i, \tup{\loc, t}}\) represents the execution of a
\emph{local} statement~\(t\) from thread~\(i\), \ie, a statement which
manipulates local variables.
For instance,
the actions labeling events~1 and~3 in~\cref{fig:3:b} are local actions.
Note that local actions do not interfere with actions of other threads.
Consequently, they are only dependent on actions of the same thread.

\paragraph{Mutex lock/unlock.}
Actions
\(\tup{i, \tup{\acq, l}}\) and
\(\tup{i, \tup{\rel, l}}\)
respectively represent that thread~\(i\) locks or unlocks mutex
\( l \in \locks \).
The semantics of these actions correspond to the so-called \verb!NORMAL!
mutexes in the POSIX standard~\cite{POSIX}.
Actions of~\(\tup{\acq, l}\) or~\(\tup{\rel, l}\) effect are only dependent
on actions whose effect is an operation on the same mutex~\(l\)
($\acq$, $\rel$, $\wa$ or $\waa$, see below).
For instance the action of event~4~(\( \rel \)) in~\cref{fig:3:b} depends on the
action of event~6~(\( \acq \)).

\paragraph{Wait on condition variables.}
The occurrence of a \mbox{\texttt{pthread\_cond\_wait(c, l)}} statement is represented by
two separate actions of effect $\tup{\wa, c, l}$ and $\tup{\waa, c, l}$.
An action $\tup{i, \tup{\wa, c, l}}$
represents that thread~$i$ has atomically released the lock~$l$ \emph{and}
started waiting on condition variable~$c$.
An action $\tup{i, \tup{\waa, c, l}}$
indicates that thread~$i$ has been woken up by a \emph{signal} or
\emph{broadcast} operation on~$c$ \emph{and} that it successfully re-acquired
mutex~$l$.
For instance the action $\tup{1, \tup{\wa, c, m}}$ of event~10
in~\cref{fig:3:c} represents that thread~1 has released mutex~$m$
and is waiting for~$c$ to be signaled. After the signal happens (event~12)
the action $\tup{1, \tup{\waa, c, m}}$ of event~14 represents
that thread~1 wakes up and re-acquires mutex~$m$.
An action $\tup{i, \tup{\wa, c, l}}$
is dependent on any action whose effect operates on mutex~$l$
($\acq$, $\rel$, $\wa$ or $\waa$) as well as
signals  directed to thread~$i$ ($\tup{\sigg, c, i}$, see below),
lost signals ($\tup{\sigg, c, 0}$, see below),
and any broadcast ($\tup{\bro, c, W}$ for any~$W \subseteq \N$, see below).
Similarly, an action $\tup{i, \tup{\waa, c,l}}$
is dependent on any action whose effect operates on lock~$l$
as well as signals and broadcasts directed to thread~$i$
(that is, either $\tup{\sigg, c, i}$ or $\tup{\bro, c, W}$ when~$i \in W$).

\paragraph{Signal/broadcast on condition variables.}
An action $\tup{i, \tup{\sigg, c, j}}$, with $j \ge 0$ indicates that
thread~$i$ executed a \mbox{\texttt{pthread\_cond\_signal(c)}} statement. If $j = 0$
then no thread was waiting on condition variable~$c$, and the
\emph{signal} had no effect, as per the POSIX semantics.
We refer to these as \emph{lost signals}.
Example: events~7 and~17 in~\cref{fig:3:b,fig:3:d} are labeled by lost signals.
In both cases thread~1 was not waiting on the condition variable when the
signal happened.
However, when $j \ge 1$ the action represents that thread~$j$ wakes up
by this signal.
Whenever a signal wakes up a thread $j \ge 1$, we can always find a (unique)
$\wa$ action of thread~$j$ that happened before the signal and a unique $\waa$
action in thread~$j$ that happens after the signal.
For instance, event~12 in~\cref{fig:3:c} signals thread~1, which
went sleeping in the $\wa$ event~10 and wakes up in the $\waa$ event~14.
Similarly, an action $\tup{i, \tup{\bro, c, W}}$, with $W \subseteq \N$
indicates that thread~$i$ executed a \mbox{\texttt{pthread\_cond\_broadcast(c)}}
statement and any thread~$j$ such that $j \in W$ was woken up.
If $W = \emptyset$,
then no thread was waiting on condition variable~$c$
(\emph{lost broadcast}).
Lost signals and broadcasts on~$c$ depend on any
action of~$\tup{\wa, c, \cdot}$ effect as well as any non-lost
signal/broadcast on~$c$.
Non-lost signals and broadcasts on~$c$ that wake up thread~$j$
depend\footnote{The formal definition is slightly more complex,
see~\aref{app:indep.programs} for the details.}
on $\wa$ and $\waa$ actions of thread~$j$ as well as any signal/broadcast
(lost or not) on the same condition variable.

A \emph{run} of~$P$ is a sequence of actions in~$A^*$ which respects the
constraints stated above for actions.
For instance, a run for the program shown in~\cref{fig:3:a} is the sequence of
actions which labels \emph{any} topological order of the events shown
in any partial order
in~\cref{fig:3:b,fig:3:c,fig:3:d,fig:3:e}.
The sequence below,
\begin{multline*}
\tup{1, \tup{\loc, \text{\texttt{x=in()}}}},
\tup{2, \tup{\loc, \text{\texttt{y=1}}}},
\tup{1, \tup{\acq, m}}, \\
\tup{1, \tup{\loc, \text{\texttt{x>=0}}}},
\tup{1, \tup{\rel, m}},
\tup{2, \tup{\acq, m}}
\end{multline*}
is a run of~\cref{fig:3:a}.
Naturally, if $\sigma \in A^*$ is a run, any prefix of~$\sigma$ is also a run.
Runs explicitly represent concurrency, using thread identifiers,
and symbolically represent data non-determinism,
using constraints, as illustrated by the 1st and 4th actions of the run above.
We let $\runs P$ denote the set of all runs of~$P$.

A \emph{concrete state of~$P$} is a tuple that represents, intuitively, the program
counters of each thread, the values of all memory locations, the mutexes locked
by each thread, and, for each condition variable, the set of threads waiting
for it
(see~\aref{app:lts} for a formal definition).
Since runs represent operations on symbolic data, they reach a symbolic
state, which conceptually corresponds to a set of concrete states of~$P$.

The \emph{state of a run} $\sigma$, written $\state \sigma$,
is the set of all concrete states of~$P$ that are reachable
when the program executes the run~$\sigma$.
For instance, the run~$\sigma'$ given above reaches a state consisting on all
program states where
\mbox{\texttt{y}} is~1,
\mbox{\texttt{x}} is a non-negative number,
thread~2 owns mutex~\mbox{\texttt{m}} and its instruction pointer is at line~3,
and thread~1 has finished.
We let $\reach P \eqdef \bigcup_{\sigma \in \runs P} \state \sigma$
denote the set of all \emph{reachable states} of~$P$.

\begin{figure*}[t]
\centering
\begin{tikzpicture}[
	font=\fontsize{6}{6.25}\selectfont{},
	event/.style={
		draw,
		font=\fontsize{7}{8.5}\selectfont{},
		line width=0.015cm,
		text width=0.3cm,
		align=center,
		inner xsep=0.03cm,
		inner ysep=0.05cm,
	},
	lbl/.style={
			inner xsep=0.02cm,
			inner ysep=0.03cm,
	},
	succ/.style={
		draw,
		->,
		>={Triangle[length=0.1cm,width=0.1cm]},
		line width=0.02cm,
	},
	conflict/.style={
		draw,
		-,
		color=red!80!black,
		>={Triangle[length=0.1cm,width=0.1cm]},
		densely dashed,
		line width=0.02cm,
	},
	deadlock/.style={
		font=\fontsize{8}{8.5}\selectfont{},
	},
	subcap/.style={
		anchor=north,
		inner sep=0,
},
]
	\begin{scope}[shift={(4.0,0)}]
		\def\xsep{1.08cm}
		\def\ysep{0.75cm}
		\def\yseps{0.50cm}
		\def\ysepsd{0.65cm}
		\def\arrowstub{0.035cm}
		\def\xspace{0.56cm}
		\def\subcapy{-4.3cm}
		\def\subcapw{1cm}
		\def\sepstart{0cm}
		\def\sepend{-2.9cm}

		\def\txtaa{\(\tup{\loc, \text{\texttt{x=in()}}}\)}
		\def\txtab{\(\tup{\acq, m}\)}
		\def\txtact{\(\tup{\loc, \text{\texttt{x<0}}}\)}
		\def\txtacf{\(\tup{\loc, \text{\texttt{x>=0}}}\)}
		\def\txtad{\(\tup{\wa, c, m}\)}
		\def\txtae{\(\tup{\waa, c, m}\)}
		\def\txtaf{\(\tup{\rel, m}\)}
		\def\txtba{\(\tup{\loc, \text{\texttt{y=1}}}\)}
		\def\txtbb{\(\tup{\acq, m}\)}
		\def\txtbct{\(\tup{\sigg, c, 1}\)}
		\def\txtbcf{\(\tup{\sigg, c, 0}\)}
		\def\txtbd{\(\tup{\rel, m}\)}

		\begin{scope}[shift={(-3.9,0)}]
			\node[anchor=base west,inner sep=0,font=\scshape{}\fontsize{8}{8.5}\selectfont{}] at (-0.15,-0.1) {Thread 1};
			\node[anchor=north west] at (0,0) {
				\begin{lstlisting}[
					basicstyle=\ttfamily{}\fontsize{6}{7}\selectfont{},
					numberstyle=\fontsize{5}{5}\selectfont{}\color{black!60},
					numbersep=3pt,
				]
x = in();
pthread_mutex_lock(m);
if(x < 0)
  pthread_cond_wait(c, m);
pthread_mutex_unlock(m);
				\end{lstlisting}
			};

			\node[anchor=base west,inner sep=0,font=\scshape{}\fontsize{8}{8.5}\selectfont{}] at (-0.15,-2.20) {Thread 2};
			\node[anchor=north west] at (0,-2.10) {
				\begin{lstlisting}[
					basicstyle=\ttfamily{}\fontsize{6}{7}\selectfont{},
					numberstyle=\fontsize{5}{5}\selectfont{}\color{black!60},
					numbersep=3pt,
				]
y = 1;
pthread_mutex_lock(m);
pthread_cond_signal(c, m);
pthread_mutex_unlock(m);
				\end{lstlisting}
			};
			\node[subcap] at (1cm,\subcapy) {\begin{minipage}[t][0cm]{1cm}\subcaption{}\label{fig:3:a}\end{minipage}};
		\end{scope}
		\begin{scope}[shift={(0,0)}]
			\node[event] (e1) at (0*\xsep,-0*\ysep-0*\yseps) {1};
			\node[lbl,anchor=north] (l1) at (e1.south) {\txtaa};

			\node[event] (e2) at (0*\xsep,-1*\ysep-0*\yseps) {2};
			\node[lbl,anchor=west] (l2) at (e2.east) {\txtab};
			\path[succ] ($(l1.south)+(0,\arrowstub)$) -- (e2.north);

			\node[event] (e3) at (0*\xsep,-1*\ysep-1*\yseps) {3};
			\node[lbl,anchor=west] (l3) at (e3.east) {\txtacf};
			\path[succ] (e2.south) -- (e3.north);

			\node[event] (e4) at (0*\xsep,-1*\ysep-2*\yseps) {4};
			\node[lbl,anchor=west] (l4) at (e4.east) {\txtaf};
			\path[succ] (e3.south) -- (e4.north);

			\path let \p1 = (e1.north) in
				node[lbl,anchor=north,inner ysep=0] (l5) at (1*\xsep,\y1) {\txtba};
			\node[event,anchor=north] (e5) at ($(l5.south)+(0,-0.025)$) {5};

			\node[event] (e6) at (1*\xsep,-1*\ysep-2*\yseps-1*\ysepsd) {6};
			\node[lbl,anchor=east] (l6) at (e6.west) {\txtbb};
			\path[succ] (e4.south) -- (e6.north);
			\path[succ] (e5.south) to[bend left=20] (e6.north);

			\node[event] (e7) at (1*\xsep,-1*\ysep-3*\yseps-1*\ysepsd) {7};
			\node[lbl,anchor=east] (l7) at (e7.west) {\txtbcf};
			\path[succ] (e6.south) -- (e7.north);

			\node[event] (e8) at (1*\xsep,-1*\ysep-4*\yseps-1*\ysepsd) {8};
			\node[lbl,anchor=east] (l8) at (e8.west) {\txtbd};
			\path[succ] (e7.south) -- (e8.north);

			\node[subcap] at (0.5*\xsep,\subcapy) {\begin{minipage}[t][0cm]{1cm}\subcaption{}\label{fig:3:b}\end{minipage}};
		\end{scope}
		\begin{scope}[shift={(1*\xsep + 2*\xspace,0)}]
			\node[event] (e1) at (0*\xsep,-0*\ysep-0*\yseps) {1};
			\node[lbl,anchor=north] (l1) at (e1.south) {\txtaa};

			\node[event] (e2) at (0*\xsep,-1*\ysep-0*\yseps) {2};
			\node[lbl,anchor=west] (l2) at (e2.east) {\txtab};
			\path[succ] ($(l1.south)+(0,\arrowstub)$) -- (e2.north);

			\node[event] (e9) at (0*\xsep,-1*\ysep-1*\yseps) {9};
			\node[lbl,anchor=west] (l9) at (e9.east) {\txtact};
			\path[succ] (e2.south) -- (e9.north);

			\node[event] (e10) at (0*\xsep,-1*\ysep-2*\yseps) {10};
			\node[lbl,anchor=west] (l10) at (e10.east) {\txtad};
			\path[succ] (e9.south) -- (e10.north);

			\path let \p1 = (e1.north) in
				node[lbl,anchor=north,inner ysep=0] (l5) at (1*\xsep,\y1) {\txtba};
			\node[event,anchor=north] (e5) at ($(l5.south)+(0,-0.025)$) {5};

			\node[event] (e11) at (1*\xsep,-1*\ysep-2*\yseps-1*\ysepsd) {11};
			\node[lbl,anchor=east] (l11) at (e11.west) {\txtbb};
			\path[succ] (e10.south) -- (e11.north);
			\path[succ] (e5.south) to[out=280,in=80] (l10.east) to[out=260,in=60] (e11.north);

			\node[event] (e12) at (1*\xsep,-1*\ysep-3*\yseps-1*\ysepsd) {12};
			\node[lbl,anchor=east] (l12) at (e12.west) {\txtbct};
			\path[succ] (e11.south) -- (e12.north);

			\node[event] (e13) at (1*\xsep,-1*\ysep-4*\yseps-1*\ysepsd) {13};
			\node[lbl,anchor=east] (l13) at (e13.west) {\txtbd};
			\path[succ] (e12.south) -- (e13.north);

			\node[event] (e14) at (0*\xsep,-1*\ysep-4*\yseps-2*\ysepsd) {14};
			\node[lbl,anchor=west] (l14) at (e14.east) {\txtae};
			\path[succ] (e10.south) to[bend right=8] (e14.north);
			\path[succ] (e13.south) -- (e14.north);

			\node[event] (e15) at (0*\xsep,-1*\ysep-5*\yseps-2*\ysepsd) {15};
			\node[lbl,anchor=west] (l15) at (e15.east) {\txtaf};
			\path[succ] (e14.south) -- (e15.north);

			\node[subcap] at (1.45cm,\subcapy) {\begin{minipage}[t][0cm]{1cm}\subcaption{}\label{fig:3:c}\end{minipage}};
		\end{scope}
		\begin{scope}[shift={(2*\xsep + 4*\xspace,0)}]
			\node[event] (e1) at (0*\xsep,-0*\ysep-0*\yseps) {1};
			\node[lbl,anchor=north] (l1) at (e1.south) {\txtaa};

			\path let \p1 = (e1.north) in
				node[lbl,anchor=north,inner ysep=0] (l5) at (1*\xsep,\y1) {\txtba};
			\node[event,anchor=north] (e5) at ($(l5.south)+(0,-0.025)$) {5};

			\node[event] (e16) at (1*\xsep,-1*\ysep-0*\yseps) {16};
			\node[lbl,anchor=east] (l16) at (e16.west) {\txtbb};
			\path[succ] (e5.south) to (e16.north);

			\node[event] (e17) at (1*\xsep,-1*\ysep-1*\yseps) {17};
			\node[lbl,anchor=east] (l17) at (e17.west) {\txtbcf};
			\path[succ] (e16.south) -- (e17.north);

			\node[event] (e18) at (1*\xsep,-1*\ysep-2*\yseps) {18};
			\node[lbl,anchor=east] (l18) at (e18.west) {\txtbd};
			\path[succ] (e17.south) -- (e18.north);

			\node[event] (e19) at (0*\xsep,-1*\ysep-2*\yseps-1*\ysepsd) {19};
			\node[lbl,anchor=west] (l19) at (e19.east) {\txtab};
			\path[succ] ($(l1.south)+(0,\arrowstub)$) to[bend right=8] (e19.north);
			\path[succ] (e18.south) -- (e19.north);

			\node[event] (e20) at (0*\xsep,-1*\ysep-3*\yseps-1*\ysepsd) {20};
			\node[lbl,anchor=west] (l20) at (e20.east) {\txtacf};
			\path[succ] (e19.south) -- (e20.north);

			\node[event] (e21) at (0*\xsep,-1*\ysep-4*\yseps-1*\ysepsd) {21};
			\node[lbl,anchor=west] (l21) at (e21.east) {\txtaf};
			\path[succ] (e20.south) -- (e21.north);

			\node[subcap] at (0.5*\xsep,\subcapy) {\begin{minipage}[t][0cm]{1cm}\subcaption{}\label{fig:3:d}\end{minipage}};
		\end{scope}
		\begin{scope}[shift={(3*\xsep + 6*\xspace,0)}]
			\node[event] (e1) at (0*\xsep,-0*\ysep-0*\yseps) {1};
			\node[lbl,anchor=north] (l1) at (e1.south) {\txtaa};

			\path let \p1 = (e1.north) in
				node[lbl,anchor=north,inner sep=0] (l5) at (1*\xsep,\y1) {\txtba};
			\node[event,anchor=north] (e5) at ($(l5.south)+(0,-0.025)$) {5};

			\node[event] (e16) at (1*\xsep,-1*\ysep-0*\yseps) {16};
			\node[lbl,anchor=east] (l16) at (e16.west) {\txtbb};
			\path[succ] (e5.south) to (e16.north);

			\node[event] (e17) at (1*\xsep,-1*\ysep-1*\yseps) {17};
			\node[lbl,anchor=east] (l17) at (e17.west) {\txtbcf};
			\path[succ] (e16.south) -- (e17.north);

			\node[event] (e18) at (1*\xsep,-1*\ysep-2*\yseps) {18};
			\node[lbl,anchor=east] (l18) at (e18.west) {\txtbd};
			\path[succ] (e17.south) -- (e18.north);

			\node[event] (e19) at (0*\xsep,-1*\ysep-2*\yseps-1*\ysepsd) {19};
			\node[lbl,anchor=west] (l19) at (e19.east) {\txtab};
			\path[succ] ($(l1.south)+(0,\arrowstub)$) to[bend right=8] (e19.north);
			\path[succ] (e18.south) -- (e19.north);

			\node[event] (e22) at (0*\xsep,-1*\ysep-3*\yseps-1*\ysepsd) {22};
			\node[lbl,anchor=west] (l22) at (e22.east) {\txtact};
			\path[succ] (e19.south) -- (e22.north);

			\node[event] (e23) at (0*\xsep,-1*\ysep-4*\yseps-1*\ysepsd) {23};
			\node[lbl,anchor=west] (l23) at (e23.east) {\txtad};
			\path[succ] (e22.south) -- (e23.north);

			\node[deadlock] at (0.5*\xsep,-1*\ysep-4*\yseps-1*\ysepsd-0.5cm) {(deadlock!)};

			\node[subcap] at (0.5*\xsep,\subcapy) {\begin{minipage}[t][0cm]{1cm}\subcaption{}\label{fig:3:e}\end{minipage}};
		\end{scope}
	\end{scope}

\end{tikzpicture} \caption{A program and its four partial-order runs.}\label{fig:3}
\end{figure*}
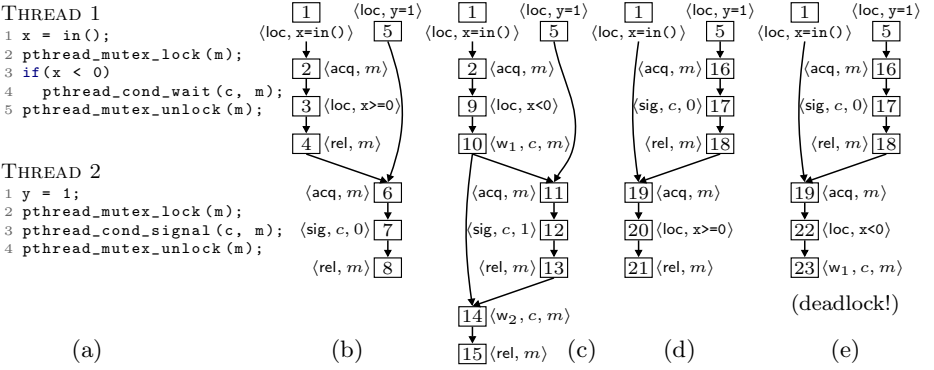

\subsection{Independence}

In the previous section, given an action $a \in A$ we informally defined the
set of actions which are \emph{dependent} on~$a$, therefore indirectly defining
an \emph{independence relation}.
We now show that this relation is a \emph{valid independence}~\cite{FG05,RSSK15}.
Intuitively, an independence relation is
\emph{valid} when every pair of actions it declares as independent
can be executed in any order while still producing the same state.

Our independence relation is valid only for \emph{data-race-free} programs.
We say that $P$ is \emph{data-race-free} iff any two local actions
$a \eqdef \tup{i, \tup{\loc, t}}$ and
$a' \eqdef \tup{i', \tup{\loc, t'}}$
from different threads ($i \ne i'$) commute at every
reachable state of~$P$.
See~\aref{app:indep} for additional details.
This ensures that local statements of different threads of~$P$ modify the memory
without interfering each other.

\begin{theorem}
If $P$ is data-race-free, then the independence relation defined in
\cref{sec:programs.actions} is valid.
\end{theorem}
\begin{proof}
See~\aref{app:indep.programs}.
\end{proof}

Our technique does not use data races as a source of thread interference for
partial-order reduction.
It will not explore two execution orders for the two statements that exhibit a
data race. However, it can be used to detect and report data races found
during the POR exploration, as we will see in \cref{sec:data.race}.

\subsection{Partial-Order Runs}

A \emph{labeled partial-order} (LPO) is a tuple
\(\tup{X, {<}, h}\) where \(X\) is a set of \emph{events},
\({<} \subseteq X \times X\) is a \emph{causality} (\aka, \emph{happens-before})
relation,
and \(h \colon X \to A\) labels each event by an \emph{action} in~\(A\).

A \emph{partial-order run} of~$P$ is an LPO that represents a
run of~$P$ without enforcing an order of execution on
actions that are independent.
All partial-order runs of \cref{fig:3:a} are shown in
\cref{fig:3:b,fig:3:c,fig:3:d,fig:3:e}.

Given a run $\sigma$ of~$P$,
we obtain the corresponding partial-order run
$\po_\sigma \eqdef \tup{E, {<}, h}$
by the following procedure:
(1) initialize $\po_\sigma$ to be the only totally-ordered LPO that consists
of $|\sigma|$ events where the i-th event is labeled by the i-th action of~$\sigma$;
(2) for every two events $e, e'$ such that $e < e'$, remove the pair
$\tup{e,e'}$ from~$<$ if $h(e)$ is independent from~$h(e')$;
(3) restore transitivity in~$<$ (\ie, if $e < e'$ and $e' < e''$, then add
$\tup{e, e''}$ to $<$).
The resulting LPO is a partial-order run of~$P$.

Furthermore, the originating run~$\sigma$ is an \emph{interleaving} of
$\po_\sigma$.
Given some LPO~$\po \eqdef \tup{E, {<}, h}$, an interleaving of~$\po$
is the sequence that labels any topological ordering of~$\po$.
Formally, it is any sequence $h(e_1), \ldots, h(e_n)$ such that
$E = \set{e_1, \ldots, e_n}$ and $e_i < e_j \implies i < j$.
We let~$\inter \po$ denote the set of all interleavings of~$\po$.
Given a partial-order run~$\po$ of~$P$, the interleavings~$\inter\po$ have two
important properties:
every interleaving in~$\inter\po$ is a run of~$P$, and
any two interleavings $\sigma, \sigma' \in \inter \po$ reach the same
state $\state \sigma = \state{\sigma'}$.

\subsection{Prime Event Structures}\label{sec:algo.pes}

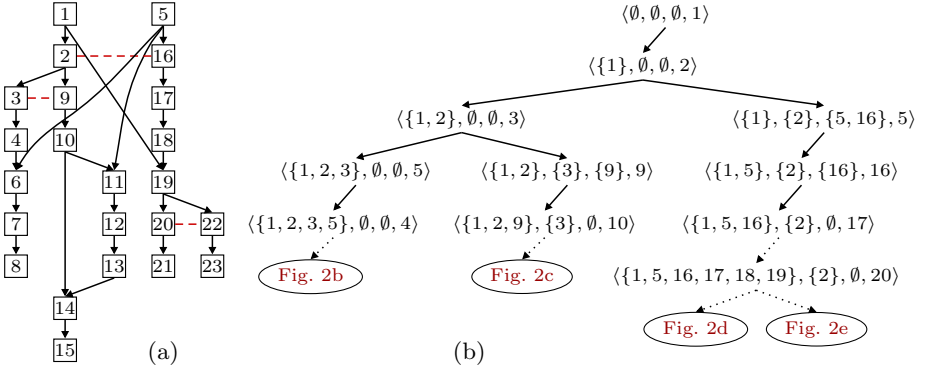
\begin{figure}[t]
\centering
\begin{tikzpicture}[
	font=\fontsize{7}{6.25}\selectfont{},
	event/.style={
		draw,
		line width=0.015cm,
		text width=0.3cm,
		align=center,
		inner xsep=0,
		inner ysep=0.07cm,
	},
	lbl/.style={
		inner xsep=0,
		inner ysep=0.05cm,
	},
	succ/.style={
		draw,
		->,
		>={Triangle[length=0.1cm,width=0.1cm]},
		line width=0.02cm,
	},
	foll/.style={
		draw,
		->,
		>={Triangle[length=0.1cm,width=0.1cm]},
		line width=0.02cm,
		dotted,
	},
	conflict/.style={
		draw,
		-,
		color=red!80!black,
		>={Triangle[length=0.1cm,width=0.1cm]},
		densely dashed,
		line width=0.02cm,
	},
	config/.style={
		draw,
		ellipse,
		inner sep=0.05cm,
	},
	subcap/.style={
		anchor=north,
		inner sep=0,
},
]
	\def\xsepb{0.65cm}
	\def\ysepb{0.56cm}
	\def\subcapy{-4.30cm}

	\begin{scope}[shift={(-8cm,0)}]
		\node[event] (e1) at (1*\xsepb,-0*\ysepb) {1};
		
		\node[event] (e2) at (1*\xsepb,-1*\ysepb) {2};
		\path[succ] (e1.south) -- (e2.north);
		
		\node[event] (e3) at (0*\xsepb,-2*\ysepb) {3};
		\path[succ] (e2.south) -- (e3.north);
		
		\node[event] (e4) at (0*\xsepb,-3*\ysepb) {4};
		\path[succ] (e3.south) -- (e4.north);

		\node[event] (e5) at (3*\xsepb,-0*\ysepb) {5};
		
		\node[event] (e6) at (0*\xsepb,-4*\ysepb) {6};
		\path[succ] (e4.south) -- (e6.north);
		
		\node[event] (e7) at (0*\xsepb,-5*\ysepb) {7};
		\path[succ] (e6.south) -- (e7.north);
		
		\node[event] (e8) at (0*\xsepb,-6*\ysepb) {8};
		\path[succ] (e7.south) -- (e8.north);
		
		\node[event] (e9) at (1*\xsepb,-2*\ysepb) {9};
		\path[succ] (e2.south) -- (e9.north);
		\path[conflict] (e3.east) -- (e9.west);
		
		\node[event] (e10) at (1*\xsepb,-3*\ysepb) {10};
		\path[succ] (e9.south) -- (e10.north);
		
		\node[event] (e11) at (2*\xsepb,-4*\ysepb) {11};
\path[succ] (e5.south) to[out=240,in=60] ++(-0.175cm,-0.25cm) to[out=240,in=80] (e11.north);
\path[succ] (e10.south) -- (e11.north);
		
		\node[event] (e12) at (2*\xsepb,-5*\ysepb) {12};
		\path[succ] (e11.south) -- (e12.north);
		
		\node[event] (e13) at (2*\xsepb,-6*\ysepb) {13};
		\path[succ] (e12.south) -- (e13.north);
		
		\node[event] (e14) at (1*\xsepb,-7*\ysepb) {14};
		\path[succ] (e10.south) -- (e14.north);
		\path[succ] (e13.south) -- (e14.north);
		
		\node[event] (e15) at (1*\xsepb,-8*\ysepb) {15};
		\path[succ] (e14.south) -- (e15.north);
		
		\node[event] (e16) at (3*\xsepb,-1*\ysepb) {16};
		\path[succ] (e5.south) -- (e16.north);
		\path[conflict] (e2.east) -- (e16.west);

		\node[event] (e17) at (3*\xsepb,-2*\ysepb) {17};
		\path[succ] (e16.south) -- (e17.north);
		
		\node[event] (e18) at (3*\xsepb,-3*\ysepb) {18};
		\path[succ] (e17.south) -- (e18.north);
		
		\node[event] (e19) at (3*\xsepb,-4*\ysepb) {19};
		\path[succ] (e1.south) -- (e19.north);
		\path[succ] (e18.south) -- (e19.north);
		
		\node[event] (e20) at (3*\xsepb,-5*\ysepb) {20};
		\path[succ] (e19.south) -- (e20.north);
		
		\node[event] (e21) at (3*\xsepb,-6*\ysepb) {21};
		\path[succ] (e20.south) -- (e21.north);
		
		\node[event] (e22) at (4*\xsepb,-5*\ysepb) {22};
		\path[succ] (e19.south) -- (e22.north);
		\path[conflict] (e20.east) -- (e22.west);
		
		\node[event] (e23) at (4*\xsepb,-6*\ysepb) {23};
		\path[succ] (e22.south) -- (e23.north);

		\path[succ] let \p1 = (e9.south), \p2 = (e10.north) in
			(e5.south) to[out=225,in=35] (\x1,\y1/2+\y2/2) to[out=215,in=60] (e6.north);

		\node[subcap] at (3*\xsepb,\subcapy) {\begin{minipage}[t][0cm]{1cm}\subcaption{}\label{fig:algo.unf.a}\end{minipage}};
	\end{scope}

   \def\xincr{0.3cm}
   \def\yincr{0.7cm}
   \def\incr{\xincr,\yincr}

	\begin{scope}[shift={(0.6cm,0cm)}]
		\node[lbl,inner ysep=0] (q1) at (0cm,0cm) {\(\tup{\emptyset, \emptyset, \emptyset, 1}\)};

		\node[lbl] (q2) at ($(q1)-(\incr)$) {\(\tup{\set{1}, \emptyset, \emptyset, 2}\)};
		\path[succ] ($(q1.south)+(0,-0.05cm)$) to (q2.north);

		\node[lbl] (q3) at ($(q2)+(-2.4cm,-\yincr)$) {\(\tup{\set{1,2}, \emptyset, \emptyset, 3}\)};
		\path[succ] (q2.south) to (q3.north);

		\node[lbl] (q4) at ($(q3)-(1.4cm,\yincr)$) {\(\tup{\set{1,2,3}, \emptyset, \emptyset, 5}\)};
		\path[succ] (q3.south) to (q4.north);

		\node[lbl] (q5) at ($(q4)-(\incr)$) {\(\tup{\set{1,2,3,5}, \emptyset, \emptyset, 4}\)};
		\path[succ] (q4.south) to (q5.north);

		\node[lbl] (q6) at ($(q3)+(1.4cm,-\yincr)$) {\(\tup{\set{1,2}, \set{3}, \set{9}, 9}\)};
		\path[succ] (q3.south) to (q6.north);

		\node[lbl] (q7) at ($(q6)-(\incr)$) {\(\tup{\set{1,2,9}, \set{3}, \emptyset, 10}\)};
		\path[succ] (q6.south) to (q7.north);

		\node[lbl] (q8) at ($(q2)+(2.4cm,-\yincr)$) {\(\tup{\set{1}, \set{2}, \set{5, 16}, 5}\)};
		\path[succ] (q2.south) to (q8.north);

		\node[lbl] (q9) at ($(q8)-(\incr)$) {\(\tup{\set{1,5}, \set{2}, \set{16}, 16}\)};
		\path[succ] (q8.south) to (q9.north);

		\node[lbl] (q10) at ($(q9)-(\incr)$) {\(\tup{\set{1,5,16}, \set{2}, \emptyset, 17}\)};
		\path[succ] (q9.south) to (q10.north);

		\node[lbl] (q11) at ($(q10)-(\incr)$) {\(\tup{\set{1,5,16,17,18,19}, \set{2}, \emptyset, 20}\)};
		\path[foll] (q10.south) to (q11.north);

		\node[config] (cb) at ($(q5)-(\incr)$) {\cref{fig:3:b}};
		\path[foll] (q5.south) to (cb.north);

		\node[config] (cc) at ($(q7)-(\incr)$) {\cref{fig:3:c}};
		\path[foll] (q7.south) to (cc.north);

		\node[config] (cd) at ($(q11)+(-.8cm,-\yincr)$) {\cref{fig:3:d}};
		\path[foll] (q11.south) to (cd.north);

		\node[config] (ce) at ($(q11)+(.8cm,-\yincr)$) {\cref{fig:3:e}};
		\path[foll] (q11.south) to (ce.north);

		\node[subcap] at (-2.6cm,\subcapy) {\begin{minipage}[t][0cm]{1cm}\subcaption{}\label{fig:algo.unf.b}\end{minipage}};
	\end{scope}
\end{tikzpicture} \caption{(a): unfolding of the program in \cref{fig:3:a}; (b): its POR exploration tree.}\label{fig:algo.unf}
\end{figure}

We use unfoldings to give semantics to multi-threaded programs.
Unfoldings are Prime Event Structures~\cite{NPW81},
tree-like representations of system behavior that use partial orders
to represent concurrent interaction.

\cref{fig:algo.unf.a} depicts an unfolding of the program in~\cref{fig:3:a}.
The nodes are events and solid arrows
represent causal dependencies: events~$1$ and $4$ must fire before $8$
can fire.
The dotted line represents conflicts: $2$ and $5$
are not in conflict and may occur in any order, but $2$ and $16$ are in
conflict and cannot occur in the same (partial-order) run.

Formally, a \emph{Prime Event Structure}~\cite{NPW81}
(PES)
is a tuple $\les \eqdef \tup{E, {<}, {\cfl}, h}$ with
a set of events $E$,
a causality relation ${<} \subseteq E \times E$, which is a strict partial order,
a conflict relation ${\cfl} \subseteq E \times E$ that is symmetric
and irreflexive, and a labeling function $h \colon E \to A$.

The \emph{causes} of an event
$\causes e \eqdef \set{e' \in E \colon e' < e}$
are the least set of events that must fire before $e$ can fire.
A \emph{configuration} of $\les$ is a finite set~$C \subseteq E$
that is 
causally closed
  ($\causes e \subseteq C$ for all $e \in C$), and
conflict-free
  ($\lnot (e \cfl e')$  for all $e, e' \in C$).
We let $\conf \les$ denote the set of all configurations of~$\les$.
For any $e \in E$, the \emph{local configuration} of~$e$ is defined as
$[e] \eqdef \causes e \cup \set e$.
In \cref{fig:algo.unf.a}, the set $\set{1,2}$ is a configuration, and in fact it is
a local configuration, \ie, $[2] = \set{1,2}$.
The local configuration of event~6 is $\set{1,2,3,4,5,6}$.
Set $\set{2,5,16}$ is not a configuration, because it is neither causally
closed (1 is missing) nor conflict-free ($2 \cfl 16$).

\subsection{Unfolding Semantics for Programs}\label{sec:algo.unfsem}

Given a program~$P$, in this section we define a PES~$\unf P$ such that
every configuration of~$\unf P$ is a partial-order run of~$P$.

Let
$\po_1 \eqdef \tup{E_1, {<}_1, h_1}, \ldots,
\po_n \eqdef \tup{E_n, {<}_n, h_n}$
be the collection of all the partial-order runs of~$P$.
The events of~$\unf P$ are the equivalence classes of the structural equality
relation that we intuitively described in \cref{sec:overview.unfolding}.

Two events are structurally equal iff their \emph{canonical name} is the same.
Given some event $e \in E_i$ in some partial-order run~$\po_i$,
the canonical name $\cn e$ of~$e$ is the pair~$\tup{a, H}$ where
$a \eqdef h_i(e)$ is the executed action and
$H \eqdef \set{\cn{e'} \colon e' <_i e}$
is the set of canonical names of those events that causally precede~$e$
in~$\po_i$.
Intuitively, canonical names indicate that action $h(e)$ runs after the
(transitively canonicalized) partially-ordered history preceding~$e$.
For instance, in \cref{fig:algo.unf.a}
for events~1 and~6 we have
$\cn 1 = \tup{\tup{1, \tup{\loc, \text{\texttt{a=in()}}}}, \emptyset}$,
and
$\cn 6 = \tup{\tup{2, \tup{\acq, m}}, \set{\cn 1, \cn 2, \cn3, \cn 4, \cn 5}}$.
Actually, the number within every event in
\cref{fig:3:b,fig:3:c,fig:3:d,fig:3:e}
identifies (is in bijective correspondence with) its canonical name.
Event~19 in \cref{fig:3:d} is the same event as event~19 in \cref{fig:3:e}
because it fires the same action
($\tup{1, \tup{\acq, m}}$) after the same causal history
($\set{1,5,16,17,18}$).
Event~2 in \cref{fig:3:c} and~19 in \cref{fig:3:d}
are \emph{not} the same event because while
$h(2) = h(19) = \tup{1, \tup{\acq, m}}$
they have a different causal history ($\set{1}$ \vs{} $\set{1,5,16,17,18}$).
Obviously events~4 and~6 in \cref{fig:3:b} are different because~$h(4) \ne h(6)$.
We can now define the \emph{unfolding} of~$P$ as the only PES
$\unf P \eqdef \tup{E, {<}, \cfl, h}$ such that
\begin{itemize}
\item
   $E \eqdef \set{\cn e \colon e \in E_1 \cup \ldots \cup E_n}$
   is the set of canonical names of all events;
\item
   Relation ${<} \subseteq E \times E$ is the union
   ${<_1} \cup \ldots \cup {<_n}$ of all happens-before relations;
\item
   Any two events $e, e' \in E$ of $\unf P$
   are in conflict, $e \cfl e'$, when
   $e \ne e'$, and
   $\lnot (e < e')$, and
   $\lnot (e' < e)$, and
   $h(e)$ is dependent on $h(e')$.
\end{itemize}

\cref{fig:algo.unf.a} shows the unfolding produced by merging all~4 partial-order runs
in~\cref{fig:3:b,fig:3:c,fig:3:d,fig:3:e}.
Note that the configurations of~$\unf P$ are partial-order runs of~$P$.
Furthermore, the $\subseteq$-maximal configurations are exactly the~4
originating partial orders.
It is possible to prove that
$\unf P$ is a semantics of~$P$.
In~\aref{app:unfolding} we show that
(1) $\unf P$ is uniquely defined,
(2) any interleaving of any local configuration of~$\unf P$ is a run of~$P$,
(3) for any run~$\sigma$ of~$P$ there is a configuration~$C$
of~$\unf P$ such that $\sigma \in \inter C$.

\subsection{Conflicting Extensions}\label{sec:algo.cex}

Our technique analyzes~$P$ by iteratively constructing (all) partial-order runs
of~$P$.
In every iteration we need to find the next partial order to explore.
We use the so-called \emph{conflicting extensions} of a configuration to detect how
to start a new partial-order run that has not been explored before.

Given a configuration~$C$ of~$\unf P$, an \emph{extension} of $C$
is any event~$e \in E \setminus C$ such that all the causal predecessors of~$e$
are in $C$.
We denote the set of extensions of $C$ as
$\ex C \eqdef \set{e \in E \colon e \notin C \land \causes e \subseteq C}$.
The \emph{enabled} events of~$C$ are extensions that can form a
larger configuration:
$\en C \eqdef \set{e \in \ex C \colon C \cup \set e \in \conf{\les}}$.
For instance, in \cref{fig:algo.unf.a},
the (local) configuration $[6]$ has 3 extensions,
$\ex{[6]} = \set{7,9,16}$
of which, however, only event 7 is enabled:
$\en{[6]} = \set{7}$.
Event~19 is not an extension of $[6]$ because~18 is
a causal predecessor of~19, but $18 \not\in [6]$.
A \emph{conflicting extension} of~$C$
is an extension for which there is at least one $e' \in C$
such that $e \cfl e'$.
The (local) configuration $[6]$ from our previous example
has two conflicting extensions, events~9 and~16.
A conflicting extension is, intuitively, an incompatible addition to the
configuration~$C$, an event~$e$ that cannot be executed together with~$C$
(without removing $e'$ and its causal successors from~$C$).
We denote by~$\cex C$ the set of all conflicting extensions of~$C$,
which coincides with the set of all extensions that are not enabled:
$\cex C \eqdef \ex C \setminus \en C$.

Our technique discovers new conflicting extension events by trying to revert the causal
order of certain events in~$C$.
Owing to space limitations we only explain how the algorithm handles events of
$\acq$ and~$\waa$ effect (\aref{app:cex} presents the remaining~4 procedures of
the algorithm).
\cref{a:cex.sec3} shows the procedure that handles this case.
It receives an event~$e$ of $\acq$ or~$\waa$ effect (\cref{l:cex.sec3.assume}).
We build and return a set of conflicting extensions, stored in
variable~$R$.
Events are added to~$R$ in~\cref{l:cex.sec3.add1,l:cex.sec3.add2}.
Note that we define events using their canonical name.
For instance, in \cref{l:cex.sec3.add1} we add
a new event whose action is~$h(e)$ and whose causal history is~$P$.
Note that we only create events that execute action~$h(e)$.
Conceptually speaking, the algorithm simply finds different
causal histories (variables~$P$ and~$e'$) within the set~$K = \causes e$ to
execute action~$h(e)$.

Procedure \lastof{$C,i$} returns the only $<$-maximal event of thread~$i$ in~$C$;
\lastnotify{$e, c, i$} returns the only immediate $<$-predecessor $e'$ of $e$ such that
the effect of $h(e')$ is either $\tup{\sigg,c,i}$ or $\tup{\bro,c,S}$ with $i
\in S$;
finally, procedure \lastlock{$C,l$} returns the only $<$-maximal event that manipulates
lock~$l$ in~$C$ (an event of effect $\acq$, $\rel$, $\wa$ or $\waa$), or $\bot$ if no such event exists.
See~\aref{app:cex} for additional details.

\begin{algorithm}[t]
\DontPrintSemicolon{}

\Fn{\cexacquire{e}}{
   Assume that $e$ is $\tup{\tup{i, \tup{\acq, l}}, K}$ or
   $\tup{\tup{i, \tup{\waa, c,l}}, K}$ \;
   \label{l:cex.sec3.assume}
   $R \eqdef \emptyset$ \;
   $e_t \eqdef \lastof{$K, i$}$ \;
   \eIf {$\effect e = \tup{\acq, l}$}
   {
      $P \eqdef [e_t]$
   }
   {
      $e_s \eqdef \lastnotify{$e, c, i$}$ \;
      $P \eqdef [e_t] \cup [e_s]$ \;
   }
   $e_m \eqdef \lastlock{$P, l$}$ \;
   $e_r \eqdef \lastlock{$K, l$}$ \;

   \lIf {$e_m = e_r$}{\KwRet $R$}

   \uIf {$e_m = \bot \lor
      \effect{e_m} \in \set{\tup{\rel, l}, \tup{\wa, \cdot, l}}$}
   {
      Add $\tup{h(e), P}$ to $R$ \;
      \label{l:cex.sec3.add1}
   }

   \ForEach {event $e' \in K \setminus (P \cup \set{e_r})$}
   {
      \uIf{$\effect{e'} \in \set{\tup{\rel, l}, \tup{\wa, \cdot, l}}$}
      {
         Add $\tup{h(e), P \cup [e']}$ to $R$ \;
         \label{l:cex.sec3.add2}
      }
   }
   \KwRet $R$ \;
}

\caption{Conflicting extensions for $\acq$/$\waa$ events.}\label{a:cex.sec3}
\end{algorithm}

\subsection{Exploring the Unfolding}\label{sec:algo.exploring}

\begin{algorithm}[t]
\DontPrintSemicolon

\BlankLine
Global variables: $U \eqdef \emptyset$ (set of events of~$\unf P$) and $N \eqdef \emptyset$ (set of tree nodes)
\BlankLine

\begin{multicols}{2}
\BlankLine
\Proc{\mainsimple{}}{
   \newnode{$\emptyset,\emptyset,\emptyset$} \;
   \Repeat {fixed point ($N$ is stable)}
   {
      \label{l:loop}
      Select $n \eqdef \tup{C, D, A, e}$ from $N$ \;
      \label{l:select}
      Add $\cex C$ to $U$ \;
      \label{l:cex}
      \If {$\encutoff{$C$} \subseteq D$}
      {
         \label{l:maximal}
         \KwContinue \;
      }
      \If{$n$ has no left child}
      {
         \label{l:left}
         $n' \eqdef \newnode{$C \cup \set e, D, A \setminus \set e$}$ \;
         \label{l:left.node}
         Make $n'$ the left child of~$n$ \;
      }
      \If{$n$ has no right child}
      {
         \label{l:right}
         $J \eqdef \alternatives{$C, D \cup \set e$}$ \;
         \label{l:alt.call}
         \If{$J \ne \emptyset$}
         {
            \label{l:check.alt}
            $n' \eqdef \newnode{$C, D \cup \set e, J \setminus C$}$ \;
            \label{l:right.node}
            Make $n'$ the right child of~$n$ \;
         }
      }
   }
}

\BlankLine
\Fn{\newnode{$C, D, A$}}{
   \If {$A \ne \emptyset$}
   {
      $e \eqdef \text{select from } \encutoff{$C$} \cap A$ \;
      \label{l:select.a}
   }
   \Else
   {
      $e \eqdef \text{select from } \encutoff{$C$} \setminus D$ \;
   }
   $n \eqdef \tup{C, D, A, e}$ \;
   Add $n$ to $N$ \;
   \KwRet $n$ \;
   
}

\BlankLine
\Fn{\encutoff{C}}{
\KwRet $\set{e \in \en C \colon \lnot \iscutoff{e}}$ \;
   \label{l:en.ret}
}

\BlankLine
\Fn{\alternatives{$C,D$}}{
Let $e$ be some event in $D \cap \en C$ \;
   \label{l:alt.select.e}
   $S \eqdef \set{e' \in U \colon
      e' \cfl e \land
      [e'] \cap D = \emptyset}$ \;
   $S \eqdef \set{e' \in S \colon
      [e'] \cup C \text{ is a config.}}$ \;
   \lIf{$S = \emptyset$} { \KwRet $\emptyset$ }
   Select some event $e'$ from $S$ \;
   \label{l:alt.found}
   \KwRet $[e']$ \;
}
\end{multicols}
\vspace*{6pt}

\caption{Main algorithm. See \cref{sec:algo.exploring}.}\label{a:main.algo}
\end{algorithm}

This section presents an algorithm that explores the state space of~$P$ by
constructing all maximal configurations of~$\unf P$.
In essence, our procedure is an improved Quasi-Optimal POR
algorithm~\cite{NRSCP18}, where the unfolding is not explored using a DFS
traversal, but a user-defined search order.
This enables us to build upon the preexisting exploration heuristics (``searchers'') in KLEE rather than having to follow a strict DFS exploration of the unfolding.

Our algorithm explores one configuration of~$\unf P$ at a time and organizes the
exploration into a binary tree.
\cref{fig:algo.unf.b} shows the tree explored for the unfolding shown
in~\cref{fig:algo.unf.a}.
A tree node is a tuple~$n \eqdef \tup{C,D,A,e}$ that
represents both the exploration of a configuration~$C$ of~$\unf P$ and a choice
to execute, or not, event $e \in \en C$.
Both~$D$ (for \emph{disabled}) and~$A$ (for \emph{add}) are sets of events.

The key insight of this tree is as follows.
The subtree rooted at a given node~$n$ explores all configurations of~$\unf P$ that
include~$C$ and exclude~$D$, with the following constraint:
$n$'s left subtree explores all configurations including event~$e$
and
$n$'s right subtree explores all configuration excluding~$e$.
Set~$A$ is used to guide the algorithm when exploring the right subtree.
For instance, in~\cref{fig:algo.unf.b} the subtree rooted at node
$n \eqdef \tup{\set{1,2},\emptyset,\emptyset,3}$
explores all maximal configurations that contain events~1 and~2 (namely, those
shown in~\cref{fig:3:b,fig:3:c}).
The left subtree of~$n$ explores all configurations
including~$\set{1,2,3}$
(\cref{fig:3:b})
and the right subtree all of those including~$\set{1,2}$ but excluding~3
(\cref{fig:3:c}).

\cref{a:main.algo} shows a simplified version of our algorithm.
The complete version, in~\aref{app:main}, specifies additional details including
how nodes are selected for exploration and how they are removed from the tree.
The algorithm constructs and stores the exploration tree in the variable~\(N\),
and the set of currently known events of~\(\unf N\) in variable~\(U\).
At the end of the
exploration, \(U\) will store all events of~$\unf N$ and the leaves of the
exploration tree in~$N$ will correspond to the maximal configurations of~$\unf N$.

The tree is constructed using a fixed-point loop (\cref{l:loop}) that repeats
the following steps as long as they modify the tree:
select a node~$\tup{C,D,A,e}$ in the tree (\cref{l:select}),
extend~$U$ with the conflicting extensions of~$C$ (\cref{l:cex}),
check if the configuration is $\subseteq$-maximal (\cref{l:maximal}), in which
case there is nothing left to do,
then try to add a left (\cref{l:left}) or right (\cref{l:right}) child node.

The subtree rooted at the left child node will explore all configurations that
include~$C \cup \set e$ and exclude~$D$ (\cref{l:left.node});
the right subtree will explore those including~$C$ and excluding~$D \cup \set e$
(\cref{l:right.node}), if any of them exists, which we detect by checking
(\cref{l:check.alt}) if we found a so-called \emph{alternative}~\cite{RSSK15}.

An alternative is a set of events which witnesses the existence of some maximal
configuration in~$\unf P$ that extends~$C$ without including~$D \cup \set e$.
Computing such witness is an NP-complete problem,
so we use an approximation called \emph{$k$-partial alternatives}~\cite{NRSCP18},
which can be computed in P-time and works well in practice.
Our procedure~\alternatives specifically computes $1$-partial alternatives:
it selects $k=1$ event~$e$ from~$D \cap \en C$, searches for an event~$e'$ in
conflict with~$e$
(we have added all known candidates in~\cref{l:cex}, using the algorithms
of~\cref{sec:algo.cex})
that can extend $C$ (\ie, such that $C \cup [e']$ is a configuration),
and returns it.
When such an event $e'$ is found (\cref{l:alt.found}), some events in its local
configuration $[e']$ become the $A$-component of the right child node
(\cref{l:right.node}), and the leftmost branch rooted at that node will
re-execute those events (as they will be selected in \cref{l:select.a}),
guiding the search towards the witnessed maximal configuration.

For instance, in \cref{fig:algo.unf.b}, assume that the algorithm has selected node
$n = \tup{\set 1, \emptyset, \emptyset, 2}$ at \cref{l:select} when event~16 is already
in~$U$. Then a call to \alternatives{$\set 1, \set 2$} is issued
at~\cref{l:alt.call}, event $e = 2$ is selected at \cref{l:alt.select.e}
and event $e' = 16$ gets selected at~\cref{l:alt.found},
because $2 \cfl 16$ and $[16] \cup \set 1$ is a configuration.
As a result, node $n' = \tup{\set 1, \set 2, \set{5,16}, 5}$ becomes the right
child of~$n$ in \cref{l:right.node}, and the leftmost branch rooted at~$n'$
adds~$\set{5,16}$ to~$C$, leading to the maximal configuration~\cref{fig:3:d}.

\subsection{Cutoffs and Completeness}\label{sec:algo.cutoffs}

All interleavings of a given configuration always reach
the same state, but interleavings of different configurations can also reach the
same state. It is possible to exclude certain such redundant configurations from the
exploration without making the algorithm incomplete, by using \emph{cutoff}
events~\cite{Mcm93}.

Intuitively, an event is a cutoff if we have already visited another event
that reaches the same state with a shorter execution.
Formally, in~\cref{a:main.algo}, \cref{l:en.ret} we let
\iscutoff{$e$} return \emph{true} iff
there is some~$e' \in U$ such that
$\state{[e]} = \state{[e']}$
and
$|[e']| < |[e]|$.
This makes \cref{a:main.algo} ignore cutoff events and any event that
causally succeeds them.
\cref{sec:exp:cutoff} explains how to effectively implement the check
$\state{[e]} = \state{[e']}$.

While cutoffs prevent the exploration of redundant configurations, the
analysis is still complete: it is possible to prove that every state reachable
via a configuration with cutoffs is also reachable via a configuration without
cutoffs. Furthermore, cutoff events not only reduce the exploration of
redundant configurations, but also force the algorithm to terminate for
non-terminating programs that run on bounded memory.

\begin{theorem}[Correctness]\label{t:algo.correctness}
For any reachable state $s \in \reach P$, \cref{a:main.algo} explores a
configuration~$C$ such that for some~$C' \subseteq C$ it holds that
$\state{C'} = s$.
Furthermore, it terminates for any program~$P$ such that~$\reach P$ is finite.
\end{theorem}

A proof sketch is available in~\aref{app:main}.
Naturally, since \cref{a:main.algo} explores~$\unf P$, and~$\unf P$ is an exact
representation of all runs of~$P$, then \cref{a:main.algo} is also \emph{sound}:
any event constructed by the algorithm (added to set~$U$) is associated
with a real run of~$P$.

 \section{Implementation}\label{sec:impl}

We implemented our approach on top of the symbolic execution engine KLEE~\cite{KLEE}, which was previously restricted to sequential programs.
KLEE already provides a minimal POSIX support library that we extended to
translate calls to pthread functions to their respective actions, enabling us to test real-world multi-threaded C programs.
We also extended already available functionality to make it thread-safe, \eg, by implementing a global file system lock that ensures that concurrent reads from the same file descriptor do not result in unsafe behavior.
The source code of our prototype is available at \url{https://github.com/por-se/por-se}.

\subsection{Standby States}\label{sec:exp:standby}
When a new alternative is explored, a symbolic execution state needs to be computed to match the new node in the POR tree.
However, creating it from scratch requires too much time and keeping a symbolic execution state around for each node consumes significant amounts of memory.
Instead of committing to either extreme, we store \emph{standby states} at regular intervals along the exploration tree and, when necessary, replay the closest standby state.
This way, significantly fewer states are kept in memory without letting the replaying of previously computed operations dominate the analysis either.

\subsection{Hash-Based Cutoff Events}\label{sec:exp:cutoff}
Schemmel~et~al.\@ presented~\cite{SymLive} an incremental hashing scheme to
identify infinite loops during symbolic execution.
The approach detects when the program under test can
transition from any one state back to that same state.
Their scheme computes \emph{fragments} for small portions of the program state, which are then hashed individually, and combined into a compound hash by bitwise xor operations.
This compound hash, called a \emph{fingerprint}, uniquely (modulo hash collisions) identifies the whole state of the program under test.
We adapt this scheme to provide hashes that identify the concurrent state of parallel programs.

To this end, we associate each configuration with a fingerprint that describes the whole state of the program at that point.
For example, if the program state consists of two variables, \(x=3\) and \(y=5\), the fingerprint would be \(fp=\text{hash}\left(\text{\texttt{"x=3"}}\right)\oplus\text{hash}\left(\text{\texttt{"y=5"}}\right)\).
When one fragment changes, \eg, from \(x=3\) to \(x=4\), the old fragment hash needs to be replaced with the new one.
This operation can be performed as \(fp'=fp\oplus\text{hash}\left(\text{\texttt{"x=3"}}\right)\oplus\text{hash}\left(\text{\texttt{"x=4"}}\right)\) as the duplicate fragments for \(x=3\) will cancel out.
To quickly compute the fingerprint of a configuration, we annotate each event with an xor of all of these update operations that were done on its thread.
Computing the fingerprint of a configuration now only requires xor-ing the values from its thread-maximal events, which will ensure that all changes done to each variable are accounted for, and cancel out one another so that only the fragment for the last value remains.

Any two local configurations that have the same fingerprint represent the same program state; each variable, program counter, \etc, has the same value.
Thus, it is not necessary to continue exploring both---we have found a potential cutoff point, which the POR algorithm will treat accordingly (\cref{sec:algo.cutoffs}).

A more in-depth look at this delta-based memory fingerprinting scheme and its efficient computation can be found in~\cref{app:fingerprint}.

\subsection{Deterministic and Repeatable Allocations}
KLEE usually uses the system allocator to determine the addresses of objects
allocated by the program under test.
But it also provides a (more) deterministic
mode, in which addresses are consumed in sequence from a large pre-allocated
array.
Since our hash-based cutoff computation uses memory address as part of the
computation, using execution replays from standby
states~(\cref{sec:exp:standby}) requires that we have
fully repeatable memory allocation.

We tackle this problem by decoupling the addresses returned by the emulated system allocator in the program under test from the system allocator of KLEE itself.
A new allocator requires a large amount of virtual memory in which it will perform its allocations.
This large virtual memory mapping is not actually used unless an external function call is performed, in which case the relevant objects are temporarily copied into the region from the symbolic execution state for which the external function call is to be performed.
Afterwards, the pages are marked for reclamation by the OS\@.
This way, allocations done by different symbolic execution states return the same address to the program under test.

While a deterministic allocator by itself would be enough for providing deterministic allocation to sequential programs, parallel programs also require an allocation pattern that is independent of which sequentialization of the same partial order is chosen.
We achieve this property by providing independent allocators for each thread (based on the thread id, thus ensuring that the same virtual memory mapping is reused for each instance of the same semantic thread).
When an object is deallocated on a different thread than it was allocated on, its address only becomes available for reuse once the allocating thread has reached a point in its execution where it is causally dependent on the deallocation.
Additionally, the thread ids that are used by our implementation are hierarchically defined:
A new thread \(t\) that is the \(i\)-th thread started by its parent thread \(p\) has the thread id \(t := \left(p,i\right)\), with the main thread being denoted as \(\left(1\right)\).
This way, thread ids and the associated virtual memory mappings are independent of how the concurrent creation of multiple threads are sequentialized.

We have also included various optimizations that promote controlled reuse of addresses to increase the chance that a cutoff event~(\cref{sec:exp:cutoff}) is found, such as binning allocations by size, which reduces the chance that temporary allocations impact which addresses are returned for other allocations.

\subsection{Data Race Detection}\label{sec:data.race}
Our data race detection algorithm simply follows the happens-before relationships established by the POR.\@
However, its implementation is complicated by the possibility of addresses becoming symbolic.
Generally speaking, a symbolic address can potentially point to any and every byte in the whole address space, thus requiring frequent and large SMT queries to be solved.

To alleviate the quadratic blowup of possibly aliasing accesses, we exploit how KLEE performs memory accesses with symbolic addresses: The symbolic state is forked for every possible memory object that the access may refer to (and one additional time if the memory access may point to unallocated memory).
Therefore, a symbolic memory access is already resolved to memory object granularity when it potentially participates in a data race.
This drastically reduces the amount of possible data races without querying the SMT solver.

\subsection{External Function Calls}
When a program wants to call a function that is neither provided by the program itself nor by the runtime, KLEE will attempt to perform an \emph{external function call} by moving the function arguments from the symbolic state to its own address space and attempting to call the function itself.
While this support for uninterpreted functions is helpful for getting some results for programs which are not fully supported by KLEE's POSIX runtime, it is also inherently incomplete and not sound in the general case.
Our prototype includes this option as well.
 \section{Experimental Evaluation}\label{sec:exp}
To explore the efficacy of the presented approach, we performed a series of
experiments including both synthetic benchmarks from the SV-COMP~\cite{svcomp2019} benchmark suite and real-world programs, namely, Memcached~\cite{memcached} and GNU sort~\cite{gnusort}.
We compare against Yogar-CBMC~\cite{YDLLW18}, which is the winner of the concurrency safety category of SV-COMP~2019~\cite{svcomp2019}, and stands in for the family of bounded model checkers.
As such, Yogar-CBMC is predestined to fare well in the artificial SV-COMP benchmarks, while our approach may demonstrate its strength in dealing with more complicated programs.

We ran the experiments on a cluster of multiple identical machines with dual Intel Xeon E5-2643 v4 CPUs and \SI{256}{\gibi\byte} of RAM.\@
We used a~\SI{4}{\hour} timeout and~\SI{200}{\giga\byte} maximum memory usage for real-world programs.
We used a~\SI{15}{\minute} timeout and~\SI{15}{\giga\byte} maximum memory for individual SV-COMP benchmarks.

\subsection{SV-COMP}
\begin{table}[t]
	\begin{center}\ttfamily \scriptsize \setlength{\tabcolsep}{2.5pt}\newcommand\cmidrules{\cmidrule(l{2.0pt}r{2.0pt}){2-6}\cmidrule(l{2.0pt}r{2.0pt}){7-11}}\begin{tabular}{lrrrrrrrrrr}\toprule \multicolumn{1}{c}{\textrm{Benchmark}} & \multicolumn{5}{c}{\textrm{Our Tool}} & \multicolumn{5}{c}{\textrm{Yogar-CBMC}} \\[-0.5pt]\cmidrules{}
	& \textrm{T} & \textrm{F} & \textrm{U} & \textrm{Time} & \textrm{RSS} & \textrm{T} & \textrm{F} & \textrm{U} & \textrm{Time} & \textrm{RSS} \\[-1.0pt]\midrule \textrm{pthread} & 29 & - & 9 & 1:50:19 & 16GB & 29 & - & 9 & 0:31:21 & 948MB \\
	\textrm{pthread-driver-races} & 16 & 1 & 4 & 1:03:08 & 6049MB & 21 & - & - & 0:00:12 & 72MB \\
	\bottomrule \end{tabular}\end{center}
 	\caption{
		Our prototype and Yogar-CBMC running SV-COMP benchmarks.\\
		\footnotesize Timeout set at \SI{15}{\min} with maximum memory usage of
      \SI{15}{\giga\byte}. Columns are: T:\@ true result, output matches expected
      verdict; F:\@ false result, output does not match expected verdict; U:\@
      unknown result, tool yields no answer; Time: total time taken; RSS:\@
      maximum resident set size over all benchmarks.
	}\label{fig:svcomp}
\vspace{-0.25cm}
\end{table}

We ran our tool and Yogar-CBMC on the ``pthread'' and ``pthread-driver-races'' benchmark suites in their newest (2020) incarnation.
As expected, \cref{fig:svcomp} shows that Yogar-CBMC clearly outperforms our tool for this specific set of benchmarks.
Not only does Yogar-CBMC not miscategorize even a single benchmark, it does so quickly and without using a lot of memory.
Our tool, in contrast, takes significantly more time and memory to analyze the target benchmarks.
In fact, several benchmarks do not complete within the \SI{15}{\min} time frame
and therefore cannot give a verdict for those.

The ``pthread-driver-races'' benchmark suite contains one benchmark that is marked as a failure for our tool in \cref{fig:svcomp}.
For the relevant benchmark, a verdict of ``target function unreachable'' is expected, which we translate to mean ``no data race occurs''.
However, the benchmark program constructs a pointer that may point to effectively any byte in memory, which, upon dereferencing it, leads to both, memory errors and data races (by virtue of the pointer also being able to touch another thread's stack).
While we report this behavior for completeness sake, we attribute it to the adaptations made to fit the SV-COMP model to ours.

\subsubsection{Preparation of Benchmark Suites.}
The SV-COMP benchmark suite does not only assume various kinds of special casing (\eg, functions whose name begins with \mbox{\texttt{\_\_VERIFIER\_atomic}} must be executed atomically), but also routinely violates the C standard by, for example, employing data races as a control flow mechanism~\cite[§5.1.2.4/35]{C18}.
Partially, this is because the analysis target is a question of reachability of
a certain part of the benchmark program, not its correctness.
We therefore attempted to guess the intention of the individual benchmarks,
making variables atomic or leaving the data race in when it is the aim of the
benchmark.

\subsection{Memcached}
Memcached~\cite{memcached} is an in-memory network object cache written in C.
As it is a somewhat large project with a fairly significant state space, we were unable to analyze it completely, even though our prototype still found several bugs.
Our attempts to run Yogar-CBMC did not succeed, as it reproducibly crashes.

\subsubsection{Faults detected.}
Our prototype found nine bugs in memcached 1.5.19, attributable to four different root causes, all of which where previously unknown.
The first bug is a misuse of the pthread API, causing six mutexes and condition
variables to be initialized twice, leading to undefined behavior.
We reported\footnote{
	\url{https://github.com/memcached/memcached/pull/566}
} the issue, a fix is included in version 1.5.20.
The second bug occurs during the initialization of memcached, where fields that will later be accessed in a thread-safe manner are sometimes accessed in a non-thread-safe manner, assuming that competing accesses are not yet possible.
We reported\footnote{
	\url{https://github.com/memcached/memcached/pull/575}
} a mistake our tool found in the initialization order that invalidates the assumption that locking is not (yet) necessary on one field.
A fix ships with memcached 1.5.21.
For the third bug, memcached utilizes a maintenance thread to manage and resize its core hash table when necessary.
Additionally, on another thread, a timer checks whether the maintenance thread should perform an expansion of the hash table.
We found\footnote{
	\url{https://github.com/memcached/memcached/pull/569}
} a data race between these two threads on a field that stores whether the maintenance thread has started expanding.
This is fixed in version 1.5.20.
The fourth and final issue is a data race on the \mbox{\texttt{stats\_state}} storing execution statistics.
We reported\footnote{
	\url{https://github.com/memcached/memcached/pull/573}
} this issue and a fix is included in version 1.5.21.

\begin{table}[t]
	\begin{center}\ttfamily \scriptsize \setlength{\tabcolsep}{2.25pt}\newcommand\cmidrules{\cmidrule(l{2.0pt}r{2.0pt}){1-3}\cmidrule(l{2.0pt}r{2.0pt}){4-6}\cmidrule(l{2.0pt}r{2.0pt}){7-7}\cmidrule(l{2.0pt}r{2.0pt}){8-12}\cmidrule(l{2.0pt}r{2.0pt}){13-15}\cmidrule(l{2.0pt}r{2.0pt}){16-16}}\begin{tabular}{rlrrrrrrrrrrrrrc}\toprule \multicolumn{3}{c}{\textrm{Program}} & \multicolumn{3}{c}{\textrm{Performance}} & \multicolumn{1}{c}{\textrm{Th}} & \multicolumn{5}{c}{\textrm{Events}} & \multicolumn{3}{c}{\textrm{Finished Paths}} & \multicolumn{1}{c}{\textrm{Halt}} \\[-0.5pt]\cmidrule(l{2.0pt}r{2.0pt}){1-3}\cmidrule(l{2.0pt}r{2.0pt}){4-6}\cmidrule(l{2.0pt}r{2.0pt}){8-12}\cmidrule(l{2.0pt}r{2.0pt}){13-15}& \textrm{Version} & \textrm{LoC} & \textrm{Time} & \textrm{RSS} & \textrm{\raisebox{0.20ex}{\#}I} & & \textrm{\textSigma} & \textrm{Mut} & \textrm{CV} & \textrm{\textlambda} & \textrm{Cut} & \textrm{Exit} & \textrm{Err} & \textrm{Cut} & \textrm{Reason} \\[-1.0pt]\midrule \multirow{5}{*}{\rotatebox{90}{\textrm{memcached}}}& \textrm{1.5.19}	& 31065	& 0:00:07	& 204MB	& 23K	& 1	& 12	& 6	& 0	& 3	& 0	& 0	& 1	& 0	& Finished	\\& \textrm{1.5.19+}	& 31051	& 1:33:42	& 208GB	& 1.2B	& 6	& 331K	& 271K	& 60K	& 3	& 24K	& 0	& 41K	& 29K	& Memory	\\& \textrm{1.5.20}	& 31093	& 0:00:07	& 197MB	& 92K	& 2	& 24	& 16	& 0	& 3	& 0	& 0	& 1	& 0	& Finished	\\& \textrm{1.5.20+}	& 31093	& 1:51:10	& 207GB	& 228M	& 10	& 745K	& 742K	& 2.7K	& 5	& 882	& 0	& 1	& 2.6K	& Memory	\\& \textrm{1.5.21}	& 31090	& 1:29:57	& 207GB	& 546M	& 10	& 1.1M	& 1.1M	& 3.1K	& 3	& 558	& 0	& 0	& 2.6K	& Memory	\\\cmidrules{}\multirow{2}{*}{\rotatebox{90}{\textrm{sort}}}& \textrm{8.31}	& 86596	& 0:24:29	& 23GB	& 266M	& 2	& 1.8M	& 1.4M	& 269K	& 25K	& 58K	& 8.0K	& 4.9K	& 55K	& Finished	\\& \textrm{8.31+}	& 86599	& 4:01:39	& 88GB	& 1.0B	& 2	& 6.9M	& 5.8M	& 777K	& 276K	& 346K	& 6.3K	& 0	& 285K	& Time	\\\bottomrule \end{tabular}\end{center}
 	\caption{
		Our prototype analyzing various versions of memcached and GNU sort.
		\footnotesize Timeout set at \SI{4}{\hour} with maximum memory usage of \SI{200}{\giga\byte}.
      Columns are: RSS:\@ maximum resident set size (swap space is not available);
      \raisebox{0.20ex}{\#}I:\@ number of instructions executed;
      Th: maximum number of threads active at the same time;
      \textSigma{}:
      total number of events in the explored unfolding; Mut: number of mutex
      lock/unlock events; CV:\@ number of wait1/wait2/signal/broadcast events;
      \textlambda{}: number of symbolic choices; Cut: number of events
      determined to be cutoffs; and the number of Finished Paths distinguish
      between normal termination of the program under test (Exit), detection of
      an error (Err) and being cut off (Cut).
	}\label{fig:eval}
\vspace{-0.25cm}
\end{table}

\subsubsection{Experiment.}
We run our prototype on five different versions of memcached, the three releases 1.5.19, 1.5.20 and 1.5.21 plus variants of the earlier releases (1.5.19\raisebox{0.15ex}{+} and 1.5.20\raisebox{0.15ex}{+}) which include patches for the two bugs we found during program initialization.
Those variants are included to show performance when not restricted by inescapable errors very early in the program execution.

\Cref{fig:eval} shows clearly how the two initialization bugs may lead to very quick analyses---versions 1.5.19 and 1.5.20 are completely analyzed in 7 seconds each, while versions  1.5.19\raisebox{0.15ex}{+}, 1.5.20\raisebox{0.15ex}{+} and 1.5.21 exhaust the memory budget of \SI{200}{\giga\byte}.
We have configured the experiment to stop the analysis once the memory limit is reached, although the analysis could continue in an incomplete manner by removing parts of the exploration frontier to free up memory.
Even though the number of error paths in \cref{fig:eval} differs between
configurations, it is notable that each configuration can only reach exactly one
of the bugs, as execution is arrested at that point.
When not restricted to the program initialization, the analysis of memcached produces hundreds of thousands of events and retires hundreds of millions of instructions in less than \SI{2}{\hour}.

Our setup delivers a single symbolic packet to memcached followed by a concrete shutdown packet.
As this packet can obviously only be processed once the server is ready to
process input, we observe symbolic choices only after program startup is complete.
(Since our prototype builds on KLEE, note that it assumes a single symbolic choice during startup, without generating an additional path.)

\subsection{GNU sort}
GNU sort uses threads for speeding up the sorting of very large workloads.
We reduced the minimum size of input required to trigger concurrent sorting to four lines to enable the analysis tools to actually trigger concurrent behavior.
Nevertheless, we were unable to avoid crashing Yogar-CBMC on this input.

During analysis of GNU sort 8.31, our prototype detected a data race, that we manually verified, but were unable to trigger in a harmful manner.
\Cref{fig:eval} shows two variants of GNU sort, the baseline version with eager parallelization (8.31) and a version with added locking to prevent the data race (8.31\raisebox{0.15ex}{+}).

Surprisingly, version 8.31 finishes the exploration, as all paths either exit, encounter the data race and are terminated or are cut off.
By fixing the data race in version 8.31\raisebox{0.15ex}{+}, we make it possible for the exploration to continue beyond this point, which results in a full \SI{4}{\hour} run that retires a full billion instructions while encountering almost seven million unique events.
 \section{Related Work}\label{sec:related}

The body of work in
\emph{systematic concurrency
testing}~\cite{FG05,AAJS14,RSSK15,NRSCP18,YNPP12,AAAJLS15,GFYS07,God97,TDB16}
is large.
These approaches explore thread interleavings under a fixed
program input. They prune the search space using
context-bounding~\cite{MQ07},
increasingly sophisticated
PORs~\cite{AABGS17,CCPSV18,FG05,AAJS14,AAAJLS15,RSSK15,NRSCP18,GFYS07}, or
random testing~\cite{CJXML18,YNPP12}.
Our main difference with these techniques is that we handle input data.

Thread-modular abstract interpretation~\cite{Mine14,KW16,FM07}
and
unfolding-based abstract interpretation~\cite{SRDK17}
aim at proving safety rather than finding bugs.
They use over-approximations to explore all behaviors,
while we focus on testing and never produce false alarms.
\emph{Sequentialization} techniques~\cite{QW04,ITFLP14,NSFTP17}
encode a multi-threaded program into
a sequential one.
While these encodings can be very effective for small programs~\cite{ITFLP14}
they grow quickly with large context bounds (5 or more, see~\cite{NSFTP17}).
However, some of the bugs found by our technique (\cref{sec:exp})
require many context switches to be reached.

\emph{Bounded-model checking}~\cite{AKT13,PSSTY17,YDLLW18,CF11,KWG09}
for multi-threaded programs encode
multiple program paths into a single logic formula,
while our technique encodes a single path.
Their main disadvantage is that for very large programs, even constructing the
multi-path formula can be extremely challenging, often producing an upfront
failure and no result.
Conversely, while our approach faces path explosion, it is always able to test
some program paths.

Techniques like~\cite{FHRV13,KSH15,SA06} operate on a data 
structure conceptually very similar to our unfolding.
They track read/write operations to every variable, which becomes a liability
on very large executions.
In contrast, we only use POSIX synchronization
primitives and compactly represent memory accesses to detect data races.
Furthermore, they do not exploit anything similar to cutoff events
for additional trace pruning.

Interpolation~\cite{WKO13,CJ14}
and weakest preconditions~\cite{GKWYG15} have been combined with POR and symbolic execution for \emph{property-guided}
analysis.
These approaches 
are mostly complementary to PORs like our technique, as they eliminate a
different class of redundant executions~\cite{GKWYG15}.

This work builds on top of previous work~\cite{NRSCP18,SRDK17,RSSK15}.
The main contributions \wrt{} those are:
(1) we use symbolic execution instead of concurrency testing~\cite{NRSCP18,RSSK15}
or abstract interpretation~\cite{SRDK17};
(2) we support condition variables, providing algorithms to compute conflicting
extensions for them; and
(3) here we use hash-based fingerprints to compute cutoff events, thus handling
much more complex partial orders than the approach described in~\cite{SRDK17}.
 \section{Conclusion }\label{sec:concl}

Our approach combines POR and symbolic execution to
analyze programs \wrt{} both input (data) and concurrency non-determinism.
We model a significant portion of the pthread API, including try-lock operations and robust mutexes.
We introduce two techniques to cope with state-space explosion in
real-world programs.
We compute cutoff events by using efficiently-computed fingerprints that
uniquely identify the total state of the program.
We restrict scheduling to synchronization points and report data races as
errors.
Our experiments found previously unknown bugs in real-world software projects (memcached, GNU sort).

\section*{Acknowledgements}
This research is supported in parts by the European Research Council (ERC) under the European Union's Horizon 2020 Research and Innovation Programme (grant agreement No.~647295 (SYMBIOSYS)).

\ifcav{}
\else
\clearpage
\appendix

\section{Model of Computation}\label{app:model}

In this section we present a model of computation suitable for describing
multi-threaded programs that use POSIX threading.

A \emph{concurrent program} is defined as a structure
$P \eqdef \tup{L, \mem, \locks, \conds, T, m_0, p_0}$, where
$L$ is the set of~\emph{program locations},
$\mem$~is the set of~\emph{memory states}
(valuations of program variables),
$\locks$~is the set of~\emph{mutexes},
$\conds$~is the set of~\emph{condition variables},
$m_0 \in \mem$~is the~\emph{initial memory state}, and
$p_0 \colon \Np \to L$~is a function that maps every thread identifier
(in $\Np \eqdef \N \setminus \set{0}$)
to its \emph{initial location}, in addition to
$T \subseteq \Np \times L \times L \times Q \times 2^{\mem \times \mem}$,
the set of \emph{thread statements}.
A given thread statement $t \eqdef \tup{i, n, n', q, r} \in T$ intuitively
represents that thread~$i$ can execute operation~$q$, updating the program
counter from~$n$ to~$n'$ and the memory from~$m \in \mem$ to $m' \in \mem$ if
$\tup{m, m'} \in r$.

An operation characterizes the nature of the activity performed by a statement.
We distinguish the following set of operations:
\[
Q \eqdef
  \set{\local} \cup
  (\set{\lock,\unlock} \times \locks) \cup
  (\set{\wait} \times \conds \times \locks) \cup
  (\set{\signal,\bcast} \times \conds)
\]
Statements carrying a~$\local$ operation model thread-local code.
A statement of $\tup{\lock, l}$ operation represents a request to acquire
mutex~$l$.
A $\tup{\wait, c,l}$ models a request to wait for a notification on
condition variable~$c$.
Finally, operations $\tup{\signal, c}$ and $\tup{\bcast, c}$ represent,
respectively, a \emph{signal} and a \emph{broadcast} operation on the
condition variable~$c$.

\section{Transition System Semantics}\label{app:lts}

We use \emph{labeled transition systems} (LTS)~\cite{CGP99} semantics for our
programs.
We associate a program~$P$ with the LTS~$M_P \eqdef \tup{S, {\to}, A, s_0}$.
The set
\[
S \eqdef (\Np \to L) \times \mem \times (\locks \to \N) \times (\conds \to \smash{2^{(\Np \cup -\Np)}})
\]
are the \emph{states} of~$M_P$, \ie, tuples of the form~$\tup{p,m,u,v}$
where~$p$ is a function that indicates the program location of every thread,
$m$ is the state of the (global) memory,
$u$~indicates when a mutex is locked (by thread~$i \ge 1$) or unlocked~(0),
and~$v$ maps every condition variable to a set of integers containing the
(positive) identifiers of those threads that currently wait on that condition
variable as well as negated thread identifiers of those threads that have
already been woken up by a signal or broadcast to that condition variable.
The \emph{initial state} is~$s_0 \eqdef \tup{p_0, m_0, u_0, v_0}$,
where $p_0$ and $m_0$ come from~$P$,
$u_0 \colon \locks \to \set 0$ is the function that maps every lock to the
number~$0$,
and $v_0 \colon \conds \to \set \emptyset$ is the function that maps every
condition variable to an empty set.

An \emph{action} in~$A \subseteq \Np \times B$ is a pair~$\tup{i,b}$ where~$i$
is the identifier of the thread that executes some statement and~$b$ is the
\emph{effect} of the statement.  Effects characterize the nature of an LTS
transition in $M_P$ similarly to how operations capture the nature of a
statement in~$P$.  We consider the following set of effects:
\begin{align*}
 B \eqdef& ~(\set{\loc} \times T) \cup (\set{\acq,\rel} \times \locks)\\
     \cup& ~(\set{\wa,\waa} \times \conds \times \locks)\\
     \cup& ~(\set{\sigg} \times \conds \times \N)
            \cup (\set{\bro} \times \conds \times \smash{2^\Np}\!)
\end{align*}
As we will see below, the execution of a statement gives rise to a transition
whose effect is in correspondence with the operation of the statement.

The \emph{transition relation}~${\to} \subseteq S \times A \times S$
contains triples of the form $s \fire{\tup{i,b}} s'$ that represent the
interleaved execution of a statement of thread~$i$ with effect~$b$ updating the
global state from~$s$ to~$s'$.
The precise definition of~$\to$ is given by the inference rules in~\cref{fig:rules}.
These rules insert transitions in~$M_P$ that reflect the execution of statements
on individual states.

As announced above, the effect of the inserted transition is in correspondence
with the operation of the statement.
For instance, the $\ruleloc$ rule inserts a transition labeled by a \loc\xspace
effect when it finds that a thread statement of \local\xspace operation can be
executed on a given LTS state.
Similarly, $\tup{\lock, l}$ and $\tup{\unlock, l}$ operations produce
transitions labeled by $\tup{\acq, l}$ and $\tup{\rel, l}$ effect.
However, a $\tup{\wait, c,l}$ action will produce \emph{two} successive LTS
transitions, as explained in \cref{sec:programs.actions}.

The POSIX standard leaves undefined the behavior of the program under various
circumstances, such as when a thread attempts to unlock a mutex not already
locked by the calling thread.
The rules in \cref{fig:rules} halt the execution of the calling thread whenever
undefined behavior would arise after executing an operation.
To the best of our knowledge, the only exception to this in our semantics is
when two threads concurrently attempt to wait on the same condition
variable using different mutexes. The standard declares this as undefined
behavior \cite{POSIX} but rule $\rulewa$ lets both threads run
normally.

Our semantics only implement the so-called non-robust \texttt{NORMAL} mutexes
\cite{POSIX}.
\texttt{RECURSIVE} and \texttt{ERRORCHECK} mutexes can easily be implemented
using auxiliary local variables.

\begin{figure*}[t]

\mathligson
\centering

\inference
   [\ruleloc]
   {t \eqdef \tup{i, n, n', \local, r} \in T & p(i) = n & \tup{m, m'} \in r}
   {\tup{p, m, u, v} \fire{\tup{i, \tup{\loc, t}}} \tup{p_{i \mapsto n'}, m', u, v}}

\bigskip
\inference
   [\ruleacq]
   {\tup{i, n, n', \tup{\lock, l}, r} \in T & p(i) = n & u(l) = 0}
   {\tup{p, m, u, v} \fire{\tup{i, \tup{\acq, l}}} \tup{p_{i \mapsto n'}, m, u_{l \mapsto i}, v}}

\bigskip
\inference
   [\rulerel]
   {\tup{i, n, n', \tup{\unlock, l}, r} \in T & p(i) = n & u(l) = i}
   {\tup{p, m, u, v} \fire{\tup{i, \tup{\rel, l}}} \tup{p_{i \mapsto n'}, m, u_{l \mapsto 0}, v}}

\bigskip
\inference
   [\rulewa]
   {\tup{i, n, n', \tup{\wait, c,l}, r} \in T & p(i) = n & u(l) = i}
   {\tup{p, m, u, v} \fire{\tup{i, \tup{\wa, c,l}}}
    \tup{p, m, u_{l \mapsto 0}, v_{c \mapsto v(c) \cup \set i}}}

\bigskip
\inference
   [\rulewaa]
   {\tup{i, n, n', \tup{\wait, c,l}, r} \in T & p(i) = n & u(l) = 0 & -i \in v(c)}
   {\tup{p, m, u, v} \fire{\tup{i, \tup{\waa, c,l}}}
    \tup{p_{i \mapsto n'}, m, u_{l \mapsto i}, v_{c \mapsto v(c) \setminus \set{-i}}}}

\bigskip
\inference
   [\rulesigg]
   {\tup{i, n, n', \tup{\signal, c}, r} \in T & p(i) = n & j \in v(c) & j > 0}
   {\tup{p, m, u, v} \fire{\tup{i, \tup{\sigg, c, j}}}
    \tup{p_{i \mapsto n'}, m, u, v_{c \mapsto v(c) \setminus \set j \cup \set{-j}}}}

\bigskip
\inference
   [\rulesiggl]
   {\tup{i, n, n', \tup{\signal, c}, r} \in T & p(i) = n & \set{j \in v(c) \colon j > 0} = \emptyset}
   {\tup{p, m, u, v} \fire{\tup{i, \tup{\sigg, c, 0}}} \tup{p_{i \mapsto n'}, m, u, v}}

\bigskip
\inference
   [\rulebro]
   {\tup{i, n, n', \tup{\bcast, c}, r} \in T & p(i) = n & \set{j \in v(c) \colon j > 0} = W}
   {\tup{p, m, u, v} \fire{\tup{i, \tup{\bro, c, W}}}
    \tup{p_{i \mapsto n'}, m, u, v_{c \mapsto v(c) \setminus W \cup {-W}}}}

\mathligsoff

\caption{LTS semantics $M_P$ of a concurrent program~$P$.
Notation $f_{x \mapsto y}$ denotes a function that behaves like~$f$ for all
inputs except for~$x$, where $f(x) = y$. For a given set
$S \subseteq \Np$, we further denote
$-S \eqdef \set{-s \in \Z \colon s \in S}$.}\label{fig:rules}
\end{figure*}

We finish this section with some additional definitions about LTSs.
If $s \fire{a} s'$ is a transition, then we say that
action~$a$ is \emph{enabled} at~$s$.
Let~$\enabl s$ denote the set of actions enabled at~$s$.
As actions may be non-deterministic,
firing $a$ may produce more than one such $s'$.
For a sequence $\sigma \eqdef a_1 \ldots a_n \in A^{*}$ and a state $s \in S$
we inductively define $\state{s,\sigma} \eqdef \set{s}$ if $|\sigma| = 0$ and
$\state{s,\sigma} \eqdef
\set{s'' \in S \colon s' \in \state{s,\sigma'} \land s' \fire{a_n} s''}$
otherwise, where $\sigma' a_n = \sigma$.
By extension we let $\state \sigma \eqdef \state{s_0,\sigma}$ denote the states
reachable by~$\sigma$ from the initial state.
We say that $\sigma$ is a \emph{run} when $\state \sigma \ne \emptyset$.
We let $\runs{M_P}$ denote the set of all runs and
$\reach{M_P} \eqdef \bigcup_{\sigma \in \runs{M_P}} \state \sigma$
the set of all \emph{reachable states}.
By extension, for a set $S' \subseteq S$ of states we say that~$a$ is enabled
at~$S'$ if~$a$ is enabled at some state in~$S'$.

In this work, we assume that the model of computation~$P$ satisfies the
following well-formedness condition:

\begin{definition}\label{def:well-formed}
   A concurrent program~$P$ is \emph{well-formed} if
   for any reachable state $s \in \reach{M_P}$
   and any two actions $a, a' \in A$ enabled at~$s$, we have that
   both actions are local, \ie,
   $a = \tup{\loc, t}$ and
   $a' = \tup{\loc, t'}$
   for some $t, t' \in T$.
\end{definition}

In simple words, a program~$P$ will therefore be well-formed when the only source
of data non-determinism is $\local$ statements in~$P$.
In this work we assume that any given program is well-formed.

\section{Independence}\label{app:indep}

Many partial-order methods use a notion called independence to avoid exploring
concurrent interleavings that lead to the same state.
We recall the standard notion of independence for actions in~\cite{God96}.
Given two actions~$a, a' \in A$ and one state $s \in S$, we say that
\emph{$a$ commutes with~$a'$ at state~$s$} iff for all $s' \in S$
we have:
\begin{itemize}
\item
  if $a \in \enabl s$ and $s \fire a s'$,
  then $a' \in \enabl s$ iff $a' \in \enabl{s'}$; and
  \eqtag{eq:com1}
\item
  if $a, a' \in \enabl s$, then $\state{s, aa'} = \state{s, a'a}$.
  \eqtag{eq:com2}
\end{itemize}
Independence between actions is an under-approximation of commutativity.
A binary relation ${\indep} \subseteq A \times A$ is a
\emph{valid independence} on~$M_P$ if it is symmetric, irreflexive, and
every pair $\tup{a, a'}$ in $\indep$ commutes at every state in $\reach{M_P}$.

In general $M_P$ has multiple independence relations.
Clearly $\emptyset$ is always one of them.
Broadly speaking, the larger an independence relation is, the less interleaved
executions a partial-order method will explore, as more of them will be
regarded as equivalent to each others.

Given a valid independence~$\indep$,
the complementary relation $(A \times A) \mathrel\setminus {\indep}$
is a \emph{dependence} relation, which we often denote by~$\depen$.

\section{Independence for Programs}\label{app:indep.programs}

Let~$P$ be a program.

\begin{definition}\label{def:indepp}
We define
the \emph{dependence relation}~${\depen_P} \subseteq A \times A$
as the only reflexive, symmetric
relation where $a \depen_P a'$ hold if either both $a$ and~$a'$ are actions of
the same thread or~$a$ matches the action described by the left column of
\cref{tab:indepp} and~$a'$ matches the right column.
Finally, we define the
\emph{independence relation}~${\indep_P} \eqdef (A \times A) \setminus {\depen_P}$
as the complementary of $\depen_P$ in~$A$.

\end{definition}

\begin{table*}[t]
\begin{center}
\begin{tabular}{p{2.3cm}p{9.7cm}}

\toprule
\textbf{Action} & \textbf{Dependent actions on threads other than $i$} \\[1pt]

\midrule
$\tup{i, \tup{\loc, t}}$
&
None.
\\

\midrule
$\tup{i,\tup{\acq, l}}$ or \newline
$\tup{i,\tup{\rel, l}}$
&
For any $j \ge 1$, and any $c \in \conds$ the following actions are dependent:
$\tup{j, \tup{\acq,l}}$,
$\tup{j, \tup{\rel,l}}$,
$\tup{j, \tup{\wa,c,l}}$,
$\tup{j, \tup{\waa,c,l}}$.
\\

\midrule
$\tup{i,\tup{\wa, c,l}}$
&
For any $j \ge 1$, and any $c' \in \conds$, and any $W \subseteq \Np$
the following actions are dependent:
$\tup{j, \tup{\acq,l}}$,
$\tup{j, \tup{\rel,l}}$,
$\tup{j, \tup{\wa,c',l}}$,
$\tup{j, \tup{\waa,c',l}}$,
$\tup{j, \tup{\sigg,c,0}}$,
$\tup{j, \tup{\sigg,c,i}}$,
$\tup{j, \tup{\bro,c,W}}$.
\\

\midrule
$\tup{i,\tup{\waa, c, l}}$
&
For any $j \ge 1$, and any $c' \in \conds$, and any $W \subseteq \Np$ such that
$i \in W$, the following actions are dependent:
$\tup{j, \tup{\acq,l}}$,
$\tup{j, \tup{\rel,l}}$,
$\tup{j, \tup{\wa,c',l}}$,
$\tup{j, \tup{\waa,c',l}}$,
$\tup{j, \tup{\sigg,c,i}}$,
$\tup{j, \tup{\bro,c,W}}$.
\\

\midrule
$\tup{i,\tup{\sigg, c, k}}$ \newline
with $k \ne 0$
&
For any $j \ge 1$, and any $l \in \locks$,
and any $W \subseteq \Np$ the following actions are dependent:
$\tup{k, \tup{\wa,c,l}}$,
$\tup{k, \tup{\waa,c,l}}$,
$\tup{j, \tup{\sigg,c,0}}$,
$\tup{j, \tup{\sigg,c,k}}$,
$\tup{j, \tup{\bro,c,W}}$.
\\

\midrule
$\tup{i,\tup{\sigg, c, 0}}$ \newline
or \newline
$\tup{i,\tup{\bro, c, \emptyset}}$
&
For any $j, k \ge 1$, and any $l \in \locks$,
and any $W \subseteq \Np$ such that $W \ne \emptyset$, the following actions are
dependent:
$\tup{j, \tup{\wa,c,l}}$,
$\tup{j, \tup{\sigg,c,k}}$,
$\tup{j, \tup{\bro,c,W}}$.
\\

\midrule
$\tup{i,\tup{\bro, c, W}}$ \newline
with $W \ne \emptyset$
&
For any $j \ge 1$, and any $j' \in W$, and any $k' \ge 0$, and any $l \in \locks$,
and any $W' \subseteq \Np$ the following actions are dependent:
$\tup{j, \tup{\wa,c,l}}$,
$\tup{j', \tup{\waa,c,l}}$,
$\tup{j, \tup{\sigg,c,k'}}$,
$\tup{j, \tup{\bro,c,W'}}$.
\\

\bottomrule
\end{tabular}
\end{center}
\caption{Two actions operating on the same thread are always dependent, here we give the part of our dependence relation~${\protect\depen_P}$ that concerns actions on distinct threads. See \cref{def:indepp} for the complete definition of~${\protect\depen_P}$ and~${\protect\indep_P}$.}\label{tab:indepp}
\end{table*}

We say that $P$ is \emph{data-race-free} iff
any two local actions
$a \eqdef \tup{i, \tup{\loc, t}}$ and
$a' \eqdef \tup{i', \tup{\loc, t'}}$
from different threads (\ie, $i \ne i'$) commute at every
reachable state $s \in \reach{M_P}$.
This ensures that $\local$ statements of~$P$ modify the memory in a
manner which does not interfere with $\local$ statements of other threads.

\begin{theorem}
If $P$ is data-race-free, then
$\indep_P$ is a valid independence relation.
\end{theorem}
\begin{proof}
Clearly, by construction, $\indep_P$ is symmetric and irreflexive.
Let
$a \eqdef \tup{i,b}$ and
$a' \eqdef \tup{i',b'}$
be two actions of~$M_P$.
Let $s \in \reach{M_P}$ be any reachable state of~$M_P$ such that
$s \eqdef \tup{p, m, u, v}$.
Assume that $a \indep_P a'$, and so, $i \ne i'$.
We need to show that~$a$ commutes with~$a'$.

Remark that the fact that~$a$ commutes with~$a'$ at~$s$ does not imply that~$a'$
commutes with~$a$ at~$s$, due to the way~\cref{eq:com1} is defined.
To simplify and organize the reasoning below better, we will not only prove
that~$a$ commutes with~$a'$ at~$s$, but for some combinations of~$a$ and~$a'$ we
will also prove that~$a'$ commutes with~$a$ at~$s$. We define the following
claims:

\begin{itemize}
\item
  If $a \in \enabl s$ and $s \fire a s'$,
  then $a' \in \enabl s$ iff $a' \in \enabl{s'}$.
  \eqtag{eq:indepcom1}
\item
  If $a' \in \enabl s$ and $s \fire{a'} s'$,
  then $a \in \enabl s$ iff $a \in \enabl{s'}$.
  \eqtag{eq:indepcom2}
\item
  If $a, a' \in \enabl s$, then $\state{s, aa'} = \state{s, a'a}$.
  \eqtag{eq:indepcom3}
\end{itemize}

Showing \cref{eq:indepcom1} and \cref{eq:indepcom3} will prove that $a$ commutes
with~$a'$ at~$s$.
And showing
\cref{eq:indepcom2} and \cref{eq:indepcom3} will prove that $a'$
commutes with~$a$ at~$s$.
The proof is by cases on~$a$:

\begin{itemize}
\item
  Assume that $b = \tup{\loc, t}$.
  If $b' = \tup{\loc, t'}$, for any $t' \in T$, then $a$ commutes with~$a'$ at~$s$ (and $a'$ commutes
  with~$a$) because we assumed that~$P$ is data-race-free.
  So assume that $b'$ is of $\acq$, $\rel$, $\wa$, $\waa$, $\sigg$, or $\bro$ effect.

  We show \cref{eq:indepcom1}.
  Clearly firing~$a$ at~$s$ cannot enable or disable~$a'$.
  This is because the preconditions of rules
  $\ruleacq$,
  $\rulerel$,
  $\rulewa$,
  $\rulewaa$,
  $\rulesigg$,
  $\rulesiggl$, and
  $\rulebro$
  cannot be affected by the modification to the state that rule
  $\ruleloc$
  produces. Rule $\ruleloc$ only modifies the memory component of the
  state and the program counter for thread~$i$. And the other rules do not
  check the memory component, they check the program counter for thread
  $i'$, where $i' \ne i$, as well as specific conditions on $u$ and~$v$,
  which rule $\ruleloc$ does not modify.
  Consequently \cref{eq:indepcom1} holds.

  We show \cref{eq:indepcom2}.
  Similarly, firing~$a'$ at~$s$ cannot enable or disable~$a$.
  The reasoning here is analogous to the previous case: the rules
  $\ruleacq$,
  $\rulerel$,
  $\rulewa$,
  $\rulewaa$,
  $\rulesigg$,
  $\rulesiggl$, and
  $\rulebro$
  never update a component of the state that can affect the enabledness of
  rule $\ruleloc$.
  Consequently \cref{eq:indepcom2} holds.

  We show \cref{eq:indepcom3}.
  For the same reasons as above, firing rule $\ruleloc$ followed by any other
  rule, or firing such other rule followed by $\ruleloc$ yields the same state.

  This proves that $\tup{i, \tup{\loc, t}}$ commutes with~$a'$ at~$s$, and
  that~$a'$ commutes with $\tup{i, \tup{\loc, t}}$ at~$s$.

\item
  Assume that $b \in \set{\tup{\acq, l}, \tup{\rel, l}}$.
  If $b' = \tup{\loc, t}$, for any $t \in T$, then we have already proven
  (above) that $a$ commutes with~$a'$ at~$s$.

  So assume that $b' \in \set{\tup{\acq, l'}, \tup{\rel, l'}}$.
  By \cref{def:indepp} we know that $l \ne l'$ as otherwise $b$ and $b'$ would
  be dependent.
  We need to prove \cref{eq:indepcom1,eq:indepcom3}.
  Showing \cref{eq:indepcom1} will also imply \cref{eq:indepcom2} because~$b$
  and~$b'$ are symmetric.
  Both \cref{eq:indepcom1,eq:indepcom3} hold owing to the same reason:
  applying rule $\ruleacq$ (\resp $\rulerel$) on a $\tup{\lock,l}$ operation
  (\resp $\tup{\unlock,l}$) of thread~$i$ does not interfere with applying
  $\ruleacq$ on $\tup{\lock,l'}$ or
  $\rulerel$ on $\tup{\unlock,l'}$
  on a different thread~$i'$ because $\ruleacq$ (\resp $\rulerel$)
  only updates and checks information that is local to each application, \ie,
  the program counter and the lock in question. Since the
  program counters and the locks are different in these applications, they
  cannot interfere between each other. In particular, observe that the
  global memory remains unchanged.

  Assume now that $b' \in \set{\tup{\wa, c,l'}, \tup{\waa, c,l'}}$, for any
  $c \in \conds$ and any $l' \in \locks$.
  Again by \cref{def:indepp} we know that $l \ne l'$.
  We need to prove \cref{eq:indepcom1}, \cref{eq:indepcom2}, and
  \cref{eq:indepcom3}.
  All the three claims hold owing to the same reason:
  the application of a rule from the set
  $\set{\ruleacq, \rulerel}$
  does not interfere with the application of any rule from the set
  $\set{\rulewa, \rulewaa}$
  if the locks on which they are applied are different and the threads on
  which they are applied are different as well.

  Finally, assume that $b' \in \set{\tup{\sigg, c, k}, \tup{\bro, c, W}}$
  for some $c \in \conds$, $k \in \N$, and $W \subseteq \Np$.
  Similarly to the above, applying one rule from the set
  $\set{\ruleacq, \rulerel}$
  does not interfere with applying any rule from the set
  $\set{\rulesigg, \rulesiggl, \rulebro}$
  if they are applied on different threads,
  because each one of the rules only update the instruction pointer of the
  thread and, while $\ruleacq$ and $\rulerel$ only update lock information,
  rules $\rulesigg$, $\rulesiggl$, and $\rulebro$ only update information
  about condition variables.

\item
  Assume that $b = \tup{\wa, c, l}$.
  If $b' \in \set{\tup{\loc, t}, \tup{\acq, l'}, \tup{\rel, l'}}$,
  for some $t \in T$ and some $l' \in \locks$, then we have already shown
  that~$a$ commutes with~$a'$ at~$s$.

  Assume that $b' = \tup{\wa, c', l'}$
  for some $l' \in \locks$ and $c' \in \conds$.
  By \cref{def:indepp} we have $l \ne l'$,
  but $c'$ might be equal to~$c$ or not.
  Clearly the application of $\rulewa$ commutes with itself because
  we have $l \ne l'$, $\rulewa$ does not check for the state of $v(c)$,
  and even if $c = c'$ with both applications updating $v(c)$, both orders of
  adding $i$ and $i'$ produce the same state.

  Assume now that $b' = \tup{\waa, c',l'}$
  for some $c' \in \conds$ and $l' \in \locks$.
  By \cref{def:indepp} we have $l \ne l'$,
  but $c'$ might be equal to~$c$ or not.
  Rules $\rulewa$ and $\rulewaa$ commute with each other because:
  $l \ne l'$, so the application of neither rule modifies the result of the
  preconditions that check~$l$ or~$l'$;
  both rules check and update instruction pointers for different threads;
  while $\rulewaa$ does check whether $-i' \in v(c)$ and $\rulewa$
  modifies~$v(c)$, $\rulewa$ cannot modify the outcome of the check
  because it only adds~$i$ to~$v(c)$, and $i \ne -i'$.

  Assume that $b' = \tup{\sigg, c', 0}$, with $c' \in \conds$.
  By \cref{def:indepp} we have $c' \ne c$.
  Then rule $\rulewa$ commutes with rule $\rulesiggl$ because
  while $\rulewa$ checks and updates lock~$l$, rule~$\rulesiggl$ does not
  check or update locks;
  and because while rule $\rulewa$ updates $v(c)$ and rule $\rulesiggl$
  performs a check on~$v(c')$, $c \ne c'$, so the rules do not interfere.

  Assume that $b' = \tup{\sigg, c', k}$, with $c' \in \conds$
  and $k \in \N \setminus \set{0, i}$.
  Then rule $\rulewa$ commutes with rule $\rulesigg$ because:
  while $\rulewa$ checks and updates lock~$l$, rule~$\rulesigg$ does not
  check or update locks;
  rule $\rulewa$ adds~$i$ to~$v(c)$ and rule~$\rulesigg$ checks
  if $k \in v(c')$, but $i \ne k$, so the update cannot interfere with the
  check, whether $c = c'$ or not;
  finally, rule $\rulewa$ adds~$i$ to~$v(c)$ and rule~$\rulesigg$ replaces~$k$
  with $-k$ in $v(c')$, but $i \ne k$ and $i \ne -k$, so both updates result
  in the same state regardless of the order of execution,
  whether $c = c'$ or not.

  Assume that $b' = \tup{\sigg, c', i}$, with $c' \in \conds$.
  By \cref{def:indepp} we have $c' \ne c$.
  Then rule $\rulewa$ commutes with rule $\rulesigg$ because:
  while $\rulewa$ checks and updates lock~$l$, rule~$\rulesigg$ does not
  check or update locks;
  rule $\rulewa$ adds~$i$ to~$v(c)$ and rule~$\rulesigg$ checks
  if $i \in v(c')$, but $c \ne c'$, so the update cannot interfere with the
  check;
  finally, both rules respectively update~$v(c)$ and~$v(c')$,
  but since $c \ne c'$ the updates result in the same state regardless of
  the order of execution.

  Finally, assume that $b' = \tup{\bro, c', W}$, with $c' \in \conds$
  and $W \subseteq \Np$.
  By \cref{def:indepp} we have $c' \ne c$.
  Then rule $\rulewa$ commutes with rule $\rulebro$ because of the following
  reasons.
  While $\rulewa$ checks and updates lock~$l$, rule~$\rulebro$ does not
  check or update locks.
  Rule $\rulewa$ adds~$i$ to~$v(c)$ and rule~$\rulebro$ modifies
  $v(c')$. But running these two updates in any order results in the same
  state because $c \ne c'$.

  Therefore, we have shown that in any case~$a$ and~$a'$ satisfy
  \cref{eq:indepcom1}, \cref{eq:indepcom2}, and \cref{eq:indepcom3} on~$s$.

\item
  Assume that $b = \tup{\waa, c,l}$.
  If $b' \in \set{\tup{\loc, t}, \tup{\acq, l'}, \tup{\rel, l'}, \tup{\wa, c',l'}}$,
  for some $t \in T$, and some $l' \in \locks$ and some $c' \in \conds$, then we
  have already shown that~$a$ commutes with~$a'$ at~$s$.

  Assume that $b' = \tup{\waa, c', l'}$
  for some $c' \in \conds$ and some $l' \in \locks$.
  By \cref{def:indepp} we know that $l \ne l'$.
  In this case rule $\rulewaa$ commutes with itself for the following reasons.
  While each application checks and updates the states of locks and condition
  variables neither application interferes with the other and regardless
  of their order, the same state is reached after applying both.
  For locks, rule $\rulewaa$ checks and updates the state of, respectively,
  locks~$l$ and~$l'$, but they are different, so there is no interference.
  For condition variables, $\rulewaa$ checks for (and removes) $-i \in v(c)$
  and $-i' \in v(c')$, respectively, but since $i \ne i'$, regardless of
  whether $c \ne c'$, the actions do not interfere.

  Assume that $b' = \tup{\sigg, c', k}$,
  for some $c' \in \conds$ and $k \in \N$.
  By \cref{def:indepp} we know that either $c' \ne c$ or $k \ne i$.
  If $k = 0$, clearly rules $\rulewaa$ and $\rulesiggl$ commute because, while
  $\rulewaa$ updates $u(l)$, $\rulesiggl$ does not check or update lock
  information; and while $\rulewaa$ checks for $v(c)$, $\rulesiggl$ does not
  update any condition variable information.
  If $k \ne 0$, then clearly also rules $\rulewaa$ and $\rulesigg$ commute.
  This is because while $\rulewaa$ checks and updates lock~$l$,
  $\rulesigg$ neither checks nor updates any lock information.
  Additionally, while both rules check the state of condition variables~$c$
  and~$c'$, respectively, they clearly commute if~$c \ne c'$, as $\rulesigg$
  only updates $v(c')$. Furthermore, if $c = c'$, then $k \ne i$, and so
  updating $v(c)$ in $\rulesigg$ by replacing~$k$ with~$-k$ cannot
  interfere with $\rulewaa$'s check for $-i \in v(c)$.

  Finally, assume that $b' = \tup{\bro, c', W}$, for some $c' \in \conds$
  and some $W \subseteq \Np$.
  By \cref{def:indepp} we know that either $c' \ne c$ or $i \notin W$.
  Clearly $\rulewaa$ and $\rulebro$ commute if $c \ne c'$, as $\rulebro$
  neither checks nor updates lock information and it only updates
  $v(c')$, while $\rulewaa$ operates on $v(c)$.
  So assume that $c = c'$, and so $i \notin W$.
  We show~\cref{eq:indepcom2}.
  Applying $\rulebro$ neither enables nor disables $\rulewaa$.
  This is because if $a$ is enabled ($\rulewaa$ is enabled) at~$s$, then
  $-i \in v(c)$.
  With applying $\rulebro$, every element of $W$ (which only contains positive
  integers) is replaced by its negation in $v(c') = v(c)$.
  This does not remove any negative element, especially $-i$, from $v(c)$ and
  also does not add $-i$ as we have $i \notin W$. Thus, the validity of
  $-i \in v(c)$ as prerequisite of $\rulewaa$ remains.
  Since $\rulebro$ additionally does not update lock information (in~$u$),
  applying $\rulebro$ leaves the enabledness of rule $\rulewaa$ unchanged.
  We show~\cref{eq:indepcom1}.
  Similarly applying rule $\rulewaa$ only updates lock information, and rule
  $\rulebro$ does not test for the state of locks (in fact it only checks
  for the instruction pointer of thread~$i'$).
  Consequently, firing~$a$ neither enables nor disables~$a'$.
  Finally, proving \cref{eq:indepcom3} is an easy exercise.

\item
  Assume that $b = \tup{\sigg, c, k}$ with $k \ne 0$.
  If for some $t \in T$, and some $l' \in \locks$ and some $c' \in \conds$, we
  have $b' \in \set{\tup{\loc, t}, \tup{\acq, l'}, \tup{\rel, l'},
  \tup{\wa, c', l'}, \tup{\waa, c', l'}}$,
  then we have already shown that~$a$ commutes with~$a'$ at~$s$.

  Assume that $b' = \tup{\sigg, c', j}$,
  for some $c' \in \conds$ and some $j \ge 1$.
  By \cref{def:indepp} we know that either $c \ne c'$ or $j \ne k$.
  As a result clearly, $a$ commutes with~$a'$, because rule~$\rulesigg$
  commutes with itself, either because it updates (and checks) different
  condition variables or because it updates different threads in the set
  of the same condition variable.

  Assume that $b' = \tup{\sigg, c', 0}$,
  for some $c' \in \conds$.
  By \cref{def:indepp} we know that $c \ne c'$.
  As a result rules $\rulesigg$ and $\rulesiggl$ trivially commute because
  they work on different condition variables.

  Assume that $b' = \tup{\bro, c', W}$,
  for some $c' \in \conds$ and some $W \subseteq \Np$.
  By \cref{def:indepp} we know that $c \ne c'$.
  Similarly rules $\rulesigg$ and $\rulebro$ trivially commute because they
  work on different condition variables.

\item
  Assume that $b = \tup{\sigg, c, 0}$.
  If for some $t \in T$, $l' \in \locks$, $c' \in \conds$,
  and $k \ge 1$ we have
  $b' \in \set{\tup{\loc, t}, \tup{\acq, l'}, \tup{\rel, l'},
  \tup{\wa, c', l'}, \tup{\waa, c', l'}, \tup{\sigg, c', k}}$,
  then we have already shown that~$a$ commutes with~$a'$ at~$s$.

  Assume that $b' = \tup{\sigg, c', 0}$,
  for some $c' \in \conds$.
  Rule $\rulesiggl$ commutes with itself because it does not update the
  state of locks or condition variables.

  Assume that $b' = \tup{\bro, c', W}$,
  for some $c' \in \conds$ and some $W \subseteq \Np$.
  By \cref{def:indepp} we know that either $c \ne c'$ or $W = \emptyset$.
  Rule $\rulesiggl$ checks condition variable~$c$ and does not update the
  state, and rule $\rulebro$ does not check the state of condition
  variables but updates the state of~$c'$.
  If $c \ne c'$ they trivially commute because they work on different
  condition variables.
  So assume that $c = c'$, and so, $W = \emptyset$.
  We show~\cref{eq:indepcom1}.
  Firing~$a$ at~$s$ does not modify~$u$ or~$v$ (rule $\rulesiggl$),
  so rule $\rulebro$ is neither enabled nor disabled by~$a$.
  We show~\cref{eq:indepcom2}.
  If we can fire~$a'$ at~$s$, it must be because
  $\set{j \in v(c') \colon j > 0} = W = \emptyset$,
  and so $\set{j \in v(c) \colon j > 0} = \emptyset$.
  After firing rule~$\rulebro$ we still have
  $\set{j \in v(c) \colon j > 0} = \emptyset$.
  So firing~$\rulebro$ neither enables nor disables~$a$ at~$s$.
  We show~\cref{eq:indepcom3}.
  Assume that both~$a$ and~$a'$ are enabled at~$s$.
  Consequently
  $\set{j \in v(c) \colon j > 0} = \set{j \in v(c') \colon j > 0} = \emptyset$.
  While rule $\rulebro$ updates $v(c')$, the update is immaterial since $W$ is
  empty.
  So applying both rules in either order produces the same state.

\item
  Assume that $b = \tup{\bro, c, W}$ with $W \subseteq \Np$.
  If for some $t \in T$, $l' \in \locks$, $c' \in \conds$, and $k \in \N$,
  $b' \in \set{\tup{\loc, t}, \tup{\acq, l'}, \tup{\rel, l'},
  \tup{\wa, c', l'}, \tup{\waa, c', l'}, \tup{\sigg, c', k}}$,
  then we have already shown that~$a$ commutes with~$a'$ at~$s$.

  Assume that $b' = \tup{\bro, c', W'}$,
  for some $c' \in \conds$ and some $W' \subseteq \Np$.
  By \cref{def:indepp} we know that either $c \ne c'$ or $W = W' = \emptyset$.
  If $c \ne c'$ clearly rule $\rulebro$ commutes with itself because each
  application checks and updates a different condition variable.
  If $c = c'$, and so $W = W' = \emptyset$, rule $\rulebro$ also
  commutes with itself because
  $\set{j \in v(c) \colon j > 0} = \set{j \in v(c') \colon j > 0} = \emptyset$
  at~$s$, and so the update
  $v_{c \mapsto v(c)\setminus W \cup - W}$ is immaterial.

\end{itemize}
\end{proof}

\section{Unfolding Semantics}\label{app:unfolding}

We recall the program PES semantics of~\cite{SRDK17}
(modulo notation differences).
For a program $P$ and any independence~$\indep$ on~$M_P$ we define a
PES~$\unfpindep$
that represents the behavior of~$P$, \ie, such that
the interleavings of its set of configurations equals $\runs{M_P}$.

Each event in~$\unfpindep$ is inductively defined by a pair
of the form $e \eqdef \tup{a,H}$, where $a \in A$ is an action of~$M_P$
and~$H$ is a configuration of $\unfpindep$.
Intuitively, $e$ represents the occurrence of~$a$
after the \emph{causes} (or the history) $H$, as described
in~\cref{sec:algo.unfsem}.
Note the inductive nature of the name, and how it
allows to uniquely identify each event.
We define the \emph{state of a configuration} as the set of states reached by \emph{any}
of its interleavings.
Formally, for $C \in \conf{\unfpindep\>\!}$ we define
$\state C$ as $\set{s_0}$ if $C = \emptyset$ and
as $\state \sigma$ for some $\sigma \in \inter C$ if $C \ne \emptyset$.
Despite its appearance $\state C$ is well-defined because \emph{all} sequences
in $\inter C$ reach the \emph{same} set of states, see~\cite{SRDK17long} for a proof.
\begin{definition}[Unfolding]
\label[definition]{def:unfd}
Given a program~$P$ and some independence relation $\indep$
on $M_P \eqdef \tup{S, \to, A, s_0}$, the
\emph{unfolding of~$P$ under~$\indep$}, denoted $\unfpindep$,
is the PES over~$A$ constructed by the following fixpoint rules:
\begin{enumerate}[topsep=0pt]
\item
  Start with a PES $\les \eqdef \tup{E, <, {\cfl}, h}$
  equal to $\tup{\emptyset, \emptyset, \emptyset, \emptyset}$.
\item
  Add a new event $e \eqdef \tup{a,C}$ to~$E$ for any
  configuration $C \in \conf \les$ and any action $a \in A$ if
  $a$~is enabled in $\state C$ and
  $\lnot (a \indep h(e'))$ holds for every $<$-maximal event $e'$ in~$C$.
  \label{def:unfd.event}
\item
  For any new $e$ in $E$, update $<$, $\cfl$, and $h$ as follows:
for every $e' \in C$, set $e' < e$;
for any $e' \in E \setminus C$,
    set $e' \cfl e$
    if $e \ne e'$ and $\lnot (a \indep h(e'))$;
set $h(e) \eqdef a$.
\item
  Repeat steps 2 and 3 until no new event can be added to~$E$;
  return $\les$.
\end{enumerate}
\end{definition}
Step 1 creates an empty PES with only one (empty) configuration.
Step 2 inserts a new event $\tup{a,C}$ by finding a configuration~$C$ that
enables an action~$a$ which is dependent with all causality-maximal events
in~$C$.

After inserting an event~$e \eqdef \tup{a,C}$, \cref{def:unfd} declares
all events in~$C$ causal predecessors of~$e$.
For any event~$e'$ in $E$ but not in $[e]$ such that~$h(e')$ is dependent
with~$a$, the order of execution of~$e$ and~$e'$ yields different states.
We thus set them in conflict.

\begin{theorem}[Soundness and completeness \cite{SRDK17}]
For any program $P$ and any independence relation~$\indep$ on~$M_P$ we have:
\begin{enumerate}
\item $\unfpindep$ is uniquely defined.
\item For any configuration $C$ of $\unfpindep$ if
   $\sigma, \sigma' \in \inter C$ then $\state \sigma = \state{\sigma'}$.
\item Any interleaving of a local configuration of $\unfpindep$ is a run of~$M_P$.
\item For any run $\sigma$ of $M_P$ there is a configuration
   $C$ of $\unfpindep$ with $\sigma \in \inter C$.
\end{enumerate}
\end{theorem}
\begin{proof}
  Item 1: see Proposition 2 in~\cite{SRDK17long}.
  Items 2 and 3: see items 1 and 3, respectively, of Theorem 1 in~\cite{SRDK17long}.
  Item 4: see Theorem 2 in~\cite{SRDK17long}.
\end{proof}

\section{Computing Conflicting Extensions}\label{app:cex}

Given a configuration~$C$ of~$\unfpindep$, procedure \cexmain in
\cref{a:cex.main} computes the conflicting extensions~$\cex C$ of~$C$.

\begin{algorithm*}[t]
\DontPrintSemicolon

\Fn{\cexmain{C}}{
   $R \eqdef \emptyset$ \;
   Add to~$C$ any event $e \in \en C$ such that $\iscutoff{e}$ \;
   \label{l:cex.main.cutoff}
   \lForEach {event $e \in C$ of $\loc$ effect}{$R \eqdef R \cup \cexlocal{$e$}$}
   \label{l:cex.main.local}
   \lForEach {event $e \in C$ of $\acq$ or $\waa$ effect}{$R \eqdef R \cup \cexacquire{$e$}$}
   \label{l:cex.main.acquire}
   \lForEach {event $e \in C$ of $\wa$ effect}{$R \eqdef R \cup \cexwait{$e$}$}
   \label{l:cex.main.wait}
   \lForEach {event $e \in C$ of $\sigg$ or $\bro$ effect}{$R \eqdef R \cup \cexnotify{$C, e$}$}
   \label{l:cex.main.notify}
   \KwRet $R$ \;
   \label{l:cex.main.return}
}

\caption{Computing conflicting extensions: main algorithm}\label{a:cex.main}
\end{algorithm*}

The algorithm works by selecting some event~$e$ from~$C$ and trying to find
one or more events that are similar to $e$ and in conflict with $C$.
Except for events of $\loc$ effect, we try to \emph{reorder} an event regarding
those of its immediate causal predecessors that synchronize with other threads.
We call an event $e'$ an \emph{immediate causal predecessor} of $e$ iff
$e' < e$ and there is no other event $e'' \in [e]$ such that $e' < e'' < e$.

For events dealing with locks or condition variables, we aim to find different
possibilities for an action to synchronize with events on other threads.
For any event of $\acq$, $\wa$ or $\waa$ effect, we try to execute its action
\emph{earlier} than its causal predecessors by including only some of the
original immediate causes in any conflicting extension.
For any event $e$ of $\bro$ or $\sigg$ effect, however, we try to execute its
action \emph{before or after} certain other events in~$C$ that are concurrent
to~$e$.
An event of $\rel$ effect does not itself introduce a new synchronization with
another thread (compared to its sole immediate causal predecessor), as it has
to occur on the same thread that previously held the lock to be released. Thus,
we cannot create any reorderings for this kind of event.

Although similar to events of $\rel$ effect \wrt to lacking immediate causal
predecessors on other threads, we still have to compute conflicting extensions
for events of $\loc$ effect. Instead of reordering, we change the effect
of any event that represents branching on symbolic values.

Since the effect of an event determines how its conflicting extensions are to
be computed, \cexmain internally relays this task to corresponding functions
in \cref{a:cex.local,a:cex.acq,a:cex.w1,a:cex.notify}.
These functions share a similar structure, each
of them receives an event~$e$ (\cexnotify additionally takes the containing
(maximal) configuration $C$ as a parameter) and returns a set of conflicting
extensions, stored in variable~$R$ during computation. The union of these sets
will be the result of \cexmain which is returned in \cref{l:cex.main.return}.
Since our main algorithm, calling this function, deals with cutoff-free maximal
configurations, we add enabled cutoff events to $C$ before calling any of the
effect-specific algorithms (\cref{l:cex.main.cutoff}).

We have proofs that our algorithms compute exactly the set of conflicting
extensions of a given configuration, but unfortunately we did not have time to
transcribe them here.

\subsection{\texorpdfstring{Events of $\loc$ Effect}{Events of loc Effect}}

In our approach, branches on symbolic values correspond to two (or more)
distinct LTS actions of $\loc$ effect that differ in the thread statements they
represent. In \cref{fig:3}, for example, we branch over \texttt{x}, which
contains a symbolic value after the execution of \texttt{x~=~in()} (event 1).
Event 3 corresponds to the then- and event 9 to the else-branch that results
from executing \texttt{if(x~<~0)} after events 1 and 2.
Thus, both events share the same set of causes:
$\causes{3} = \causes{9} = \set{1, 2}$.
This can also be seen in \cref{fig:algo.unf.a}, along with the fact that these
events are in (immediate) conflict with each other.
Accordingly, a call to \cexlocal on event 3 will return a singleton set
containing event 9 (and vice versa).

In both figures, we can also see another pair of events resulting from the same
branching statement in the program, events 20 and 22.
However, these events follow a common set of causes that is different from that
of 3 and 9.

\begin{algorithm*}[t]
\DontPrintSemicolon

\Fn{\cexlocal{e}}{
   Assume that $e$ is $\tup{\tup{i, \tup{\loc, t}}, K}$ \;
   \label{l:cex.local.assume}
   $R \eqdef \emptyset$ \;
   \ForEach {statement  $t' \in T_i \setminus \set t$}
   {
      \label{l:cex.local.select1}
      $a \eqdef \tup{i,\tup{\loc, t'}}$ \;
      \If{$a$ is enabled at $\state K$}
      {
         \label{l:cex.local.select2}
         $R \eqdef R \cup \set{\tup{a, K}}$
         \label{l:cex.local.add}
      }
   }
   \KwRet $R$ \;
   \label{l:cex.local.return}
}

\caption{Computing conflicting extensions: $\loc$ events}\label{a:cex.local}
\end{algorithm*}

Function \cexlocal (\cref{a:cex.local}) receives an event~$e$ of $\loc$ effect
representing the execution of thread statement $t$ on thread $i$
(\cref{l:cex.local.assume}).
In \cref{l:cex.local.select1}, we iterate over alternative thread statements
from $T_i$, the set of thread statements restricted to thread $i$.
For each such thread statement $t_i \neq t$, we test whether the corresponding
action $a$ is enabled in the set of states reached by $K = \causes{e}$
(\cref{l:cex.local.select2}).
For any action that is enabled, we add an event that executes $a$ after the
causes $K$ to $R$ (\cref{l:cex.local.add}).
Finally, \cexlocal returns the set of conflicting extensions
(\cref{l:cex.local.return}).

\subsection{\texorpdfstring{Events of $\acq$ or $\waa$ Effect}{Events of acq or w2 Effect}}

The function \cexacquire in \cref{a:cex.acq} is identical to the one presented in
\cref{sec:algo.cex} for computing conflicting extensions for events of $\acq$
or $\waa$ effect.
Due to space constraints, however, we did not yet describe its function in
adequate detail.

\begin{algorithm*}[t]
\DontPrintSemicolon

\setlength{\columnsep}{2pt}

\Fn{\cexacquire{e}}{
   Assume that $e$ is $\tup{\tup{i, \tup{\acq, l}}, K}$ or
   $\tup{\tup{i, \tup{\waa, c,l}}, K}$ \;
   \label{l:cex.acquire.assume}
   $R \eqdef \emptyset$ \;
   $e_t \eqdef \lastof{$K, i$}$ \;
   \label{l:cex.acquire.et}
   \eIf {$\effect e = \tup{\acq, l}$}
   {
      $P \eqdef [e_t]$
      \label{l:cex.acquire.P.acq}
   }
   {
      $e_s \eqdef \lastnotify{$e, c, i$}$ \;
      $P \eqdef [e_t] \cup [e_s]$ \;
      \label{l:cex.acquire.P.waa}
   }
   $e_m \eqdef \lastlock{$P, l$}$ \;
   \label{l:cex.acquire.em}
   $e_r \eqdef \lastlock{$K, l$}$ \;
   \label{l:cex.acquire.er}

   \lIf {$e_m = e_r$}{\KwRet $R$}
   \label{l:cex.acquire.early}

   \uIf {$e_m = \bot \lor
      \effect{e_m} \in \set{\tup{\rel, l}, \tup{\wa, \cdot, l}}$}
   {
      \label{l:cex.acquire.check1}
      Add $\tup{h(e), P}$ to $R$ \;
      \label{l:cex.acquire.add1}
   }

   \ForEach {event $e' \in K \setminus (P \cup \set{e_r})$}
   {
      \label{l:cex.acquire.loop}
      \uIf{$\effect{e'} \in \set{\tup{\rel, l}, \tup{\wa, \cdot, l}}$}
      {
         \label{l:cex.acquire.check2}
         Add $\tup{h(e), P \cup [e']}$ to $R$ \;
         \label{l:cex.acquire.add2}
      }
   }
   \KwRet $R$ \;
   \label{l:cex.acquire.return}
}

\caption{Computing conflicting extensions: $\acq$ and $\waa$ events}\label{a:cex.acq}
\end{algorithm*}

To find conflicting extensions, the algorithm systematically constructs sets of
causes that include less than those of a given event of $\acq$ or $\waa$ effect
in order to execute the action of such an event \emph{earlier} \wrt its causes.
To explain how to do this, we first take a look at structural commonalities in
the set of causes for any event of $\acq$ or $\waa$ effect. Along the way, we
will also introduce some utility functions (see \cref{a:utility}) used in the
algorithm.

One event that we can always distinguish from any other in the set we will call
the \emph{thread predecessor}, which is the only event immediately preceding
the event at hand on the same thread. Our utility function \lastof{$C,i$}
(\cref{a:utility}) returns the only $<$-maximal event of thread~$i$ in a
configuration~$C$, so to retrieve the thread predecessor of a given event $e$
on thread $t$ we can call \lastof{$K,t$}.
As we are dealing with lock operations here, another important event is called
\emph{lock predecessor}. It can be retrieved by calling \lastlock{$C,l$}
(\cref{a:utility}) which returns the only $<$-maximal event that manipulates
lock~$l$ in~$C$ (an event of effect $\acq$, $\rel$, $\wa$ or $\waa$), or $\bot$
if no such event exists.
In the case of events of $\waa$ effect, we can also distinguish another event
from the rest of the causes: the \emph{notification predecessor}, which is an
event of $\sigg$ or $\bro$ effect immediately preceding a given $\waa$ event.
If $e$ is an event of $\waa$ effect, executed on thread $i$ and relating to
condition variable $c$, we can retrieve its notification predecessor using
\lastnotify{$e, c, i$} (\cref{a:utility}) which returns the only immediate
$<$-predecessor $e'$ of $e$ such that the effect of $h(e')$ is either
$\tup{\sigg,c,i}$ or $\tup{\bro,c,S}$ with $i \in S$.

\begin{algorithm*}[t]
\DontPrintSemicolon

\Fn{\lastof{C, j}}{
   \lIf{$j = 0$}{\KwRet $\bot$}
   \lIf{$C$ has no events from thread~$j$}{\KwRet $\bot$}
   \KwRet the $<$-maximal event of thread~$j$ in~$C$ \;
}

\BlankLine
\Fn{\lastlock{C, l}}{
	\KwAssert $C$ is conflict-free \;
   $S \eqdef $
   all events $e \in C$ \st $\effect{e}$ is
   $\tup{\acq, l}, \tup{\rel, l}, \tup{\wa, \cdot, l}$
   or
   $\tup{\waa, \cdot, l}$ \;
   \lIf{$S = \emptyset$}{\KwRet $\bot$}
   \KwRet the $<$-maximal event in~$S$
}

\BlankLine
\Fn{\lastnotify{e, c, i}}{
   \KwAssert $\effect{e}$ is $\tup{\waa, c, \cdot}$ \;
   \KwRet the only immediate $<$-predecessor $e'$ of $e$ such that
      $\effect{e'}$ is either $\tup{\sigg,c,i}$ or $\tup{\bro,c,S}$ with $i \in S$ \;
}

\BlankLine
\Fn{\concurrent{S}}{
   \KwAssert $S$ is conflict-free \;
   \KwRet \True iff for all $e_1$, $e_2 \in S$:
      not ($e_1 < e_2$) and not ($e_2 < e_1$) \;
}

\Fn{\outstandingw{C, c}}{

   $W \eqdef \emptyset$ \;
   \ForEach {thread id $k$ in the program}
   {
      $X \eqdef \set \bot$ \;
      $X \eqdef X \cup
         \set{e \in C \colon \tid e = k}$ \;
      $X \eqdef X \cup
         \set{e \in C \colon \effect e = \tup{\sigg, c, k}}$ \;
      $X \eqdef X \cup
         \set{e \in C \colon \effect e = \tup{\bro, c, S} \land k \in S}$ \;
      $e \eqdef \max_{<}(X)$ \;
      \lIf{$\effect e = \tup{\wa, c, \cdot}$}{Add $e$ to $W$}
   }
   \KwRet $W$ \;
}

\caption{Auxiliary functions.}\label{a:utility}
\end{algorithm*}

Conceptually, the algorithm tries to find lock releases (events of either
$\rel$ or $\wa$ effect) in the past, \ie the set of causes, that a potential
conflicting extension could use as its lock predecessor.

Function \cexacquire receives an event~$e$ of $\acq$ or $\waa$ effect
(\cref{l:cex.acquire.assume}).
Central to its mode of operation is a subset of $e$'s causal predecessors,
stored in a variable called $P$.
This set is (included in) the set of causes for any conflicting extension added
to $R$ (\cref{l:cex.acquire.add1,l:cex.acquire.add2}) and how it is constructed
poses the major difference between handling $\acq$ (\cref{l:cex.acquire.P.acq})
and $\waa$ (\cref{l:cex.acquire.P.waa}) effects.

For constructing $P$, the algorithm determines the thread predecessor, which is
stored in a variable called $e_t$ (\cref{l:cex.acquire.et}) and then
determines how to build $P$, based on the effect of $e$.
In case $e = \tup{\tup{i, \tup{\acq, l}}, K}$, it assigns $P$ the local
configuration of $e_t$ (\cref{l:cex.acquire.P.acq}).
This partitions the causes of $e$ such that $K \setminus P$ contains exactly
all events that were synchronized by the action of $e$, namely the locking of
$l$ following the execution of $e_t$.
In the other case, $e = \tup{\tup{i, \tup{\waa, c,l}}, K}$, the set $P$ also
includes the local configuration of $e_s$, the notification predecessor, along
with $[e_t]$ (\cref{l:cex.acquire.P.waa}).

We include $e_s$ into the set of causal predecessors (which will contain $P$
for any conflicting extension), because the notification predecessor is a
unique event waking up exactly $e_t$ (which will be of $\wa$ effect iff
$e$ is of $\waa$ effect).
Since we set out not to modify any event in $[e_t]$, there cannot be any
conflicting extension that excludes $e_s$ from its set of causes.
Thus, for the following steps, we can essentially treat any event of $\waa$
effect as if it were one of $\acq$ effect, since we only have to concern
ourselves with the reacquisition of the lock done as part of its action.
The only exception to this is making sure that conflicting extensions for $e$
always executes the same action $h(e)$
(\cref{l:cex.acquire.add1,l:cex.acquire.add2}).

In addition to $P$, the algorithm determines two lock events: $e_r$, which is
the lock predecessor of $e$ (\cref{l:cex.acquire.er}), and $e_m$, which is the
$<$-maximal lock event in $P$ (\cref{l:cex.acquire.em}). For the former, we
know that it has to be a lock release (an event of either $\rel$ or $\wa$
effect), while $e_m$ might not exist ($\bot$) or be any event of $\acq$,
$\rel$, $\wa$ or $\waa$ effect.
If $e_m = e_r$, then \cexacquire can abort the search and return an empty set
(\cref{l:cex.acquire.early}), since we know that outside of $P$, there exists no
further relevant lock event.

Otherwise, the algorithm determines whether the lock can be acquired from $e_m$
(\cref{l:cex.acquire.check1}). This is true when $e_m$ is either of $\rel$ or
$\wa$ effect, but also when $e_m = \bot$. In the latter case, there is no prior
acquisition on the same lock in $P$, constituting the need for a release.
In both cases, the lock can be acquired from $e_m$, so the algorithm adds a
conflicting extension to $R$ that only uses $P$ as its causes
(\cref{l:cex.acquire.add1}).
After handling special cases involving $e_m$, the algorithm searches for other
possible lock predecessors in $K$, outside of $P$ and excluding $e_r$
(\cref{l:cex.acquire.loop}).
For any event $e'$, releasing the relevant lock (\cref{l:cex.acquire.check2}),
we add a conflicting extension to $R$ (\cref{l:cex.acquire.add2}).
Note, that we include the local configuration of $e'$ to the set of causes
since $e'$ might have causes outside of $P$.
We exclude $e_r$ in \cref{l:cex.acquire.loop} as otherwise we would (at least)
add $e$ to $R$, which is (trivially) not in conflict with $[e]$.
Finally, \cexacquire returns $R$ (\cref{l:cex.acquire.return}).

There is one situation not covered by the algorithm as we just described it,
which are configurations caught in a deadlock while trying to acquire a lock.
That is, in an LTS state reached by some configuration $C$, the only available
action for some thread is not enabled and this action is either of $\acq$ or
$\waa$ effect.
Even though there cannot be an event $e$ resulting from the execution of such
an action in $C$, there might be a viable event (which is in fact a conflicting
extension to $C$) that can be found akin to \cexacquire.
To find such events, we take the following steps for each deadlocked thread:
We determine $e_m$ and $P$ like before, albeit modified to account for the lack
of $e$. Afterwards, we do a search within the set $(C \setminus P) \cup e_m$,
analogously to \crefrange{l:cex.acquire.loop}{l:cex.acquire.add2}.

\subsection{\texorpdfstring{Events of $\wa$ Effect}{Events of w1 Effect}}

When concerned with computing conflicting extensions for events of $\waa$
effect, we briefly discussed the interplay of such events with those of $\wa$
effect. We concluded that any reordering not caused by $\waa$'s interaction
with locks hinges on the related event of $\wa$ effect.
In return, we do not have to consider the lock-releasing part of the $\wait$
operation for $\wa$ events.

Condition variables, in contrast to locks which can only ever be manipulated
by one event at a time, allow for concurrent operations.
To account for this, we introduce another auxiliary function.
Given a set $S \subseteq E$ of conflict-free events,
a call to \concurrent{$S$} (\cref{a:utility}) returns \texttt{true} iff there
is no causal dependence between any two events in $S$.

\begin{algorithm*}[t]
\DontPrintSemicolon

\Fn{\cexwait{e}}{
   Assume that $e$ is $\tup{\tup{i, \tup{\wa, c, l}}, K}$ \;
   \label{l:cex.wait.assume}
   $R \eqdef \emptyset$ \;
   $e_t \eqdef \lastof{$K, i$}$ \;
   \label{l:cex.wait.et}
   $X \eqdef K \setminus [e_t]$ \;
   \label{l:cex.wait.X}
   Let $X'$ be the subset of~$X$ that contains all events $e'$ such
   that: $\effect{e'}$ is
   $\tup{\sigg, c, 0}$ or
   $\tup{\bro, c, \emptyset}$ or
   $\tup{\bro, c, S}$ with $i \not\in S$. \;
   \label{l:cex.wait.Xp}
   \ForEach {set $M \subseteq X'$ such that \concurrent{M} and $[M] \cup [e_t] \ne K$}
   {
     \label{l:cex.wait.loop}
     Add $\tup{h(e), [M] \cup [e_t]}$ to $R$ \;
     \label{l:cex.wait.add}
   }
   \KwRet $R$ \;
   \label{l:cex.wait.return}
}

\caption{Computing conflicting extensions: $\wa$ events}\label{a:cex.w1}
\end{algorithm*}
Function \cexwait (\cref{a:cex.w1}) receives an event~$e$ of $\wa$ effect
(\cref{l:cex.wait.assume}) and determines its thread predecessor $e_t$
(\cref{l:cex.wait.et}). The local configuration of $e_t$ is subtracted from the
set of $e$'s causes (\cref{l:cex.wait.X}) and the resulting set is filtered
such that only some specific events of $\sigg$ and $\bro$ effect remain
(\cref{l:cex.wait.Xp}).

The events that make up this filtered set, $X'$, are certain events from
$\causes e$ which are
lost signals, or
events of $\bro$ effect that did not notify $e$'s thread (including lost
broadcasts).
Specifically, each selected event corresponds to the occurrence of an action
that can notify an event of $\wa$
effect which has $e_t$ as thread predecessor.
Thus, the underlying operation from any event in $X'$ would have resulted in a
successful notification, had it happened right after, and not (as is the case)
before $e$'s action.
This is because of how the involved actions depend on each other (see
\cref{app:indep.programs}). For instance, a lost signal or broadcast from $X'$
can only happen before (causally preceding) $e$, and no event representing the
same action can happen before an event of $\sigg$ or $\bro$ effect includes
$e$ in its set of causes (indicating a successful notification).

In \cref{l:cex.wait.add}, conflicting extensions are added to $R$, each
corresponding to some set $M \subseteq X'$, \ie with
$\causes{M} \cup \causes{e_t}$ as causes.
Each $M$ is chosen such that its elements are not causally related (using
\concurrent) and the resulting event does not equal $e$
(\cref{l:cex.wait.loop}).
As each conflicting extension excludes at least one event from $X'$ from
its causes (compared to $e$), an event of $\waa$ effect following it might
receive its notification earlier (\ie with fewer events in its local
configuration) than one following $e$.
Once all suitable subsets were found, \cexwait returns $R$
(\cref{l:cex.wait.return}).

\begin{lemma}
\Cref{l:cex.main.wait} of a \cexmain{C} function call (\cref{a:cex.main})
produces exactly all conflicting extensions $e'$ of a configuration $C$ such
that $e' = \tup{\tup{\cdot, b}, \cdot}$ and $b = \tup{\wa, \cdot, \cdot}$ if
$C$ is maximal.
\end{lemma}

\begin{proof}
As \cexmain{C} computes conflicting extensions that have \wa{} effect using a
loop, calling \cexwait (\cref{a:cex.w1}) on all events of \wa{} effect in $C$,
proving this lemma reduces to proving that the following assertion is valid:
\[\bigcup_{\mathclap{e \in C_{\wa}}}{\cexwait{e}} =
\set{e' \in \cex{C} \colon e' \text{ has } \wa \text{ effect}}\]
where $C_{\wa} \eqdef \set{e \in C \colon e \text{ has } \wa \text{ effect}}$.
We do this by first showing that every event that is returned by
\cexwait{e} for any event $e \in C$ of \wa{} effect is a conflicting extension
of $C$ of the same effect (left-to-right below).
In a second step, we will show that every \wa{} event that is contained in
$\cex{C}$ will eventually be computed by some call to \cexwait in the loop
(right-to-left below).
In the following, let $C$ be a maximal configuration and $e \in C$ an event of
the form $e \eqdef \tup{\tup{i, \tup{\wa, c, l}}, K}$
(corresponding to \cref{l:cex.wait.assume}).

\textbf{Left-to-right.}
We show that any event $e'$ returned in \cref{l:cex.wait.return} of \cexwait{e}
satisfies $e' \in \cex{C}$ with $h(e')$ being
$\tup{\cdot, \tup{\wa, \cdot, \cdot}}$.
Any event $e'$ returned by the algorithm is constructed in
\cref{l:cex.wait.add} firing action $h(e)$ after the causes $[M] \cup [e_t]$.
We therefore need to prove the following:
\begin{itemize}
\item
   The pair $e' \eqdef \tup{h(e), [M] \cup [e_t]}$ is an event, and\eqtag{eq:is.event}
\item
   $h(e')$ is of the form $\tup{\cdot, \tup{\wa, \cdot, \cdot}}$.\eqtag{eq:form.w1}
\end{itemize}
Note that~\cref{eq:form.w1} above is true simply because \cexmain calls \cexwait
only for \wa{} events (\cref{l:cex.main.wait} of \cref{a:cex.main}) and
\cexwait (\cref{a:cex.w1}) only returns events of the same effect.
In the sequel we focus on~\cref{eq:is.event}.
To show that $e'$ is an event, by~\cref{def:unfd.event} of~\cref{def:unfd},
we need to show that:
\begin{enumerate}
\item Its causes $\causes{e'}$ are a configuration,
\item its action $h(e')$ is enabled at $\state{\causes{e'}}$, and
\item $\lnot(h(e') \indep h(e''))$ holds for every $<$-maximal event $e''$ in $\causes{e'}$.
\end{enumerate}
In the sequel we prove each one of these three items:
\begin{enumerate}
\item We show that $\causes{e'} = [M] \cup [e_t]$ is a configuration.
Because $e$ is an event, we know that $K$ is a configuration.
As such, none of its events are in conflict with each other.
Since both $e_t$ and $M$ are taken from $K$ (\cref{l:cex.wait.et},
\crefrange{l:cex.wait.X}{l:cex.wait.loop}), we see that $[M] \cup [e_t]$ is
conflict-free.
Forming the union of two local configurations moreover implies that the result
is causally closed.
Together, this means that $\causes{e'} \eqdef [M] \cup [e_t]$
is a configuration.

\item
We now show that action $h(e')$ is enabled at
$s' \eqdef \state{[M] \cup [e_t]}$.
As $h(e') = h(e)$, the action is clearly enabled at $s \eqdef \state{K}$.
To show that it is also enabled at $s'$, we prove that rule
\rulewa{} (of \cref{fig:rules}) can be applied to state~$s'$:
\begin{itemize}
\item $\tup{i, n, n', \tup{\wait, c,l}, r} \in T$:
As $e'$ executes the same action as $e$, we know that there exists such a
thread statement in $T$, giving rise to $h(e)$.
This thread statement executes the operation $\tup{\wait, c, l}$ on thread $i$,
while this thread is at program location $n$.
\item $p(i) = n$ holds at state $s'$:
Observe that $M$ is a subset of $X'$, and so, a subset of $X$.
As a result it only contains events of threads other than~$i$
(because all events from thread~$i$ are causal predecessors of~$e_t$,
and $X$ is constructed to avoid those events, in \cref{l:cex.wait.X}).
Furthermore, as we include~$[e_t]$, we have that $e_t$ is not only included in
$\causes{e'}$, but also the maximal event on thread $i$.
Because $e_t$ is the $<$-maximal event on thread $i$ in both $K$ and
$\causes{e'}$, both $s$ and $s'$ are at the same program location~$n$ for
thread~$i$ (also corresponding to the above thread statement).
\item $u(l) = i$ holds at state $s'$:
Since $h(e)$ is enabled at $s$, we know that $u(l) = i$ at $s$.
This means that there must be some event in~$K$ that acquires lock~$l$ for
thread $i$. Let this event be known as~$e_m$.
Note that~$e_m$ must be an event of thread~$i$.
This is because, according to the semantic rules in~\cref{fig:rules},
the only way $u(l)$ can be set to~$i$ is by having
$h(e_m) = \tup{i, \tup{\acq, l}}$ or
$h(e_m) = \tup{i, \tup{\waa, c, l}}$,
for some condition variable~$c \in \conds$.
Consequently, $e_m \in [e_t]$, because all events of thread~$i$ in~$K$ are
in $[e_t]$, as we already justified above.
Now, note that only events of effect
$\tup{\acq, l}$, $\tup{\rel, l}$, $\tup{\wa, \cdot, l}$ or $\tup{\waa, \cdot, l}$
can modify $u(l)$.
Since $X$ does not contain any event of thread $i$ (\cref{l:cex.wait.X}),
we have that
$X$ cannot contain any event of those effects.
This is because if~$X$ contained an event of thread~$j \ne i$ of such effect,
that event would be a $<$-successor of $e_m$, and it would set
$u(l) = 0$ or $u(l) = j$, and no event in~$X$ could make~$u(l) = i$ again,
because $X$ does not contain events from thread $i$.
Consequently, state $s'$ makes $u(l) = i$ true.
\end{itemize}

\item
Finally, we show that $\lnot(h(e') \indep h(e''))$ holds for every $<$-maximal
event $e''$ in $\causes{e'}$.
The causes of $e'$ are defined as $[M] \cup [e_t]$ (\cref{l:cex.wait.add}).
Consequently $\max_{<}(\causes{e'}) \subseteq M \cup \set{e_t}$.
We now show that all events in $M \cup \set{e_t}$ are dependent on~$h(e')$,
which implies what we want to prove.
By construction~$h(e_t)$ and~$h(e') = h(e)$ are actions of thread~$i$, so they
are dependent.
Also, by construction, any event in~$M$ is an event in~$X'$.
Set~$X'$ (\cref{l:cex.wait.Xp}) only contains events of the following effects:
$\tup{\sigg, c, 0}$, $\tup{\bro, c, \emptyset}$ or $\tup{\bro, c, S}$
with $i \notin S$.
According to \cref{tab:indepp} from \cref{def:indepp}, any action of either one
of these effects (for any thread) is dependent on $h(e')$.
Thus, the actions of all events in $M$ are always dependent on $h(e')$.
\end{enumerate}
Now that we have established that $e'$ is indeed an event, it remains to show
that $e'$ is in $\cex{C}$.
Since $\causes{e'}$ only contains events from $\causes{e} \subseteq C$, we have
that $e' \in \ex{C}$, \ie all causes of $e'$ are contained in $C$.
To show that an event $e' \in \ex{C}$ is also in $\cex{C}$, we have to show
that it is in conflict with some other event in $C$ (see \cref{sec:algo.cex}).
We will show that $e'$ is in conflict with $e$, \ie $e' \cfl e$.
For this, the following needs to hold according to the definition given in
\cref{sec:algo.unfsem}:
\begin{enumerate}
	\item $e \neq e'$: Since we explicitly prevent the creation of any
        conflicting extension with $K = \causes{e}$ as its set of causes in
        \cref{l:cex.wait.loop}, we have $e \neq e'$.
	\item $\lnot(e < e')$: The set of causes for $e'$ comes from $\causes{e}$,
        so clearly $e \notin \causes{e'}$.
	\item $\lnot(e' < e)$: By contradiction, assume that
        $e' \in \causes{e}$. As $[e_t]$ is included entirely in the causes of
        $e'$, we know that $e' \notin [e_t]$, but instead $e' \in X$.
        Because $e_t$ is defined as the $<$-maximal predecessor of $e$ on
        thread $i$, event $e'$ has to be on some other thread $j$
        (\ie $i \neq j$). We also know that at $\state{[e_t]}$, thread $i$ is
        holding lock $l$. However, $e'$ requires lock $l$ being held by thread
        $j$ which is not possible without another event on thread $i$ following
        $e_t$ (contradiction).
	\item $h(e) \depen h(e')$: Both $h(e)$ and $h(e')$ are actions of the same
        thread.
\end{enumerate}

\textbf{Right-to-left.}
We show that for any event~$e' \in \cex C$, if~$e'$ is of $\wa$ effect,
then~$e'$ is eventually returned by a call to \cexwait.
Let~$e'$ be an event of thread~$i$.
We proceed in two steps:
\begin{enumerate}
\item
   We show that~$C$ contains some event~$e$ such that~$h(e) = h(e')$.\eqtag{eq:step1}
\item
   We show that a call to \cexwait{$e$} returns~$e'$.\eqtag{eq:step2}
\end{enumerate}
Note that if~\cref{eq:step1} holds,
then such~$e$ will be found at \cref{l:cex.main.wait} of
\cexmain, and so, by~\cref{eq:step2},
the call made there to \cexwait{$e$} will compute and
return~$e'$, what we wanted to prove.

We start by proving \cref{eq:step1}.
We first establish that there exists some event~$e_t$ of thread $i$ that is
causally preceding $e'$.
Such an event always exists in the case of a $\wa$ event, because~$\wa$ releases
a mutex previously acquired by some other event of thread~$i$.
By definition we have that~$e_t$ is in~$C$.
We now show~\cref{eq:step1} by breaking it down in two statements that we prove
separatedly.
We show that there exists some event $e$ such that:
\begin{itemize}
\item $e \in C$, and\eqtag{eq:e.in.c}
\item $h(e) = h(e')$.\eqtag{eq:same.h}
\end{itemize}

To prove these statements we define (below) some set~$C'$ and prove that~$h(e')$
is enabled at at~$C'$. We later define the causes of~$e$ from~$C'$.

First, let us define the notation $\future{\hat e}$, for some event~$\hat e$, as
the set containing all events from ${\unfpindep}$ that are causally dependent
on~$\hat e$, as well as~$\hat e$ itself.
Similarly, for a set $S$ of events, we define
$\future{S} \eqdef \bigcup_{\hat e \in S} \future{\hat e}$.

We start the proof by defining
$$C' \eqdef C \setminus \future{S}$$
with $S \eqdef \set{e \in C \colon e_t < e \land \tid{e} = i}$.
This way, $C'$ excludes exactly all events from $C$ that depend on any event
of thread $i$ succeeding $e_t$.

We now show by induction that $h(e')$ is enabled at~$C'$.
Let us make the following definitions.  Let
$$\sigma \eqdef e_1 \dots e_n \in \inter{C' \setminus \causes{e'}}$$
be any interleaving of the events in $C' \setminus \causes{e'}$.
Note that events $e_1$, \dots, $e_n$ are numbered in such a way
$e_j < e_k \implies j < k$.
Furthermore, let:
\begin{align*}
C_0 &\eqdef \causes{e'} \\
   C_j &\eqdef C_{j-1} \cup \set{e_j} \quad \text{for } 1 \le j \le n \\
\end{align*}
Observe that, by construction, $C_n = C'$ and that
$e_j$ is enabled at $C_{j-1}$ for $1 \le j \le n$.
We now show by induction on~$j$ that~$h(e')$ is enabled at~$C_j$,
for~$0 \le j \le n$.

\begin{itemize}
\item \textbf{Base case:} $j = 0$. $C_j = C_0 = \causes{e'}$.
Clearly, configuration $C_0$ enables $h(e')$,
as otherwise $e'$ would not be an event.

\item \textbf{Inductive case:}
Assume $\state{C_{j-1}}$ enables $h(e')$.
By induction, we also have that $\state{C_{j-1}}$ enables $h(e_j)$.
If $h(e_j) \indep h(e')$ then by definition of independence (\cref{app:indep})
they commute and so $\state{C_j}$ also enables $h(e')$.
If however, $h(e_j) \depen h(e')$, then clearly $k \eqdef \tid{e_j} \neq i$,
because C' does not contain any event of thread $i$.
For $h(e')$ to remain enabled, $u(l) = i$ needs to hold
(as it does because $\state{C_{j-1}}$ already enables the action) and
$p(i)$ must not be changed (see \cref{tab:indepp}).
Because there is no rule that changes $p$ for a thread other than that of the
action, we only need to consider changes to $u$.
According to \cref{def:indepp}, we have the following cases for $h(e_j$):
\begin{itemize}
    \item $\tup{k, \tup{\bro, c, \cdot}}$ or $\tup{k, \tup{\sigg, c, 0}}$.
        By rule \rulebro{} or \rulesiggl{}, respectively,
        $u(l)$ remains untouched, so $\state{C_j}$ enables $h(e')$.
    \item $\tup{k, \tup{\sigg, c, i}}$.
        Rule \rulesigg{} does not modify $u$, so $h(e')$ remains enabled.
    \item $\tup{k, \tup{\acq, l}}$. Action $h(e_j)$ cannot be enabled at
        $\state{C_{j-1}}$ because $u(l) = i$ at $\state{C_{j-1}}$ and
        so thread $k$ cannot take the lock.
    \item $\tup{k, \tup{\rel, l}}$. Impossible, because we have $u(l) = i$,
        but rule \rulerel{} requires $u(l) = k$ to release lock $l$
        on thread $k$.
    \item $\tup{k, \tup{\wa, c', l}}$. Impossible, because we have $u(l) = i$,
        but rule \rulewa{} requires $u(l) = k$ to release lock $l$
        on thread $k$
    \item $\tup{k, \tup{\waa, c', l}}$. Impossible, because we have $u(l) = i$,
        but rule \rulewaa{} requires $u(l) = 0$ to acquire lock $l$
        on thread $k$.
\end{itemize}
In all cases~$h(e')$ is enabled at~$\state{C_j}$.
\end{itemize}

Consequently, $h(e')$ is enabled at~$\state{C'}$.
We now show that there exists some event~$e$ enabled at $C'$ such that
$h(e) = h(e')$. This will show~\cref{eq:same.h}.

The idea is simple: if all maximal events of~$C'$ are dependent
with~$h(e')$, then $e \eqdef \tup{h(e'), C'}$ is the event we are searching for.
Otherwise, $C'$ contains some $<$-maximal event $\hat e$ such that $h(\hat e)$
is independent of $h(e')$.
Then define $C'' \eqdef C' \setminus \set{\hat e}$.
Since~$h(e)$ and~$h(\hat e)$ are independent, $\state{C''}$ enables $h(e')$.
If all maximal events of~$C''$ are dependent with~$h(e')$, then
$e \eqdef \tup{h(e'), C''}$ is the event we are searching for.
Otherwise this argument can be iterated again and again, reducing the size
of~$C'$ on every step, until we find the event~$e$ that we are searching for.

We have proved that there exists some event $e$ such that $e \in \en{C'}$
and $h(e) = h(e')$.  We now show that $e \in C$, therefore
proving~\cref{eq:e.in.c}.
Let $\hat{e}$ be the $<$-minimal event in $S$.
By definition, $\tid{\hat{e}} = i$ and $\hat{e} \in \en{C'}$.
Since $\state{C'}$ enables $h(\hat{e})$ and $h(e)$ and
$\tid{h(e')} = \tid{h(e)}$ by \cref{def:well-formed},
we get that $h(\hat{e}) = h(e)$.
By \cref{lemma:p14} then $e = \hat{e}$.

We have established~\cref{eq:e.in.c,eq:same.h}. Since~$e \in C$,
when \cref{a:cex.main} is called on~$C$,
\cref{l:cex.main.wait} will find~$e$ and call \cexwait{$e$}.
We now need to show~\cref{eq:step2}.
That is, we need to show that this call will eventually add~$e'$ to $R$
(\cref{l:cex.wait.add} of~\cref{a:cex.w1}) and return it
(\cref{l:cex.wait.return}).

The key ideas to do this are as follows.
One should establish that the event~$e_t$ defined above is identical with
the one selected in \cref{l:cex.wait.et} of \cref{a:cex.w1}.
Then, one should show (by contradiction, considering all other effects) that the
effects of $\max_{<}(\causes{e'} \setminus \set{e_t})$ are exactly those that
we select for $X'$ in \cref{l:cex.wait.Xp}.
One should then show that $\max_{<}(\causes{e'})$ is concurrent and
$\max_{<}(\causes{e}) \neq \max_{<}(\causes{e'})$.
From this, it follows that $\causes{e'}$ will be considered in the loop of
\cref{l:cex.wait.loop} and \cexwait{$e$} will return event~$e'$, what we wanted
to show.
\end{proof}

\begin{lemma}\label{lemma:p14}
Let $C$ be a finite configuration of $\unfpindep$.
Let $a \in A$ be an action enabled at $\state C$.
Then there exists exactly one event $e \in \en C$ such that $h(e) = a$.
\end{lemma}
\begin{proof}
Let $\set{e_1, \ldots, e_n}$ be all $<$-maximal events of~$C$.
If $a$ is dependent (\wrt $\indep$) with all of them, then
$e \eqdef \tup{a, C}$ is an event (by \cref{def:unfd}, item 2) and it
is the event we are searching.
Otherwise assume that $e_i$ is such that $h(e_i) \indep a$.
Then let $C' \eqdef C \setminus \set{e_i}$ be the result of removing $e_i$
from~$C$ and try again. If all maximal events of $C'$ are now dependent with~$a$
then $\tup{a, C'}$ is the event we are searching. If not, then iterate this reasoning again and again on $C'$ to ``peel off'' all
events independent with~$a$. The resulting configuration~$C''$ will necessarily
be smaller than~$C$. Since~$C$ is of finite size, this procedure will terminate,
it will produce a uniquely-defined event $\tup{a, C''}$,
and such event will be enabled at~$C$.
\end{proof}

\subsection{\texorpdfstring{Events of $\sigg$ or $\bro$ Effect}{Events of sig or bro Effect}}

To find conflicting extensions, the algorithms \cexacquire and \cexwait pick a
subset of a given event's causes.
For events of $\sigg$ and $\bro$ effect, however, we will also consider events
that happen concurrently to a given event.
Thus, \cexnotify (\cref{a:cex.notify}) receives a (maximal) configuration $C$
in addition to an event $e$ of $\sigg$ or $\bro$ effect
(\cref{l:cex.notify.assume}) with $e \in C$.
After determining the thread predecessor as usual (\cref{l:cex.notify.et}),
instead of using subsets of $K \setminus [e_t]$ as basis for conflicting
extensions, the algorithm prepares a set that additionally includes events
concurrent to $e$.
This set $X$ is constructed by subtracting both events that causally precede
$e_t$ and events that succeed $e$ (\cref{l:cex.notify.X}).
The latter are denoted by $\future{e}$, which is the set containing all events
from $\unfpindep$ that are causally dependent on $e$ as well as $e$ itself.
In other words, $X$ is the set of all events in the configuration that is
neither $e$, causally dependent on $e$ nor in $e_t$'s local configuration.
To avoid adding a conflicting extension that equals $e$, we additionally
prepare $M$ to be the set of immediate causal predecessors of $e$
(\cref{l:cex.notify.M}).

\begin{algorithm*}[!t]
\DontPrintSemicolon

\Fn{\cexnotify{C, e}}{
   Assume that $e$ is
   $\tup{\tup{i, \tup{\sigg, c, j}}, K}$ or
   $\tup{\tup{i, \tup{\bro, c, S}}, K}$ \;
   \label{l:cex.notify.assume}

   $e_t \eqdef \lastof{$K, i$}$ \;
   \label{l:cex.notify.et}

   $X \eqdef C \setminus (\future e \cup [e_t])$ \;
   \label{l:cex.notify.X}
   $M \eqdef \max_{<} (K \setminus [e_t])$ \;
   \label{l:cex.notify.M}
   $R \eqdef \emptyset$ \;

   \If{$\effect{e} = \tup{\sigg, c, \cdot}$}
   {
      \label{l:cex.notify.sig}
      $W \eqdef
         \set{e' \in X \colon \effect{e'} = \tup{\wa, c, \cdot}}
         \cup
         \outstandingw{$[e_t], c$}$ \;
         \label{l:cex.notify.sig.W}

      \ForEach {$e_w \in W$}
      {
         \label{l:cex.notify.sig.loop}
         \lIf{$e_w = \lastof{$K, j$}$}{\KwContinue}\label{l:cex.notify.sig.continue}
         Add $\tup{\tup{i, \tup{\sigg, c, \tid{e_w}}}, [e_t] \cup [e_w]}$ to~$R$ \;
         \label{l:cex.notify.sig.add}
      }
   }
   \Else
   {
      \label{l:cex.notify.bro}
      Let $\mathcal A$ be the collection of sets $M' \subseteq X$ such that: \;
         \label{l:cex.notify.bro.sets}
         \Indp
         -- \concurrent{M'} \;
         \label{l:cex.notify.bro.conc}
         -- $M \ne M'$ \;
         \label{l:cex.notify.bro.MneMp}
         -- {For all $e' \in M'$ either $\effect{e'} = \tup{\wa, c, \cdot}$ \;
            \label{l:cex.notify.bro.conda}
            \Indp\Indp or $\effect{e'} = \tup{\sigg, c, k}$ with $k \ne 0$ \; \Indm\Indm
            \label{l:cex.notify.bro.condb}
         }
         \Indm
      \ForEach {$M' \in \mathcal A$}
      {
         \label{l:cex.notify.bro.loop}
         $W \eqdef \outstandingw{$[M'] \cup [e_t], c$}$ \;
         \label{l:cex.notify.bro.W}
         \If{$W \ne \emptyset$}
         {
            \label{l:cex.notify.bro.if}
            Add $\tup{\tup{i, \tup{\bro, c, \tid W}}, [M’] \cup [e_t]}$ to $R$
            \label{l:cex.notify.bro.add}
         }
      }
   }

   Let $\mathcal B$ be the collection of sets $M' \subseteq X$ such that: \;
      \label{l:cex.notify.lost}
      \Indp
      -- \concurrent{M'} \;
      \label{l:cex.notify.lost.conc}
      -- $M \ne M'$ \;
      \label{l:cex.notify.lost.MneMp}
      -- $\outstandingw{$[M'] \cup [e_t], c$} = \emptyset$ \;
      \label{l:cex.notify.lost.outstanding}
      -- {Either $M' = \set{e'}$ and $\effect{e'} = \tup{\bro, c, S}$, with $S \ne \emptyset$, \;
         \label{l:cex.notify.lost.conda}
         \Indp\Indp or $\forall e' \in M'$ we have $\effect{e'} = \tup{\sigg, c, k}$ with $k \ne 0$ \;\Indm\Indm
         \label{l:cex.notify.lost.condb}
      }
      \Indm

   Let $t \eqdef
   \begin{cases}
      \tup{\sigg, c, 0} & \text{if } \effect{e} = \tup{\sigg, c, \cdot} \\
      \tup{\bro, c, \emptyset} & \text{otherwise}
   \end{cases}
   $ \;
   \label{l:cex.notify.lost.t}

   \ForEach {$M' \in \mathcal B$}
   {
      \label{l:cex.notify.lost.loop}
      Add $\tup{\tup{i, t}, [M’] \cup [e_t]}$ to $R$
      \label{l:cex.notify.lost.add}
   }

   \KwRet $R$;
   \label{l:cex.notify.return}
}

\caption{Computing conflicting extensions: $\sigg$ and $\bro$ events}\label{a:cex.notify}
\end{algorithm*}
The algorithm is divided into three sections, adding conflicting extensions in
the form of successful signals
(\crefrange{l:cex.notify.sig}{l:cex.notify.sig.add}),
successful (non-empty) broadcasts
(\crefrange{l:cex.notify.bro}{l:cex.notify.bro.add})
and lost notifications
(\crefrange{l:cex.notify.lost}{l:cex.notify.lost.add}).
For the former two, the algorithm only executes the section that corresponds
to the effect of $e$, while the last (for lost notifications) is always
executed.

Common to all three sections is that they need to reason about events of $\wa$
effect. For that, we will introduce another utility function from
\cref{a:utility}: With \outstandingw{C, c}, we compute the (concurrent) set of
$\wa$ events (on a condition variable $c$) that are not matched by a
corresponding event of $\sigg$ or $\bro$ effect in a configuration $C$.
In other words, we receive the set of $\wa$ events corresponding to $c$ that
are still awaiting a notification in $C$.

We will start the description with the section concerning successful signals
(\crefrange{l:cex.notify.sig}{l:cex.notify.sig.add}).
As already established, this part of the function is only executed when $e$ is
of $\sigg$ effect (\cref{l:cex.notify.sig}).
It will add an event to $R$ in form of a signal corresponding to each one $\wa$
event (corresponding to the same condition variable $c$) that can be notified
following $e_t$.
To this end, a set $W$ is built (\cref{l:cex.notify.sig.W}), containing all
events of $\wa$ effect that are concurrent to $e$ or outstanding in $[e_t]$.
For each event $e_w$ in $W$ (\cref{l:cex.notify.sig.loop}), we add an event of
$\sigg$ effect to $R$ (\cref{l:cex.notify.sig.add}), notifying the thread of
$e_w$, which is denoted by $\tid{e_w}$.
Note however, that we do not add a conflicting extension for the $\wa$ event
that was already notified by $e$ (\cref{l:cex.notify.sig.continue}).

In case $e$ is of $\bro$ effect, the algorithm tries to add successful
broadcasts as conflicting extensions in
\crefrange{l:cex.notify.bro}{l:cex.notify.bro.add}.
As broadcasts, in contrast to signals, do not only notify a single event of
$\wa$ effect, we need to reason about sets of those events.
We chose candidate subsets from $X$ in
\crefrange{l:cex.notify.bro.sets}{l:cex.notify.bro.condb}, such that they
adhere to three conditions.
First, any $M' \subset X$ has to be a concurrent set
(\cref{l:cex.notify.bro.conc}).
We do this as only the maximal event of each thread determines what events are
included in $[M']$, \ie additionally including any preceding events in $M'$
does not change $[M']$.
This also easily allows us to compare $M$ to $M'$ and mandate they should not
be equal (\cref{l:cex.notify.bro.MneMp}), to avoid adding the event that equals
$e$ to $R$.
Finally, we are only interested in sets that only contain events of $\wa$
(\cref{l:cex.notify.bro.conda}) or successful $\sigg$ effects
(\cref{l:cex.notify.bro.condb}).
Events representing successful signals are included into the set of causes as
they are dependent (see \cref{app:indep.programs}) on events of $\bro$ effect
in the following way: The action must have happened before, as otherwise the
broadcast would have included the notified thread in its set, eliminating the
need for the signal.
For each suitable subset (\cref{l:cex.notify.bro.loop}), we determine the
outstanding events of $\wa$ effect in $[M'] \cup [e_t]$
(\cref{l:cex.notify.bro.W}) and if the resulting set ($W$) is non-empty
(\ie the resulting broadcast will not be lost, \cref{l:cex.notify.bro.if}),
we add a conflicting extension to $R$ (\cref{l:cex.notify.bro.add}).
This new event is a broadcast notifying exactly the set of threads for which
there is an outstanding event of $\wa$ effect in $W$ (retrieved by $\tid{W}$).

Finally, we look at the generation of conflicting extensions representing lost
notifications (\crefrange{l:cex.notify.lost}{l:cex.notify.lost.add}).
Again, we determine candidate subsets, with the first two conditions matching
those of the broadcast section
(\cref{l:cex.notify.lost.conc,l:cex.notify.lost.MneMp}).
Then, however, we mandate that there are no outstanding events of $\wa$ effect
(\cref{l:cex.notify.lost.outstanding}), as the resulting event will not notify
any thread.
We also restrict $M'$ to contain either a single event of $\bro$ effect,
notifying at least one thread (\cref{l:cex.notify.lost.conda}), or a number of
successful $\sigg$ events (\cref{l:cex.notify.lost.condb}).
These two possibilities cover all the immediate predecessor sets a lost
notification can have (see \cref{app:indep.programs}).
In \cref{l:cex.notify.lost.t}, we determine the effect of the resulting
conflicting extensions, based on the effect of $e$.
Then, for each suitable subset (\cref{l:cex.notify.lost.loop}), we add a
conflicting extension to $R$ (\cref{l:cex.notify.lost.add}).

As always, \cexnotify returns $R$, the set of conflicting extensions found,
once it is complete (\cref{l:cex.notify.return}).

\section{Main Algorithm}\label{app:main}

In \cref{sec:algo.exploring} we presented a simplified version of the main
algorithm proposed in this work.
We now present the complete algorithm. \Cref{a:main} contains the main
procedure, \cref{a:main2,a:main3} show auxiliary functions called from
the main procedure.

\begin{algorithm*}[t]
\DontPrintSemicolon

\begin{multicols}{2}

\Record{Node}{
   $C$: a configuration (set of events) \;
   $D$: set of events. \; $e$: an event. \; $l$: the \emph{left child} Node \;
   $r$: the \emph{right child} Node \;
   $s$: Boolean, \True iff \emph{sweep node} \;
}

\BlankLine
\BlankLine
Global variables:\\
-- $U$: set of known events, grows monotonically\\
-- $N$: set of \emph{Nodes} in the tree, grows and shrinks

\BlankLine
\Proc{\main{}}{
   $N, U, W \eqdef \emptyset, \emptyset, \emptyset$ \;
   \label{l:wdef}
   $n \eqdef \allocatenode(\emptyset, \emptyset)$ \;
   \label{l:n0def}
   $n.s \eqdef \True$ \;
   Add $n$ to $W$ \;
   \ForEach {node $n \in W$}
   {
      \label{l:select2}
      $\tup{n', B} \eqdef \expandleft{$n$}$ \;
      \label{l:call-expandleft}
      Add $\cexmain{$n'.C$}$ to $U$ \;
      \label{l:call-cex}
      \createrightbranches{$B$} \;
      \label{l:call-createrightbranches}
      \backtrack{$n'$} \;
      \label{l:call-backtrack}
   }
}

\BlankLine
\Fn{\encutoff{C}}{
\KwRet $\set{e \in \en C \colon \lnot \iscutoff{e}}$ \;
}

\BlankLine
\Fn{\allocatenode{C, D}}{
   \label{l:allocatenode}
   $n \eqdef \text{new \emph{Node}}$ \;
	$n.C \eqdef C$ \;
	$n.D \eqdef D$ \;
   $n.e, n.l, n.r \eqdef \Null, \Null, \Null$ \;
   $n.s \eqdef \False$ \;
	Add $n$ to $N$ \;
   \KwRet $n$ \;
}

\BlankLine
\Fn{\expandleft{n}}{
	\KwAssert $n.l = \Null$ \;
$B = \set n$ \;

   \ForEach {$e \in \encutoff{$n.C$} \setminus n.D$}
   {
      $n \eqdef \makeleft{$n, e$}$ \;
		Add $n$ to $B$ \;
		Add $e$ to $U$ \;
      \label{l:add-u}
   }
\KwRet $\tup{n, B}$ \;
}

\BlankLine
\Fn{\makeleft{$n, e$}}{
	\KwAssert $n.e = \Null \land n.l = \Null$ \;
   $n' \eqdef \allocatenode{$n.C \cup \set e, n.D$}$ \;
	$n.e \eqdef e$ \;
$n'.s \eqdef n.s$ \;
   \label{l:sweep-bit-inherit}
   $n.s \eqdef \False$ \;
   $n.l \eqdef n’$ \;
	\KwRet $n'$ \;
}

\end{multicols}
\vspace*{6pt}
\caption{Main algorithm}\label{a:main}
\end{algorithm*}

This algorithm is conceptually similar to Algorithm 1 in~\cite{RSSK15long}.
In~\cite{RSSK15long} the algorithm performs a DFS traversal over the POR tree.
Here we let
the user select (\cref{l:select2} of \main) the order in which the tree is explored.
This brings some complications, which we will explain below, but allows for a more
effective combination of the POR search with the exploration heuristics commonly used in symbolic execution engines such as KLEE.\@

The algorithm visits all cutoff-free, $\subseteq$-maximal configurations of the
unfolding, and organizes the exploration into a tree.  Each node of the tree is
an instance of the data structure \emph{Node}, shown in \cref{a:main}.
Conceptually, a \emph{Node} is a tuple $n \eqdef \tup{C, D, e, l, r, s}$
that represents a state of the exploration algorithm.
Specifically, given a Node \(n\), we need to explore all maximal configurations
that include~$C$, exclude~$D$, and where all the maximal configurations visited
through the left child node of~$n$ include~$e$.
The event~$e$ always satisfies that $e \in \en C$.
The left child node, field~$l$, represents a state of the search where
we explore event~$e$ (adding~$e$ to the $C$-component).
The right child node, field~$r$, represents a state of the exploration
where we have decided to exclude~$e$ from any subsequently visited
maximal-configuration. We keep track of this by adding~$e$ to the $D$-component
of that node.
The node is a \emph{sweep node} iff bit $s$ is true.
We call~$s$ the \emph{sweep bit}.
We will explain in \cref{sec:main.sweep.bit} what the~$s$ bit is used for.

\subsection{Exploring the left branch}\label{sec:main.left}

The algorithm maintains a queue of nodes to be explored in variable~\(W\)
(\cref{l:wdef}).
Initially, the queue contains a node representing the exploration of the empty
configuration (\cref{l:n0def}).
Note that the function \allocatenode{} adds the new node to~$N$.
The user is free to select what node from the queue gets explored next
(\cref{l:select2}).
After a node~$n$ is selected, the first task is expanding the leftmost branch of the
tree rooted at~$n$. We do this in the call to \expandleft{},
\cref{l:call-expandleft}.

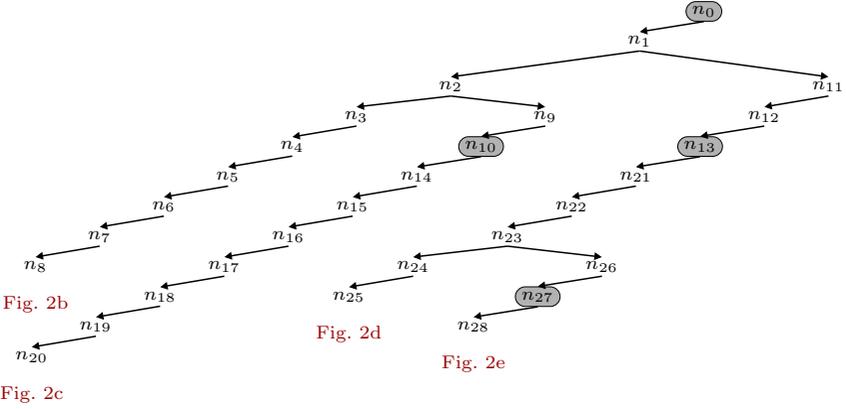
\Cref{fig:big.tree} shows an example of this.
This tree will be explored when the algorithm is executed on the program shown
in~\cref{fig:3:a} (whose unfolding is shown in~\cref{fig:algo.unf.a}).
Initially, the queue contains $n_0 \eqdef \tup{\emptyset, \emptyset, \Null, \Null, \Null}$.
At \cref{l:select2} we extract $n_0$ from the queue.
The call to \expandleft{$n_0$} will create all nodes from~$n_1$ to~$n_8$
(the call will also indirectly update the left-child pointer in node~$n_0$ so that it
becomes $n_1$).
The call also returns the last (leaf) node ($n_8$ in our example)
and the set $B$ of nodes in the branch ($\set{n_0, \ldots, n_8}$ in our example).
As new events are discovered, they are also added to the set~$U$
(\cref{l:add-u}).

\begin{figure*}[t]
\centering
\begin{tikzpicture}[
	font=\fontsize{7}{6.25}\selectfont{},
	lbl/.style={
		inner xsep=0cm,
		inner ysep=0.05cm,
	},
	highlight/.style={
		draw=black,
		line width=0.01cm,
		fill=black!30,
		rounded rectangle,
		inner sep=0.06cm,
	},
	succ/.style={
		draw,
		->,
		>={Triangle[length=0.1cm,width=0.1cm]},
		line width=0.02cm,
	},
	config/.style={
inner sep=0.05cm,
	},
]
	\def\xincr{0.85cm}
	\def\yincr{0.4cm}
	\def\incr{\xincr,\yincr}
	\def\refdist{0.5cm}
	\def\xone{1.25cm}
	\def\xtwo{2.5cm}

	\begin{scope}[shift={(0.6cm,0cm)}]
		\node[lbl,highlight] (n0) at (0cm,0cm) {\(n_{0}\)};

		\node[lbl] (n1) at ($(n0)-(\incr)$) {\(n_{1}\)};
		\path[succ] (n0.south) to (n1.north);

		\node[lbl] (n2) at ($(n1)-(\xtwo,1.5*\yincr)$) {\(n_{2}\)};
		\path[succ] (n1.south) to (n2.north);

		\node[lbl] (n3) at ($(n2)-(\xone,\yincr)$) {\(n_{3}\)};
		\path[succ] (n2.south) to (n3.north);

		\node[lbl] (n4) at ($(n3)-(\incr)$) {\(n_{4}\)};
		\path[succ] (n3.south) to (n4.north);

		\node[lbl] (n5) at ($(n4)-(\incr)$) {\(n_{5}\)};
		\path[succ] (n4.south) to (n5.north);

		\node[lbl] (n6) at ($(n5)-(\incr)$) {\(n_{6}\)};
		\path[succ] (n5.south) to (n6.north);

		\node[lbl] (n7) at ($(n6)-(\incr)$) {\(n_{7}\)};
		\path[succ] (n6.south) to (n7.north);

		\node[lbl] (n8) at ($(n7)-(\incr)$) {\(n_{8}\)};
		\path[succ] (n7.south) to (n8.north);
		
		\node[config] (cb) at ($(n8)-(0,\refdist)$) {\cref{fig:3:b}};

		\node[lbl] (n9) at ($(n2)-(-\xone,\yincr)$) {\(n_{9}\)};
		\path[succ] (n2.south) to (n9.north);

		\node[lbl,highlight] (n10) at ($(n9)-(\incr)$) {\(n_{10}\)};
		\path[succ] (n9.south) to (n10.north);

		\node[lbl] (n14) at ($(n10)-(\incr)$) {\(n_{14}\)};
		\path[succ] (n10.south) to (n14.north);

		\node[lbl] (n15) at ($(n14)-(\incr)$) {\(n_{15}\)};
		\path[succ] (n14.south) to (n15.north);

		\node[lbl] (n16) at ($(n15)-(\incr)$) {\(n_{16}\)};
		\path[succ] (n15.south) to (n16.north);

		\node[lbl] (n17) at ($(n16)-(\incr)$) {\(n_{17}\)};
		\path[succ] (n16.south) to (n17.north);

		\node[lbl] (n18) at ($(n17)-(\incr)$) {\(n_{18}\)};
		\path[succ] (n17.south) to (n18.north);

		\node[lbl] (n19) at ($(n18)-(\incr)$) {\(n_{19}\)};
		\path[succ] (n18.south) to (n19.north);

		\node[lbl] (n20) at ($(n19)-(\incr)$) {\(n_{20}\)};
		\path[succ] (n19.south) to (n20.north);
		
		\node[config] (cb) at ($(n20)-(0,\refdist)$) {\cref{fig:3:c}};

		\node[lbl] (n11) at ($(n1)-(-\xtwo,1.5*\yincr)$) {\(n_{11}\)};
		\path[succ] (n1.south) to (n11.north);

		\node[lbl] (n12) at ($(n11)-(\incr)$) {\(n_{12}\)};
		\path[succ] (n11.south) to (n12.north);

		\node[lbl,highlight] (n13) at ($(n12)-(\incr)$) {\(n_{13}\)};
		\path[succ] (n12.south) to (n13.north);

		\node[lbl] (n21) at ($(n13)-(\incr)$) {\(n_{21}\)};
		\path[succ] (n13.south) to (n21.north);

		\node[lbl] (n22) at ($(n21)-(\incr)$) {\(n_{22}\)};
		\path[succ] (n21.south) to (n22.north);

		\node[lbl] (n23) at ($(n22)-(\incr)$) {\(n_{23}\)};
		\path[succ] (n22.south) to (n23.north);

		\node[lbl] (n24) at ($(n23)-(\xone,\yincr)$) {\(n_{24}\)};
		\path[succ] (n23.south) to (n24.north);

		\node[lbl] (n25) at ($(n24)-(\incr)$) {\(n_{25}\)};
		\path[succ] (n24.south) to (n25.north);
		
		\node[config] (cb) at ($(n25)-(0,\refdist)$) {\cref{fig:3:d}};

		\node[lbl] (n26) at ($(n23)-(-\xone,\yincr)$) {\(n_{26}\)};
		\path[succ] (n23.south) to (n26.north);

		\node[lbl,highlight] (n27) at ($(n26)-(\incr)$) {\(n_{27}\)};
		\path[succ] (n26.south) to (n27.north);

		\node[lbl] (n28) at ($(n27)-(\incr)$) {\(n_{28}\)};
		\path[succ] (n27.south) to (n28.north);
		
		\node[config] (cb) at ($(n28)-(0,\refdist)$) {\cref{fig:3:e}};

\end{scope}
\end{tikzpicture}\scriptsize \begin{align*}
   n_0 &= \tup{\emptyset, \emptyset, 1, n_1, \circ} &
      n_9 &= \tup{\set{1,2}, \set{3}, 9, n_{10}, \circ} &
      n_{11} &= \tup{\set{1}, \set{2}, 5, n_{12}, \circ}
      \\
   n_1 &= \tup{\set{1}, \emptyset, 2, n_2, n_{11}} &
      n_{10} &= \tup{\set{1,2,9}, \set{3}, 5, n_{14}, \circ} &
      n_{12} &= \tup{\set{1,5}, \set{2}, 16, n_{13}, \circ}
      \\
   n_2 &= \tup{\set{1,2}, \emptyset, 3, n_3, n_9} &
      n_{14} &= \tup{\set{1,2,9,5}, \set{3}, 10, n_{15}, \circ} &
      n_{13} &= \tup{\set{1,5,16}, \set{2}, 17, n_{21}, \circ} &
      \\
   n_3 &= \tup{\set{1,2,3}, \emptyset, 5, n_4, \circ} &
      \ldots &&
      \ldots
      \\
   \ldots &&
      n_{20} &= \tup{[15], \set{3}, \circ, \circ, \circ} &
      n_{25} &= \tup{[21], \set{2}, \circ, \circ, \circ}
      \\
   n_8 &= \tup{[8], \emptyset, \circ, \circ, \circ} &
      &&
      n_{28} &= \tup{[23], \set{2,20}, \circ, \circ, \circ}
      \\
\end{align*}
\vspace{-1.0cm}
 \caption{Complete POR exploration tree for~\cref{fig:3:a}. Symbol~$\circ$ denotes a
   \protect\Null pointer.
   We do not represent the sweep bit here.
   Nodes~$n_0, n_{10}, n_{13}, n_{27}$ will be added to the queue~$W$ of
   \cref{a:main}.
}\label{fig:big.tree}
\end{figure*}

\begin{algorithm*}[t]
\DontPrintSemicolon
\setlength{\columnsep}{2pt}
\begin{multicols}{2}

\BlankLine
\Proc{\createrightbranches{B}}{
   \ForEach{node $n \in B$}
   {
      \lIf{$n.e = \Null$}{ \KwContinue }
      \lIf{$n.r \ne \Null$}{ \KwContinue }
      $J \eqdef \alternatives{$n.C, n.D \cup \set{n.e}$}$ \;
      \If{$J \ne \emptyset$}
      {
         \makerightbranch{$n, J \setminus n.C$} \;
      }
   }
}

\BlankLine
\Fn{\makerightbranch{n, A}}{
	\KwAssert $n.e \ne \Null$ \;
   $n = \makeright{$n$}$ \;
   \label{l:call.makeright}

\ForEach{$e \in min_{<}(A)$}
   {
		Remove $e$ from $A$ \;
      $n \eqdef \makeleft{$n, e$}$ \;
   \label{l:call.makeleft}
   }
   Add $n$ to $W$ \;
   \label{l:mrb.add}
	\KwRet $n$ \;
}

\BlankLine
\Fn{\alternatives{$C,D$}}{
   \KwAssert $D \subseteq \ex C$ \;
   \KwAssert $D \cap \en C \ne \emptyset$ \;
   Let $e$ be some event in $\en C \cap D$ \;
   $S \eqdef \set{e' \in U \colon
      e' \cfl e \land
      [e'] \cap D = \emptyset}$ \;
   $S \eqdef \set{e' \in S \colon
      [e'] \cup C \text{ is a config.}}$ \;
   \lIf{$S = \emptyset$} { \KwRet $\emptyset$ }
   Select some event $e'$ from $S$ \;
   \KwRet $[e']$ \;
}

\Fn{\makeright{$n$}}{
	\KwAssert $n.e \ne \Null \land n.r = \Null$ \;
   $n' \eqdef \allocatenode{$n.C, n.D \cup \set{n.e}$}$ \;
	$n.r \eqdef n'$ \;
	\KwRet $n'$ \;
}

\end{multicols}
\vspace*{6pt}
\caption{Main algorithm: utility functions to explore right branches.}\label{a:main2}
\end{algorithm*}

\begin{algorithm*}[t]
\DontPrintSemicolon
\setlength{\columnsep}{2pt}
\begin{multicols}{2}

\BlankLine
\Proc{\backtrack{n}}{
   \KwAssert $\lnot \haschildren{$n$}$ \;
   \lIf{$\lnot n.s$}{\KwRet}
   \label{l:backtrack-returns}
   \While{$n \ne \Null \land \lnot \haschildren{$n$}$}
   {
      $p \eqdef \parent{$n$}$ \;
      \label{l:parent}
      \createrightbranches{$\set p $} \;
      \label{l:call-createrightbranches-backtrack}
      \remove{$n$} \;
		$n \eqdef p$ \;
   }
   \lIf{$n = \Null$}{ \KwRet }
	\KwAssert $n.r \ne \Null$ \;
   \updatesweepbit{$n$} \;
}

\BlankLine
\Fn{\haschildren{n}}{
	\KwRet $n.l \ne \Null$ or $n.r \ne \Null$ \;
}

\BlankLine
\Fn{\parent{n}}{
   \If{$n$ is the root of the tree}{\KwRet \Null}
	\KwRet the parent node of $n$ \;
}

\BlankLine
\Proc{\remove{n}}{
   \KwAssert $n \in N \land \lnot \haschildren{$n$}$ \;
   $p \eqdef \parent{$n$}$ \;
   \If{$p \ne \Null \land p.l = n$}{$p.l \eqdef \Null$}
   \If{$p \ne \Null \land p.r = n$}{$p.r \eqdef \Null$}
	Remove $n$ from $N$ \;
}

\BlankLine
\Proc{\updatesweepbit{n}}{
\KwAssert $n.l = \Null \land n.r \ne \Null$ \;
$n' \eqdef n.r$ \;
\lWhile{$n'.l \ne \Null$}{$n' = n'.l$}
   $n'.s \eqdef \True$ \;
\KwAssert{$n' \in W \land n.s = \False$} \;
}

\end{multicols}
\vspace*{6pt}
\caption{Main algorithm: backtracking utility functions}\label{a:main3}
\end{algorithm*}

Back to the \main() function, we now add (\cref{l:call-cex}) all conflicting
extensions of the maximal configuration visited by the leaf node.
In our example, the maximal configuration
will be \([8] = \set{1, 2, 3, 4, 5, 6, 7, 8}\), and we will add events~9 and~16
(see the full unfolding in \cref{fig:algo.unf.a}).
We use the algorithms described in \cref{app:cex} to perform this step.

\subsection{Exploring the right branch}\label{sec:main.right}

Next we make a call to \createrightbranches{$B$} in \cref{l:call-createrightbranches}.
This function is defined in~\cref{a:main2}.
The purpose of this call is inserting in the tree (set $N$) \emph{some} of the
right children nodes for nodes contained in $B$.
Some right-hand side children nodes that the tree will eventually have might not
be inserted here, because the algorithm does not have yet visibility over (the
conflicting extensions that trigger) them.
Those will provably be inserted in the call to \backtrack{}, \cref{l:call-backtrack}
of~\cref{a:main}, at some point in the future (at a later iteration of
of the loop in \main).

A right-hand side child node from a node $n = \tup{C, D, e, \cdot, \cdot, \cdot}$
represents the exploration of all maximal configurations
that include~$C$ and exclude~$D \cup \set e$.
To create a right branch, the algorithm first needs to decide if one such
maximal configuration (including~$C$ and excluding~$D \cup \set e$) exists,
which is an NP-complete problem~\cite{NRSCP18}.
In this work we use $k$-partial alternatives~\cite{NRSCP18} with $k=1$
to compute (or rather, approximate) such decision.
$k$-partial alternatives can be computed in P-time and are a good approximation
for the algorithm.

In simple terms, a $k$-partial alternative for node~$n$ is a hint for the
algorithm that a maximal configuration which includes~$C$ and excludes $D \cup
\set e$ exists.
When such a maximal configuration exists, a 1-partial alternative will also
exist. On the other hand, when such a maximal configuration does not exist, a 1-partial alternative
may still exist, therefore misguiding the algorithm. This is the cost we pay to
compute an approximate solution for the NP-complete problem using a P-time
algorithm.

We compute 1-partial alternatives in function \alternatives, shown in \cref{a:main2},
and called from \createrightbranches.
In \cref{fig:big.tree}, \createrightbranches would find that both~$n_2$ and~$n_1$ have
right branches.
The created right branch always consist on one right-node
(\cref{l:call.makeright} in \makerightbranch) followed by one or more left-nodes
(\cref{l:call.makeleft} in \makerightbranch).
In our example, this would create the nodes~$n_9$ and~$n_{10}$ for~$n_2$
and $n_{11}, n_{12}, n_{13}$ for~$n_1$.
The terminal node of any created right branch (both $n_{10}$ and $n_{13}$)
is always added to the queue (\cref{l:mrb.add} in \cref{a:main2}), so that they can later be
selected (\cref{l:select2} in \main) and expanded.

\subsection{Sweep bit}\label{sec:main.sweep.bit}

The original algorithm~\cite{RSSK15long} on which~\cref{a:main} is based
performs an in-order DFS traversal of the tree.\footnote{Recall that an
\emph{in-order} DFS traversal of a binary tree first visits the left subtree of
a node, then the node, then the right-subtree.}
This means that, given a node~$n$, it evaluates whether a right child node
exists for~$n$ after having completely visited the left subtree of~$n$.
However, in our algorithm, the user decides in which order the nodes are
expanded. This means that the left subtree might not have been entirely visited
when our algorithm needs to decide if a right child node exists in the call to
\createrightbranches in \cref{l:call-createrightbranches} from \main.
Consequently, even if a right node should be visited, the algorithm might not
have yet seen enough events to conclude so.
Not exploring such right node would mean that we miss maximal configurations,
leading to an incomplete exploration of the unfolding.

In order for our algorithm to explore all configurations we need to keep track of
what parts of the tree our algorithm has visited.
We do so in the simplest manner possible: we keep track of which nodes a
hypothetical, in-order DFS traversal of the tree would have
explored so far using the \emph{sweep bit}.

Our algorithm updates the sweep bit of the nodes in a manner that simulates an
in-order DFS traversal of the tree. At any point in time there is exactly one
node in the tree whose sweep bit is set to true. We call such node the \emph{sweep node}.
When the user happens to select the sweep node for exploration,
then the algorithm simulates (part of) the in-order DFS traversal, updating
the sweep bit accordingly.
When the user selects a node whose sweep bit is set to false,
then the sweep node remains unchanged, as the algorithm is not simulating
DFS traversal right now.

The management of the sweep bit is done in the \backtrack and \makeleft
functions.
Function \makeleft{$n,e$} creates a new left child node for node~$n$.
The sweep bit of the new node equals that of the parent node ($n$),
see \cref{l:sweep-bit-inherit} of \makeleft. One line below we
also turn off the sweep bit of the parent.
If the parent was the sweep node, now the left child is.
If the parent was not the sweep node, its left child will not become the sweep node.
When the sweep node is selected in \cref{l:select2} of \main, the calls
to \makeleft{} issued by \expandleft{} will simulate that the DFS explores the
left-most branch rooted at the sweep node, pushing the sweep bit
down the branch to the leaf node.

Function \backtrack{} updates the sweep bit to simulate the backtracking
operations of the DFS.\@
When \backtrack{$n$} is called (\cref{l:call-backtrack} of \main),
$n$ is always a leaf node.
If~$n$ is not the sweep node, then there is nothing to simulate and \backtrack
returns, \cref{l:backtrack-returns}.
When~$n$ is the sweep node, we backtrack moving upwards (\cref{l:parent})
along the branch and checking again every node for right branches
(\cref{l:call-createrightbranches-backtrack}).\footnote{Note that we have
already checked at some point in the past if such nodes had right branches, in
\cref{l:call-createrightbranches} of \main.
Back then the check was inconclusive, as we explained in~\cref{sec:main.right}.
The check we do now is decisive: if no right branches are found now, then no
maximal configuration exists that can be visited by exploring a right branch
from this node.}

If the node has a right child, then the leaf node of the leftmost branch rooted at the right child becomes the new sweep node.
If the node does not have a right child, then we remove it (because we know that
the algorithm has already explored the entire left and right subtrees) and continue backtracking.

\subsection{Correctness}

We re-state now the theorem in \cref{sec:algo.cutoffs}, splitting the statement
in two:

\begin{theorem}[Completeness, see \cref{t:algo.correctness}]
For any reachable state $s \in \reach{M_P}$, function \main in \cref{a:main}
explores a configuration~$C$ such that for some~$C' \subseteq C$ it holds that
$\state{C'} = s$.
\end{theorem}

\begin{proof}[Sketch]

\cref{a:main}
explores the same execution tree that we have characterized in Appendix B of
\cite{NRSCP18long}.
The procedure \alternatives of \cref{a:main}, to compute alternatives,
implements $k$-partial alternatives~\cite{NRSCP18long} with $k=1$.
Theorem 2 in~\cite{NRSCP18long} guarantees that all maximal configurations of
the unfolding will be explored.

\end{proof}

\begin{theorem}[Termination, see \cref{t:algo.correctness}]
\Cref{a:main} terminates
for any program~$P$ such that~$\reach P$ is finite.
\end{theorem}

\begin{proof}[Sketch]

\cref{a:main}
explores the same execution tree that we have characterized in Appendix B of
\cite{NRSCP18long}.
Theorem 1 in~\cite{NRSCP18long} guarantees that the exploration always finishes.

\end{proof}

\section{Delta-Based Memory Fingerprinting}\label{app:fingerprint}
Our cutoff events rely on a quick and memory-efficient way to identify whether any two events are associated with the same program state.
To this end, we define the \emph{fingerprint} of a program state, which is, in essence, a strong hash of the program state.
We then store and compare only these fingerprints.
To work as intended, the fingerprinting scheme has to satisfy three properties:

\begin{enumerate}
	\item Each fingerprint uniquely (with very high probability) identifies the program state that is equivalent to exactly the local configuration of a given event.
	\item Fingerprints are efficiently computable and comparable.
	\item The fingerprint of an event is computable without explicitly creating a program state that contains exactly the local configuration of the given event, i.e., no events concurrent to it.
\end{enumerate}

Instead of defining the fingerprint directly over the program state, we define it over \emph{fragments} of the program state.
By deconstructing the program state into a fine-grained set of fragments, we can limit the amount of work needed to update the fingerprint to only the fragments that changed.
The only necessary property for the set of fragments associated with a program state is, that the fragments of two program states are equal iff the states are equal as well.

We use a tagged tuple representation to ensure that no structural weaknesses exist that would cause multiple program states to map to the same set of fragments.
For example, a byte of memory on the program's heap, which has the concrete (i.e., not symbolic) value \(v\) and is located at location \(l\), is represented by the fragment \(\tup{\text{\texttt{"concrete heap byte"}}, l, v}\).

\begin{definition}\label{def:fingerprint}
Given a strong cryptographic hash function \(h\) and an arbitrary bit vector \(I\) of the same size as the output of \(h\), we define the \emph{fingerprint} of a program state \(s\) to be \(\text{fingerprint}\left(s\right) \eqdef I \oplus \bigoplus_{f \in \text{fragments}\left(s\right)} h\left(f\right) \), with the \({\oplus}\) symbol denoting the bitwise exclusive or operation.
\end{definition}

By basing the fingerprint on a strong cryptographic hash function (we use BLAKE2b) and ensuring that the fragment representation of the program state itself is collision-free, we can ensure that collisions in the fingerprint will occur with negligible probability and thus that the first of the stated properties holds.

\subsection{Efficiently Computing Fingerprints}
The core insight of our delta-based memory fingerprinting scheme is that we can efficiently update the fingerprint when the program state changes upon executing an instruction.
When an instruction changes the program state such that the new program state differs from the old program state in one or more fragments, we can compute a delta that encompasses this change and that can be used to easily transform the old fingerprint into the new fingerprint.

Given a change that causes a fragment \(f\) to be replaced with a new fragment \(f'\), the delta is the bitwise xor of their hashes.
To update the fingerprint from the old to the new value, this delta can be xor-ed onto the old fingerprint.
To see why this is the case, consider that the old fingerprint is the xor of many different values, one of which is the hash of the old fragment.
The new fingerprint differs from the old fingerprint only in that this element is now the hash of the new fingerprint.
To see why that is the case, consider the following transformation:

\begin{align*}
	\text{fingerprint}_\text{new} &= \textcolor{red!75!black}{\text{fingerprint}_\text{old}} \oplus \textcolor{blue!75!black}{\Delta} \\
	&= \textcolor{red!75!black}{\text{fingerprint}_\text{old}} \oplus \textcolor{blue!75!black}{h\left(f\right) \oplus h\left(f'\right)} \\
	&= \textcolor{red!75!black}{\left( I \oplus \ldots \oplus h\left(f\right) \oplus \ldots \right)} \oplus \textcolor{blue!75!black}{h\left(f\right) \oplus h\left(f'\right)} \\
	&= \textcolor{red!75!black}{I \oplus \ldots \oplus} \left( \textcolor{red!75!black}{h\left(f\right)} \oplus \textcolor{blue!75!black}{h\left(f\right) \oplus h\left(f'\right)} \right) \textcolor{red!75!black}{\oplus \ldots} \\
	&= \textcolor{red!75!black}{I \oplus \ldots \oplus} \left( 0 \oplus \textcolor{blue!75!black}{h\left(f'\right)} \right) \textcolor{red!75!black}{\oplus \ldots} \\
	&= \textcolor{red!75!black}{I \oplus \ldots \oplus \textcolor{blue!75!black}{h\left(f'\right)} \oplus \ldots} \\
	&= \text{fingerprint}_\text{new}
\end{align*}

The delta updates inherit many interesting properties from the xor operation they are based on, such as commutativity, associativity and self-inversion.
For example, multiple delta-updates can be combined into one via bitwise xor.

For our algorithm, we choose the initial value \(I\) to be equal to the xor of the hashes of all fragments of the initial program state, i.e., such that the fingerprint of the initial program state is \(0\).
Choosing \(I\) this way, simplifies further computations by letting us restate the fingerprint as the bitwise xor of all changes made to the initial program state.
This choice is mostly based on convenience, but also gives a tiny runtime improvement by letting us disregard the initial fingerprint when composing fingerprints from deltas.

\begin{figure*}[t]
	\centering
	\begin{center}\newcommand{\delt}[2]{\text{\tikz[baseline=(d.base)]{\node[draw={#1},fill={#1!25!white},rounded corners=0.2ex] (d) {\(\Delta_{#2}\)\hspace*{-0.1em}};}}}\definecolor{thread1}{HTML}{D00000}\definecolor{thread2}{HTML}{0000FF}\begin{tikzpicture}[
	event/.style={
		draw,
		font=\fontsize{7}{8}\selectfont{},
		line width=0.015cm,
		text width=0.3cm,
		align=center,
		inner xsep=0.03cm,
		inner ysep=0.05cm,
	},
	lbl/.style={
		font=\fontsize{7}{8}\selectfont{},
		inner xsep=0.02cm,
		inner ysep=0.03cm,
		outer xsep=0.1cm,
		outer ysep=0.025cm,
	},
	succ/.style={
		draw,
		->,
		>={Triangle[length=0.1cm,width=0.1cm]},
		line width=0.02cm,
	},
	conflict/.style={
		draw,
		-,
		color=red!80!black,
		>={Triangle[length=0.1cm,width=0.1cm]},
		densely dashed,
		line width=0.02cm,
	},
]

	\def\y{-1cm};
	\def\x{1cm};

	\path[draw=thread1,fill=thread1!5!white,line width=1pt] (-0.1cm,1cm) rectangle (-6.0cm,-3.5cm);
	\node[anchor=north east,font=\rmfamily\scshape{}\fontsize{9}{9}\selectfont{},inner xsep=0] at (-0.35cm,1cm) {\textcolor{thread1}{Thread 1}};
	\path[draw=thread2!80!black,fill=thread2!5!white,line width=1pt] (0.1cm,1cm) rectangle (6.0cm,-3.5cm);
	\node[anchor=north west,font=\rmfamily\scshape{}\fontsize{9}{9}\selectfont{},inner xsep=0] at (0.35cm,1cm) {\textcolor{thread2!80!black}{Thread 2}};

	\node[event] (e1) at (-0.5*\x,0) {1};
	\node[lbl,anchor=north east] at (e1.west) {\(\text{fp}=\delt{thread1}{1}\)};
	\node[lbl,anchor=south east] at (e1.west) {\(\text{thread updates}=\delt{thread1}{1}\)};

	\node[event] (e2) at (-0.5*\x,\y) {2};
	\path[succ] (e1.south) to (e2.north);
	\node[lbl,anchor=north east] at (e2.west) {\(\text{fp}=\delt{thread1}{1} \oplus \delt{thread1}{2}\)};
	\node[lbl,anchor=south east] at (e2.west) {\(\text{thread updates}=\delt{thread1}{1} \oplus \delt{thread1}{2}\)};

	\node[event] (e3) at (0.5*\x,\y) {3};
	\path[succ] (e1.south) to (e3.north);
	\node[lbl,anchor=north west] at (e3.east) {\(\text{fp}=\delt{thread1}{1} \oplus \delt{thread2}{3}\)};
	\node[lbl,anchor=south west] at (e3.east) {\(\text{thread updates}=\delt{thread2}{3}\)};

	\node[event] (e4) at (-0.5*\x,2*\y) {4};
	\path[succ] (e2.south) to (e4.north);
	\node[lbl,anchor=north east] at (e4.west) {\(\text{fp}=\delt{thread1}{1} \oplus \delt{thread1}{2} \oplus \delt{thread1}{4}\)};
	\node[lbl,anchor=south east] at (e4.west) {\(\text{thread updates}=\delt{thread1}{1} \oplus \delt{thread1}{2} \oplus \delt{thread1}{4}\)};

	\node[event] (e5) at (0.5*\x,2*\y) {5};
	\path[succ] (e3.south) to (e5.north);
	\node[lbl,anchor=north west] at (e5.east) {\(\text{fp}=\delt{thread1}{1} \oplus \delt{thread2}{3} \oplus \delt{thread2}{5}\)};
	\node[lbl,anchor=south west] at (e5.east) {\(\text{thread updates}=\delt{thread2}{3} \oplus \delt{thread2}{5}\)};

	\node[event] (e6) at (-0.5*\x,3*\y) {6};
	\path[succ] (e4.south) to (e6.north);
	\path[succ] (e5.south) to (e6.north);
	\node[lbl,anchor=north east] at (e6.west) {\(\text{fp}=\delt{thread1}{1} \oplus \delt{thread1}{2} \oplus \delt{thread2}{3} \oplus \delt{thread1}{4} \oplus \delt{thread2}{5} \oplus \delt{thread1}{6}\)};
	\node[lbl,anchor=south east] at (e6.west) {\(\text{thread updates}=\delt{thread1}{1} \oplus \delt{thread1}{2} \oplus \delt{thread1}{4} \oplus \delt{thread1}{6}\)};

	\node[event] (e7) at (0.5*\x,3*\y) {7};
	\path[succ] (e5.south) to (e7.north);
	\node[lbl,anchor=north west] at (e7.east) {\(\text{fp}=\delt{thread1}{1} \oplus \delt{thread2}{3} \oplus \delt{thread2}{5} \oplus \delt{thread2}{7}\)};
	\node[lbl,anchor=south west] at (e7.east) {\(\text{thread updates}=\delt{thread2}{3} \oplus \delt{thread2}{5} \oplus \delt{thread2}{7}\)};
\end{tikzpicture}\end{center}
\end{figure*}

Following this restatement, the fingerprint of an event is the application of all deltas in its history.
Recall that we can freely choose the order in which to apply the delta updates, because delta application is commutative and associative.
\Cref{fig:fingerprint} shows an example partial order with the events annotated with their fingerprints, as well as the per-thread updates.

The per-thread updates are a field associated with an event that contains all delta updates from the history of that event, that were caused by the same thread as the event is executed on.
They can be computed iteratively by keeping score of an accumulating delta for each thread of the program state while exploring.
We can compute the fingerprint for an event \(e\) by xor-ing the per-thread updates from the thread-maximal events in \(\left[e\right]\), i.e., the most recent event for each thread that is in \(\left[e\right]\).
For example, in \cref{fig:fingerprint}, the fingerprint of event 6 is the xor of the thread updates of event 6 and event 5 and the fingerprint of event 7 is the xor of the thread updates of events 1 and 7.
The fingerprint of event 2 is equal to its own thread updates, as \(\left[e_2\right]\) only contains events of one thread. \clearpage
\fi

\bibliography{main}

\begin{thebibliography}{10}
\providecommand{\url}[1]{\texttt{#1}}
\providecommand{\urlprefix}{URL }
\providecommand{\doi}[1]{https://doi.org/#1}

\bibitem{gnusort}
{{GNU}} sort, \url{https://www.gnu.org/software/coreutils/}

\bibitem{helgrind}
Helgrind: A thread error detector,
  \url{https://valgrind.org/docs/manual/hg-manual.html}

\bibitem{memcached}
Memcached, \url{https://www.memcached.org/}

\bibitem{POSIX}
{{IEEE Standard}} for {{Information Technology}}--{{Portable Operating System
  Interface}} ({{POSIX}}({{R}})) {{Base Specifications}}, {{Issue}} 7. Standard
  IEEE Std 1003.1-2017 (Revision of IEEE Std 1003.1-2008) (2018)

\bibitem{AAJS14}
Abdulla, P., Aronis, S., Jonsson, B., Sagonas, K.: Optimal {{Dynamic Partial
  Order Reduction}}. In: Proceedings of the 41st {{ACM SIGPLAN}}-{{SIGACT
  Symposium}} on {{Principles}} of {{Programming Languages}}. pp. 373--384.
  {{POPL}} '14, {Association for Computing Machinery}, {San Diego, California,
  USA} (Jan 2014). \doi{10.1145/2535838.2535845}

\bibitem{AAAJLS15}
Abdulla, P.A., Aronis, S., Atig, M.F., Jonsson, B., Leonardsson, C., Sagonas,
  K.: Stateless {{Model Checking}} for {{TSO}} and {{PSO}}. In: Baier, C.,
  Tinelli, C. (eds.) Tools and {{Algorithms}} for the {{Construction}} and
  {{Analysis}} of {{Systems}}. pp. 353--367. Lecture {{Notes}} in {{Computer
  Science}}, {Springer}, {Berlin, Heidelberg} (2015).
  \doi{10.1007/978-3-662-46681-0_28}

\bibitem{AABGS17}
Albert, E., {de la Banda}, M.G., {G{\'o}mez-Zamalloa}, M., Isabel, M., Stuckey,
  P.J.: Optimal {{Context}}-{{Sensitive Dynamic Partial Order Reduction}} with
  {{Observers}}. In: Proceedings of the 28th {{ACM SIGSOFT International
  Symposium}} on {{Software Testing}} and {{Analysis}}. pp. 352--362. {{ISSTA}}
  2019, {Association for Computing Machinery}, {Beijing, China} (Jul 2019).
  \doi{10.1145/3293882.3330565}

\bibitem{AKT13}
Alglave, J., Kroening, D., Tautschnig, M.: Partial {{Orders}} for {{Efficient
  Bounded Model Checking}} of~{{Concurrent~Software}}. In: Sharygina, N.,
  Veith, H. (eds.) Computer {{Aided Verification}}. pp. 141--157. Lecture
  {{Notes}} in {{Computer Science}}, {Springer}, {Berlin, Heidelberg} (2013).
  \doi{10.1007/978-3-642-39799-8_9}

\bibitem{svcomp2019}
Beyer, D.: Automatic {{Verification}} of {{C}} and {{Java Programs}}:
  {{SV}}-{{COMP}} 2019. In: Beyer, D., Huisman, M., Kordon, F., Steffen, B.
  (eds.) Tools and {{Algorithms}} for the {{Construction}} and {{Analysis}} of
  {{Systems}}. pp. 133--155. Lecture {{Notes}} in {{Computer Science}},
  {Springer International Publishing}, {Cham} (2019).
  \doi{10.1007/978-3-030-17502-3_9}

\bibitem{KLEE}
Cadar, C., Dunbar, D., Engler, D.R.: {{KLEE}}: {{Unassisted}} and {{Automatic
  Generation}} of {{High}}-{{Coverage Tests}} for {{Complex Systems Programs}}.
  In: Proceedings of the 8th {{USENIX Symposium}} on {{Operating Systems
  Design}} and {{Implementation}} ({{OSDI}}'08). vol.~8, pp. 209--224 (2008)

\bibitem{threedecades}
Cadar, C., Sen, K.: Symbolic {{Execution}} for {{Software Testing}}: {{Three
  Decades Later}}. Commun. ACM  \textbf{56}(2),  82--90 (Feb 2013).
  \doi{10.1145/2408776.2408795}

\bibitem{CCPSV18}
Chalupa, M., Chatterjee, K., Pavlogiannis, A., Sinha, N., Vaidya, K.:
  Data-{{Centric Dynamic Partial Order Reduction}}. Proceedings of the ACM on
  Programming Languages  \textbf{2}(POPL),  31:1--31:30 (Dec 2017).
  \doi{10.1145/3158119}

\bibitem{CJXML18}
Chen, D., Jiang, Y., Xu, C., Ma, X., Lu, J.: Testing {{Multithreaded Programs}}
  via {{Thread Speed Control}}. In: Proceedings of the 2018 26th {{ACM Joint
  Meeting}} on {{European Software Engineering Conference}} and {{Symposium}}
  on the {{Foundations}} of {{Software Engineering}}. pp. 15--25.
  {{ESEC}}/{{FSE}} 2018, {Association for Computing Machinery}, {Lake Buena
  Vista, FL, USA} (Oct 2018). \doi{10.1145/3236024.3236077}

\bibitem{CJ14}
Chu, D.H., Jaffar, J.: A {{Framework}} to {{Synergize Partial Order Reduction}}
  with {{State Interpolation}}. In: Yahav, E. (ed.) Hardware and {{Software}}:
  {{Verification}} and {{Testing}}. pp. 171--187. Lecture {{Notes}} in
  {{Computer Science}}, {Springer International Publishing}, {Cham} (2014).
  \doi{10.1007/978-3-319-13338-6_14}

\bibitem{CGP99}
Clarke, Jr, E.M., Grumberg, O., Peleg, D.: Model {{Checking}}. {The MIT Press}
  (Dec 1999)

\bibitem{CF11}
Cordeiro, L., Fischer, B.: Verifying {{Multi}}-threaded {{Software}} using
  {{SMT}}-based {{Context}}-{{Bounded Model Checking}}. In: 2011 33rd
  {{International Conference}} on {{Software Engineering}} ({{ICSE}}). pp.
  331--340 (May 2011). \doi{10.1145/1985793.1985839}

\bibitem{ERV02}
Esparza, J., R{\"o}mer, S., Vogler, W.: An {{Improvement}} of {{McMillan}}'s
  {{Unfolding Algorithm}}. Formal Methods in System Design  \textbf{20}(3),
  285--310 (May 2002). \doi{10.1023/A:1014746130920}

\bibitem{FHRV13}
Farzan, A., Holzer, A., Razavi, N., Veith, H.: Con2colic {{Testing}}. In:
  Proceedings of the 2013 9th {{Joint Meeting}} on {{Foundations}} of
  {{Software Engineering}}. pp. 37--47. {{ESEC}}/{{FSE}} 2013, {ACM}, {Saint
  Petersburg, Russia} (2013). \doi{10.1145/2491411.2491453}

\bibitem{FM07}
Farzan, A., Madhusudan, P.: Causal {{Dataflow Analysis}} for {{Concurrent
  Programs}}. In: Grumberg, O., Huth, M. (eds.) Tools and {{Algorithms}} for
  the {{Construction}} and {{Analysis}} of {{Systems}}. pp. 102--116. Lecture
  {{Notes}} in {{Computer Science}}, {Springer}, {Berlin, Heidelberg} (2007).
  \doi{10.1007/978-3-540-71209-1_10}

\bibitem{FG05}
Flanagan, C., Godefroid, P.: Dynamic {{Partial}}-order {{Reduction}} for
  {{Model Checking Software}}. In: Proceedings of the 32nd {{ACM
  SIGPLAN}}-{{SIGACT}} Symposium on {{Principles}} of Programming Languages.
  pp. 110--121. {{POPL}} '05, {Association for Computing Machinery}, {Long
  Beach, California, USA} (Jan 2005). \doi{10.1145/1040305.1040315}

\bibitem{God96}
Godefroid, P.: Partial-{{Order Methods}} for the {{Verification}} of
  {{Concurrent Systems}}--{{An Approach}} to the {{State}}-{{Explosion
  Problem}}. Lecture {{Notes}} in {{Computer Science}}, {Springer}, {Berlin,
  Heidelberg} (1996). \doi{10.1007/3-540-60761-7}

\bibitem{God97}
Godefroid, P.: Model {{Checking}} for {{Programming Languages}} using
  {{VeriSoft}}. In: Proceedings of the 24th {{ACM SIGPLAN}}-{{SIGACT}}
  Symposium on {{Principles}} of Programming Languages. pp. 174--186. {{POPL}}
  '97, {Association for Computing Machinery}, {Paris, France} (Jan 1997).
  \doi{10.1145/263699.263717}

\bibitem{sage}
Godefroid, P., Levin, M.Y., Molnar, D.: Automated {{Whitebox Fuzz Testing}}.
  In: 16th {{Annual Network}} \& {{Distributed System Security Symposium}}. pp.
  151--166 (Feb 2008)

\bibitem{GFYS07}
Gueta, G., Flanagan, C., Yahav, E., Sagiv, M.: Cartesian {{Partial}}-{{Order
  Reduction}}. In: Bo{\v s}na{\v c}ki, D., Edelkamp, S. (eds.) Model {{Checking
  Software}}. pp. 95--112. Lecture {{Notes}} in {{Computer Science}},
  {Springer}, {Berlin, Heidelberg} (2007). \doi{10.1007/978-3-540-73370-6_8}

\bibitem{GKWYG15}
Guo, S., Kusano, M., Wang, C., Yang, Z., Gupta, A.: Assertion {{Guided Symbolic
  Execution}} of {{Multithreaded Programs}}. In: Proceedings of the 2015 10th
  {{Joint Meeting}} on {{Foundations}} of {{Software Engineering}}. pp.
  854--865. {{ESEC}}/{{FSE}} 2015, {Association for Computing Machinery},
  {Bergamo, Italy} (Aug 2015). \doi{10.1145/2786805.2786841}

\bibitem{C18}
{International Organization for Standardization}: Information technology
  \textemdash{} {{Programming}} languages \textemdash{} {{C}}. Standard ISO/IEC
  9899:2018 (2018), \url{https://www.iso.org/standard/74528.html}

\bibitem{ITFLP14}
Inverso, O., Tomasco, E., Fischer, B., La~Torre, S., Parlato, G.: Bounded
  {{Model Checking}} of {{Multi}}-threaded {{C Programs}} via {{Lazy
  Sequentialization}}. In: Biere, A., Bloem, R. (eds.) Computer {{Aided
  Verification}}. pp. 585--602. Lecture {{Notes}} in {{Computer Science}},
  {Springer International Publishing}, {Cham} (2014).
  \doi{10.1007/978-3-319-08867-9_39}

\bibitem{KSH15}
K{\"a}hk{\"o}nen, K., Saarikivi, O., Heljanko, K.: Unfolding based automated
  testing of multithreaded programs. Automated Software Engineering
  \textbf{22}(4),  475--515 (Dec 2015). \doi{10.1007/s10515-014-0150-6}

\bibitem{KWG09}
Kahlon, V., Wang, C., Gupta, A.: Monotonic {{Partial Order Reduction}}: {{An
  Optimal Symbolic Partial Order Reduction Technique}}. In: Bouajjani, A.,
  Maler, O. (eds.) Computer {{Aided Verification}}. pp. 398--413. Lecture
  {{Notes}} in {{Computer Science}}, {Springer}, {Berlin, Heidelberg} (2009).
  \doi{10.1007/978-3-642-02658-4_31}

\bibitem{kingSymbolicExecutionProgram1976}
King, J.C.: Symbolic {{Execution}} and {{Program Testing}}. Commun. ACM
  \textbf{19}(7),  385--394 (Jul 1976). \doi{10.1145/360248.360252}

\bibitem{KW16}
Kusano, M., Wang, C.: Flow-{{Sensitive Composition}} of {{Thread}}-{{Modular
  Abstract Interpretation}}. In: Proceedings of the 2016 24th {{ACM SIGSOFT
  International Symposium}} on {{Foundations}} of {{Software Engineering}}. pp.
  799--809. {{FSE}} 2016, {Association for Computing Machinery}, {Seattle, WA,
  USA} (Nov 2016). \doi{10.1145/2950290.2950291}

\bibitem{Maz87}
Mazurkiewicz, A.: Trace {{Theory}}. In: Brauer, W., Reisig, W., Rozenberg, G.
  (eds.) Petri {{Nets}}: {{Applications}} and {{Relationships}} to {{Other
  Models}} of {{Concurrency}}. pp. 278--324. Lecture {{Notes}} in {{Computer
  Science}}, {Springer}, {Berlin, Heidelberg} (1987).
  \doi{10.1007/3-540-17906-2_30}

\bibitem{Mcm93}
McMillan, K.L.: Using unfoldings to avoid the state explosion problem in the
  verification of asynchronous circuits. In: {von Bochmann}, G., Probst, D.K.
  (eds.) Computer {{Aided Verification}}. pp. 164--177. Lecture {{Notes}} in
  {{Computer Science}}, {Springer}, {Berlin, Heidelberg} (1993).
  \doi{10.1007/3-540-56496-9_14}

\bibitem{Mine14}
Min{\'e}, A.: Relational {{Thread}}-{{Modular Static Value Analysis}} by
  {{Abstract Interpretation}}. In: McMillan, K.L., Rival, X. (eds.)
  Verification, {{Model Checking}}, and {{Abstract Interpretation}}. pp.
  39--58. Lecture {{Notes}} in {{Computer Science}}, {Springer}, {Berlin,
  Heidelberg} (2014). \doi{10.1007/978-3-642-54013-4_3}

\bibitem{MQ07}
Musuvathi, M., Qadeer, S.: Iterative {{Context Bounding}} for {{Systematic
  Testing}} of {{Multithreaded Programs}}. In: Proceedings of the 28th {{ACM
  SIGPLAN Conference}} on {{Programming Language Design}} and
  {{Implementation}}. pp. 446--455. {{PLDI}} '07, {Association for Computing
  Machinery}, {San Diego, California, USA} (Jun 2007).
  \doi{10.1145/1250734.1250785}

\bibitem{NRSCP18}
Nguyen, H.T.T., Rodr{\'i}guez, C., Sousa, M., Coti, C., Petrucci, L.:
  Quasi-{{Optimal Partial Order Reduction}}. In: Chockler, H., Weissenbacher,
  G. (eds.) Computer {{Aided Verification}}. pp. 354--371. Lecture {{Notes}} in
  {{Computer Science}}, {Springer International Publishing} (2018).
  \doi{10.1007/978-3-319-96142-2_22}

\bibitem{NRSCP18long}
Nguyen, H.T.T., Rodr{\'i}guez, C., Sousa, M., Coti, C., Petrucci, L.:
  Quasi-{{Optimal Partial Order Reduction}}. arXiv:1802.03950 [cs]  (Apr 2018),
  \url{https://arxiv.org/abs/1802.03950}

\bibitem{NSFTP17}
Nguyen, T.L., Schrammel, P., Fischer, B., La~Torre, S., Parlato, G.: Parallel
  {{Bug}}-{{Finding}} in {{Concurrent Programs}} via {{Reduced Interleaving
  Instances}}. In: Proceedings of the 32nd {{IEEE}}/{{ACM International
  Conference}} on {{Automated Software Engineering}}. pp. 753--764. {{ASE}}
  2017, {IEEE Press}, {Urbana-Champaign, IL, USA} (Oct 2017).
  \doi{10.1109/ASE.2017.8115686}

\bibitem{NPW81}
Nielsen, M., Plotkin, G., Winskel, G.: Petri {{Nets}}, {{Event Structures}} and
  {{Domains}}, {{Part I}}. Theoretical Computer Science  \textbf{13}(1),
  85--108 (1981)

\bibitem{symbolicpathfinder}
P{\u a}s{\u a}reanu, C.S., Rungta, N.: Symbolic {{PathFinder}}: {{Symbolic
  Execution}} of {{Java Bytecode}}. In: Proceedings of the {{IEEE}}/{{ACM}}
  International Conference on {{Automated}} Software Engineering. pp. 179--180.
  {{ASE}} '10, {Association for Computing Machinery}, {Antwerp, Belgium} (Sep
  2010). \doi{10.1145/1858996.1859035}

\bibitem{PSSTY17}
Prabhu, S., Schrammel, P., Srivas, M., Tautschnig, M., Yeolekar, A.: Concurrent
  {{Program Verification}} with~{{Invariant}}-{{Guided Underapproximation}}.
  In: D'Souza, D., Narayan~Kumar, K. (eds.) Automated {{Technology}} for
  {{Verification}} and {{Analysis}}. pp. 241--248. Lecture {{Notes}} in
  {{Computer Science}}, {Springer International Publishing}, {Cham} (2017).
  \doi{10.1007/978-3-319-68167-2_17}

\bibitem{QW04}
Qadeer, S., Wu, D.: {{KISS}}: {{Keep It Simple}} and {{Sequential}}. In:
  Proceedings of the {{ACM SIGPLAN}} 2004 Conference on {{Programming}}
  Language Design and Implementation. pp. 14--24. {{PLDI}} '04, {Association
  for Computing Machinery}, {Washington DC, USA} (Jun 2004).
  \doi{10.1145/996841.996845}

\bibitem{RSSK15}
Rodr{\'i}guez, C., Sousa, M., Sharma, S., Kroening, D.: Unfolding-based
  {{Partial Order Reduction}}. In: Aceto, L., Escrig, D.d.F. (eds.) 26th
  {{International Conference}} on {{Concurrency Theory}} ({{CONCUR}} 2015).
  Leibniz {{International Proceedings}} in {{Informatics}} ({{LIPIcs}}),
  vol.~42, pp. 456--469. {Schloss Dagstuhl\textendash{}Leibniz-Zentrum fuer
  Informatik}, {Dagstuhl, Germany} (2015). \doi{10.4230/LIPIcs.CONCUR.2015.456}

\bibitem{RSSK15long}
Rodr{\'i}guez, C., Sousa, M., Sharma, S., Kroening, D.: Unfolding-based
  {{Partial Order Reduction}}. arXiv:1507.00980 [cs]  (Jul 2015),
  \url{https://arxiv.org/abs/1507.00980}

\bibitem{porse-cav}
Schemmel, D., B{\"u}ning, J., Rodr{\'i}guez, C., Laprell, D., Wehrle, K.:
  Symbolic {{Partial}}-{{Order Execution}} for {{Testing Multi}}-{{Threaded
  Programs}}. In: Lahiri, S.K., Wang, C. (eds.) Computer {{Aided
  Verification}}. pp. 376--400. Lecture {{Notes}} in {{Computer Science}},
  {Springer International Publishing}, {Cham} (2020).
  \doi{10.1007/978-3-030-53288-8_18}

\bibitem{SymLive}
Schemmel, D., B{\"u}ning, J., Soria~Dustmann, O., Noll, T., Wehrle, K.:
  Symbolic {{Liveness Analysis}} of {{Real}}-{{World Software}}. In: Chockler,
  H., Weissenbacher, G. (eds.) Computer {{Aided Verification}}. pp. 447--466.
  Lecture {{Notes}} in {{Computer Science}}, {Springer International
  Publishing} (Jul 2018). \doi{10.1007/978-3-319-96142-2_27}

\bibitem{SA06}
Sen, K., Agha, G.: A {{Race}}-{{Detection}} and {{Flipping Algorithm}} for
  {{Automated Testing}} of {{Multi}}-threaded {{Programs}}. In: Bin, E., Ziv,
  A., Ur, S. (eds.) Hardware and {{Software}}, {{Verification}} and
  {{Testing}}. pp. 166--182. Lecture {{Notes}} in {{Computer Science}},
  {Springer}, {Berlin, Heidelberg} (2007). \doi{10.1007/978-3-540-70889-6_13}

\bibitem{tsan}
Serebryany, K., Iskhodzhanov, T.: {{ThreadSanitizer}}: Data race detection in
  practice. In: Proceedings of the {{Workshop}} on {{Binary Instrumentation}}
  and {{Applications}}. pp. 62--71. {{WBIA}} '09, {Association for Computing
  Machinery}, {New York, New York, USA} (Dec 2009).
  \doi{10.1145/1791194.1791203}

\bibitem{SRDK17}
Sousa, M., Rodr{\'i}guez, C., D'Silva, V., Kroening, D.: Abstract
  {{Interpretation}} with {{Unfoldings}}. In: Majumdar, R., Kun{\v c}ak, V.
  (eds.) Computer {{Aided Verification}}. pp. 197--216. Lecture {{Notes}} in
  {{Computer Science}}, {Springer International Publishing}, {Cham} (2017).
  \doi{10.1007/978-3-319-63390-9_11}

\bibitem{SRDK17long}
Sousa, M., Rodr{\'i}guez, C., D'Silva, V., Kroening, D.: Abstract
  {{Interpretation}} with {{Unfoldings}}. arXiv:1705.00595 [cs]  (May 2017),
  \url{https://arxiv.org/abs/1705.00595}

\bibitem{TDB16}
Thomson, P., Donaldson, A.F., Betts, A.: Concurrency {{Testing Using Controlled
  Schedulers}}: {{An Empirical Study}}. ACM Transactions on Parallel Computing
  \textbf{2}(4),  23:1--23:37 (Feb 2016). \doi{10.1145/2858651}

\bibitem{WKO13}
Wachter, B., Kroening, D., Ouaknine, J.: Verifying {{Multi}}-threaded
  {{Software}} with {{Impact}}. In: 2013 {{Formal Methods}} in
  {{Computer}}-{{Aided Design}}. pp. 210--217 (Oct 2013).
  \doi{10.1109/FMCAD.2013.6679412}

\bibitem{YDLLW18}
Yin, L., Dong, W., Liu, W., Li, Y., Wang, J.: {{YOGAR}}-{{CBMC}}: {{CBMC}} with
  {{Scheduling Constraint Based Abstraction Refinement}}. In: Beyer, D.,
  Huisman, M. (eds.) Tools and {{Algorithms}} for the {{Construction}} and
  {{Analysis}} of {{Systems}}. pp. 422--426. Lecture {{Notes}} in {{Computer
  Science}}, {Springer International Publishing}, {Cham} (2018).
  \doi{10.1007/978-3-319-89963-3_25}

\bibitem{YNPP12}
Yu, J., Narayanasamy, S., Pereira, C., Pokam, G.: Maple: {{A
  Coverage}}-{{Driven Testing Tool}} for {{Multithreaded Programs}}. In:
  Proceedings of the 27th {{Annual ACM SIGPLAN Conference}} on
  {{Object}}-{{Oriented Programming}}, {{Systems}}, {{Languages}}, and
  {{Applications}}, {{OOPSLA}} 2012, Part of {{SPLASH}} 2012. pp. 485--502.
  {{OOPSLA}} '12, {Association for Computing Machinery}, {Tucson, Arizona, USA}
  (Oct 2012). \doi{10.1145/2384616.2384651}

\end{thebibliography}

\end{document}